\documentclass{lmcs} 


\usepackage{hyperref}
\usepackage{amsmath}
\usepackage{amssymb}
\usepackage{xfrac} 
\usepackage{faktor}
\usepackage{stmaryrd}
\usepackage{relsize}

\usepackage[noend]{algorithmic}

\newcommand*\circlediamond{\vcenter{\hbox{\includegraphics[height=10pt]{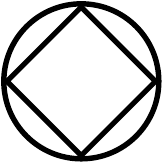}}}}
\newcommand*\circlebox{\vcenter{\hbox{\includegraphics[height=10pt]{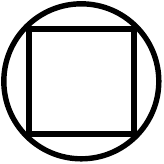}}}}

\newcommand{\op}[1]{\textbf{\texttt{#1}}}
\newcommand{\true}{{\small{\op{true}}}}

\newcommand{\sysFormulas}{\ensuremath{{\mathcal{L}_{\mu}^{\sys}}}}
\newcommand{\sysEFormulas}{\ensuremath{{\mathcal{L}_{e\mu}^{\sys}}}}
\newcommand{\envFormulas}{\ensuremath{{\mathcal{L}_{\mu}^{\env}}}}
\newcommand{\envEFormulas}{\ensuremath{{\mathcal{L}_{e\mu}^{\env}}}}
\newcommand{\EFormulas}{\ensuremath{\mathcal{L}_{e\mu}}}
\newcommand{\ECpre}{\ensuremath{\mathit{ECpre_{\sys}}}}
\newcommand{\ECpreEnv}{\ensuremath{\mathit{ECpre_{\env}}}}
\newcommand{\gfeStrategy}{\ensuremath{g_{\mathit{gfe}}}}

\newcommand{\GS}{\ensuremath{G = \langle \mathcal{V},
\mathcal{X}, \mathcal{Y}, \rho^e,\rho^s, \varphi
\rangle}} 
\newcommand{\WGS}{\ensuremath{G^{w} = \langle \mathcal{V},\mathcal{X}, \mathcal{Y}, \rho^e,\rho^s, \varphi, w^s\rangle}}
\newcommand{\GEc}{\ensuremath{{G^{w}(c)}}}
\newcommand{\twolinescurly}[2]{{\scriptscriptstyle\left\{\begin{array}{l}
#1\\
#2
\end{array}\right\}}}

\newcommand{\sys}{{\scriptstyle\mathit{sys}}}
\newcommand{\env}{{\scriptstyle\mathit{env}}}

\newcommand{\sysbig}{$\mathit{sys}$}
\newcommand{\envbig}{$\mathit{env}$}

\newcommand{\sysmu}{\ensuremath{\mathit{sys}\text{-}\mu}}
\newcommand{\envmu}{\ensuremath{\mathit{env}\text{-}\mu}}

\newcommand{\sysEmu}{\ensuremath{\mathit{sys}\text{-energy-}\mu}}
\newcommand{\envEmu}{\ensuremath{\mathit{env}\text{-energy-}\mu}}

\newcommand{\prio}{\ensuremath{\mathit{prio}}}
\newcommand{\prioA}{\ensuremath{\mathit{prio_{\mathcal{A}}}}}

\newcommand{\into}{\rightarrow}
\newcommand{\bigdiamond}{\mathlarger{\mathlarger{\mathlarger{\diamond}}}}

\newcommand{\sysEnergy}{\ensuremath{\mathit{sys}\text{-energy}}}
\newcommand{\envEnergy}{\ensuremath{\mathit{env}\text{-energy}}}
\newcommand{\eValuation}{\ensuremath{\mathcal{D}}}
\newcommand{\playerset}{\ensuremath{\{\mathit{env},\mathit{sys}\}}}

\newcommand{\sysCapEnergy}{\ensuremath{\mathit{sys}\text{-Energy}}}
\newcommand{\envCapEnergy}{\ensuremath{\mathit{env}\text{-Energy}}}

\theoremstyle{plain}\newtheorem{strategy}[thm]{Strategy}

\begin{document}

\title[Energy $\mu$-Calculus]{\texorpdfstring{Energy $\mu$-Calculus: Symbolic Fixed-Point Algorithms for $\omega$-Regular Energy Games}{Energy mu-Calculus: Symbolic Fixed-Point Algorithms for omega-Regular Energy Games}}


\author[G.~Amram]{Gal Amram}  
\address{Tel Aviv University, Tel Aviv, Israel} 
\email{galam1483@gmail.com}  

\author[S.~Maoz]{Shahar Maoz}  
\address{Tel Aviv University, Tel Aviv, Israel} 
\email{maoz@cs.tau.ac.il}  

\author[O.~Pistiner]{Or Pistiner}  
\address{Tel Aviv University, Tel Aviv, Israel} 
\email{orpistiner@gmail.com}

\author[J.O.~Ringert]{Jan Oliver Ringert}  
\address{University of Leicester, Leicester, United Kingdom}
\email{jor4@leicester.ac.uk}




\begin{abstract}
$\omega$-regular energy games,
which are weighted two-player turn-based games with the quantitative objective
to keep the energy levels non-negative, 
have been used in the context of verification and synthesis.
The logic of modal $\mu$-calculus, when applied over game graphs with
$\omega$-regular winning conditions, allows
defining symbolic algorithms in the form of fixed-point formulas for
computing the sets of winning states.

In this paper, we introduce energy $\mu$-calculus, a multi-valued extension of the
$\mu$-calculus that serves as a symbolic framework for solving $\omega$-regular energy games. 
Energy $\mu$-calculus enables the seamless reuse of existing, 
well-known symbolic $\mu$-calculus algorithms for $\omega$-regular games, 
to solve their corresponding energy augmented variants.
We define the syntax and semantics of energy $\mu$-calculus over symbolic
representations of the game graphs, and show how to use it to solve
the decision and the minimum credit problems for $\omega$-regular energy games,
for both bounded and unbounded energy level accumulations.
\end{abstract}

\maketitle

\section{Introduction}

\emph{Energy games} have been introduced by Chakrabarti et al.~\cite{ChakrabartiAHS03} to model components' energy interfaces, specifically
the requirement to avoid the exhaustion of an initially available resource, e.g., disk space or battery capacity.
Since their inception, they have been studied extensively in the context of
verification and synthesis, e.g.,~\cite{BohyBFR13,BouyerFLMS08,BrimC12,BrimCDGR11,ChakrabartiAHS03,ChatterjeeD12,ChatterjeeRR14,VelnerC0HRR15}.

Energy games are weighted two-player turn-based games with the
quantitative objective to keep the \emph{energy level}, the accumulated
sum of an initial credit and weights of transitions traversed thus far, non-negative in each prefix of a play.
Energy games induce a \emph{decision problem} that checks for
the existence of a finite initial credit sufficient for winning, and an
optimization problem for the \emph{minimum initial credit}.

The work~\cite{BouyerFLMS08} has introduced an \emph{upper bound} $c$ that specifies the maximal energy level
allowed to be accumulated throughout a play. In our work, we consider both the \emph{unbounded} energy objective of~\cite{BouyerFLMS08,ChakrabartiAHS03} where $c=+\infty$, 
and the \emph{bounded} energy objective of~\cite{BouyerFLMS08} where $c \in \mathbb{N}$ is finite and whenever the energy level exceeds $c$, it is truncated to $c$.

Energy games may be viewed as safety games with an additional quantitative
objective. Nevertheless, they have also been generalized to \emph{$\omega$-regular games with
energy objectives}~\cite{ChatterjeeD12,ChatterjeeRR14}, which are the focus of
our work.

We consider \emph{symbolic} algorithms for solving games, as opposed to \emph{explicit} ones. 
Symbolic algorithms operate on an implicit representation of the underlying game graph
and manipulate \emph{sets} of game states, whereas explicit algorithms operate on the explicit game graph representation
and manipulate individual states. Symbolic algorithms have been
shown to be scalable and practical for solving $\omega$-regular games, e.g.,~\cite{AlurMN05,ChatterjeeDHL17,ChatterjeeDHS18,JacobsBBEHKPRRS17,SanchezWW18,StasioMV18}.

\emph{Modal $\mu$-calculus}~\cite{Kozen} is an extension of propositional logic with
modal operators and least and greatest fixed-point operators.
Rather than the classical version of~\cite{Kozen}, we consider the \emph{game $\mu$-calculus}~\cite{EmersonJ91} and its application
over finite symbolic game structures~\cite{BJP+12} to solve games with $\omega$-regular winning conditions (see, e.g.,~\cite{AlfaroHM01,BJP+12,BruseFL14,KonighoferHB13,RaskinCDH07,Walukiewicz1996}).
For every $\omega$-regular condition, $\varphi$, there is a (game) $\mu$-calculus formula
that defines a symbolic fixed-point algorithm for computing the set of states that win $\varphi$~\cite{AlfaroHM01}.

Modal $\mu$-calculus has been extended to a multi-valued or quantitative
semantics where the value of a formula in a state is from some lattice, e.g.,~\cite{deAlfaro,AlfaroHM01,AlfaroM04,BrunsG04,Fischer2010,GrumbergLLS05,RaskinCDH07}.

\emph{We summarize the contributions of our work as follows.}

\begin{enumerate}
\item \textbf{Energy $\mu$-calculus as a symbolic framework for solving $\omega$-regular energy games.} We introduce \emph{energy $\mu$-calculus}, a multi-valued extension of 
the game $\mu$-calculus~\cite{EmersonJ91} over symbolic game structures~\cite{BJP+12}. Energy $\mu$-calculus serves as a framework for solving 
both the decision and the minimum credit problems with a bounded energy level accumulation.
While a game $\mu$-calculus formula characterizes a set of states, an energy $\mu$-calculus formula is
interpreted w.r.t. an upper bound $c \in \mathbb{N}$ and returns an \emph{energy function} that
assigns a value in $\{0,\dots,c\} \cup\{+\infty \}$ to each state of the underlying game.
Every $\omega$-regular condition is solved by evaluating a game $\mu$-calculus formula~\cite{AlfaroHM01},
and we show that this formula can be seamlessly reused as an energy $\mu$-calculus formula to solve the corresponding energy augmented game.

\item \textbf{Computation of a sufficient upper bound.} We bridge the gap between bounded and unbounded energy level accumulations
by showing that every $\omega$-regular winning condition admits a \emph{sufficiently large upper bound} on the energy level accumulation.
That is, we show that if the system player wins with an unbounded energy level accumulation,
then it also wins w.r.t. a finite upper bound with no need to increase the initial credit. Specifically, 
if the $\mu$-calculus formula $\psi$ solves the $\omega$-regular game, then the system wins w.r.t. the bound $(d+1)((N^2+N)m-1)K$, where $N$ is the size of the state space, $K$ is the maximal absolute weight, and $m$ and $d$
are the length and alternation depth of $\psi$, respectively. Through this sufficient bound, energy $\mu$-calculus 
also solves the decision and the minimum credit problems with an unbounded energy level accumulation.
\end{enumerate}

\noindent On the way to achieving the above results, we have obtained two additional contributions for \emph{energy parity games}~\cite{ChatterjeeD12}:

\begin{enumerate}\setcounter{enumi}{2}
  \item \textbf{Sufficient upper bound.} We show that if $\text{player}_0$ wins from a state $s$ in an energy parity game with an unbounded energy level accumulation, then she can also win from $s$ w.r.t. the energy upper bound $d(n-1)K$ without increasing her initial credit, where $d$ is the number of different priorities, $n$ is the number of states, and $K$ is the maximal absolute weight. 
  \item\textbf{Strategy complexity.} We show that if $\text{player}_0$ wins from a state $s$ in an energy parity game, then she has a strategy that wins from $s$ with a memory of size $d(n-1)K+1$ and without requiring to increase the initial credit.
  This slightly improves the best known memory upper bound of $dnK$~\cite{ChatterjeeD12}.
\end{enumerate}

\noindent To solve energy games with $\omega$-regular winning conditions, researchers have suggested to apply
a reduction to energy games with parity winning conditions; see, e.g.,~\cite{ChatterjeeD12,ChatterjeeRR14}.
In contrast, our approach uses a game $\mu$-calculus formula to describe the set of states that win the $\omega$-regular condition~\cite{AlfaroHM01}.
Then, it evaluates this formula w.r.t. the semantics of energy $\mu$-calculus and obtains
the energy function that maps each state to its minimal winning initial credit, and 
to $+\infty$ if there is no such an initial credit.
We identify two appealing key attributes of our approach.

First, our approach enables the use of existing results from the literature. Specifically,
Thm.~\ref{thm:sysEngMuCalcCorrectness} enables to seamlessly
transform well-known $\mu$-calculus formulas that solve games with $\omega$-regular
conditions $\varphi$, e.g., safety, reachability, B\"{u}chi, co-B\"{u}chi,
GR(1)~\cite{BJP+12}, counter-GR(1)~\cite{KonighoferHB13},
parity~\cite{BruseFL14,EmersonJ91}, etc. into solvers of corresponding $\varphi$-energy
games.

Second, the aforementioned transformation additionally results in algorithms that are symbolic, in the sense that they manipulate energy functions over symbolic weighted game structures.
Such symbolic algorithms can be implemented using,
e.g., Algebraic Decision Diagrams~\cite{BaharFGHMPS97,FujitaMY97}, as was done in~\cite{MaozPR16}.

To illustrate these key attributes, we consider the following well-known $\mu$-calculus formula
that solves B\"{u}chi games with target states $J$~\cite{2001automata,Thomas95}:
\begin{align}\label{eq:buchiExample:muCalcFormula}
\psi_{B_{J}} &= \nu Z (\mu Y (J \wedge \circlediamond Z) \vee \circlediamond Y)
\end{align}

In such a game, the system wins from a state if it can enforce infinitely many visits to $J$. Relying on Thm.~\ref{thm:sysEngMuCalcCorrectness},
we replace each occurrence of the modal operator $\circlediamond$ in Eq.~\ref{eq:buchiExample:muCalcFormula} with $\circlediamond_\op{E}$,
and obtain the following energy $\mu$-calculus formula that solves B\"{u}chi-energy $J$-states games:
\begin{align}\label{eq:buchiExample:energyMuCalcFormula}
\psi_{B_{J}}^\op{E} &= \nu Z (\mu Y (J \wedge \circlediamond_{\op{E}} Z) \vee \circlediamond_{\op{E}} Y)
\end{align}
 
That is, Eq.~\ref{eq:buchiExample:energyMuCalcFormula} defines the energy function that maps each state to the minimal initial credit
sufficient for the system to win the B\"{u}chi $J$-states condition while keeping the energy levels of all plays' prefixes
non-negative.

\begin{figure*}
\begin{tabular}{p{0.45\textwidth}|p{0.5\textwidth}}
\begin{minipage}{.45\textwidth}
\begin{algo}
\small Symbolic fixed-point algorithm to compute the set of states $Z$ that win in a B\"{u}chi $J$-states game.
An implementation of Eq.~\ref{eq:buchiExample:muCalcFormula} according to Def.~\ref{def:prop_mu_calculus_semantics}\\
\label{alg:buchi}
\begin{algorithmic}[1]
  \STATE $Z \gets \text{the state space}$\label{alg:buchi:initZ}
  \WHILE {not reached fixed-point of $Z$} \label{alg:buchi:fixZ}
  \STATE {$\mathit{recurrJ} \gets {J} \cap {\mathit{Cpre}_{\sys}(Z)}$} \label{alg:buchi:recurrJ}
    \STATE $Y \gets \emptyset$
    \WHILE {not reached fixed-point of $Y$}
    \label{alg:buchi:fixY}
      \STATE {$Y \gets \mathit{recurrJ} \cup {\mathit{Cpre}_{\sys}(Y)}$} \label{alg:buchi:leastY}
    \ENDWHILE \label{alg:buchi:fixYEnd}
    \STATE $Z\gets Y$ \label{alg:buchi:ZAssigned}
  \ENDWHILE \label{alg:buchi:fixZEnd}
  \RETURN $Z$ 
\end{algorithmic}
\end{algo}
\end{minipage}&
\begin{minipage}{.5\textwidth}
\begin{algo}
\small Symbolic fixed-point algorithm to compute
the mapping $Z$ between states and minimum initial credits
sufficient for winning in a B\"{u}chi-energy game. An implementation of Eq.~\ref{eq:buchiExample:energyMuCalcFormula} according to Def.~\ref{def:sysEngMuCalcSemantics}
\label{alg:buchiEnergy}
\begin{algorithmic}[1]
  \STATE $Z \gets \text{mapping of all states to $0$}$  
  \WHILE {not reached fixed-point of $Z$}
  \label{alg:buchiEnergy:fixZ}
    \STATE $\mathit{recurrJ} \gets {\max(f_{J}, \ECpre(Z))}$\label{alg:buchiEnergy:recurrJ}
    \STATE $Y \gets \text{mapping of all states to $+\infty$}$
    \WHILE {not reached fixed-point of $Y$}
      \STATE {${Y} \gets {\min(\mathit{recurrJ}, \ECpre(Y))}$}\label{alg:buchiEnergy:leastY}
    \ENDWHILE
    \STATE $Z \gets Y$
  \ENDWHILE
  \RETURN $Z$ 
\end{algorithmic}
\end{algo}
\end{minipage}
\end{tabular}
\end{figure*}

Alg.~\ref{alg:buchi} is a symbolic fixed-point algorithm that implements Eq.~\ref{eq:buchiExample:muCalcFormula} according Def.~\ref{def:prop_mu_calculus_semantics},
which defines the game $\mu$-calculus' semantics following~\cite{BJP+12}.
Likewise, Alg.~\ref{alg:buchiEnergy} is a symbolic fixed-point algorithm
that implements Eq.~\ref{eq:buchiExample:energyMuCalcFormula} according Def.~\ref{def:sysEngMuCalcSemantics}, which defines the new energy $\mu$-calculus' semantics.
Alg.~\ref{alg:buchi} uses the controllable predecessor operator $\mathit{Cpre}_{\sys}$ that implements $\circlediamond$, whereas
Alg.~\ref{alg:buchiEnergy} uses the energy controllable predecessor operator $\ECpre$ that implements $\circlediamond_\op{E}$.
$\mathit{Cpre}_{\sys}$ is defined in Sect.~\ref{sec:propMuCalculus}; $\mathit{ECpre}_{\sys}$ is defined in Def.~\ref{def:ECpre}.

We prove in Sect.~\ref{sec:EngMuCalc} and Sect.~\ref{sec:solvingBEGames} that our approach solves both the decision and the minimum credit problems with a bounded energy level accumulation.
Moreover, we augment energy $\mu$-calculus with \emph{negation} to enable $\omega$-regular energy games to be solved via their \emph{dual} games.
That is, we show that if a game $\mu$-calculus formula $\psi$ solves the $\omega$-regular game,
the energy $\mu$-calculus formula $\neg \psi^\op{E}$ dually assigns each state the \emph{maximal} initial credit for which the adversary, namely the environment, wins.

We prove the results of Sect.~\ref{sec:EngMuCalc} by using a reduction to $\omega$-regular games, which
encodes the bounded energy objective as safety constraints, following~\cite{BrimCDGR11}.

Importantly, however, our approach also solves the decision and the minimum credit problems w.r.t. the unbounded energy objective from~\cite{ChakrabartiAHS03}, namely
when the upper bound on the energy levels is set to $+\infty$.
We obtain this key result in Sect.~\ref{sec:sufficientbound} by providing answers to the following three questions for
all $\omega$-regular winning conditions:

\begin{enumerate}
  \item \label{boundQuestion:1}Is there a state that does not win w.r.t. all finite upper bounds but wins w.r.t. the bound $+\infty$? (No)
  \item \label{boundQuestion:2}Is there a sufficient finite upper bound whose increase would not introduce additional winning states? (Yes)
  \item \label{boundQuestion:3}Is there such a sufficient bound that also does not require an increase in the initial credit to win? (Yes)
\end{enumerate}

\noindent We answer the above questions by showing how to compute a sufficiently large upper bound for any $\omega$-regular winning condition.
Most importantly, this complete bound enables the use of the results obtained in Sect.~\ref{sec:EngMuCalc} and Sect.~\ref{sec:solvingBEGames},
also in case of an unbounded energy level accumulation.

\subsection{Related Work}\label{sec:related}

\subsubsection{Energy Games}\label{sec:related:energyGames}

Energy games were introduced in~\cite{ChakrabartiAHS03}. Bouyer et al.~\cite{BouyerFLMS08} further studied these games, presented fixed-point solutions,
and showed that these games are log-space equivalent to mean-payoff games~\cite{EM79}. 
Brim et al.~\cite{BrimC12,BrimCDGR11} presented strategy improvement and improved fixed-point algorithms, both of which are explicit, for energy and mean-payoff games.

The application of energy $\mu$-calculus to the $\mu$-calculus formula ${\psi}={\nu X(\circlediamond X)}$,
which solves safety games with the winning condition $\op{G}(\text{\true})$, results in
the symbolic fixed-point algorithm ${\psi^\op{E}} = {\nu X(\circlediamond_{\op{E}} X)}$ for energy games. 
Interestingly, essentially, $\psi^\op{E}$ prescribes the algorithm that was described in~\cite{BouyerFLMS08,BrimCDGR11,ChakrabartiAHS03}.
Thus, the algorithms of~\cite{BouyerFLMS08,BrimCDGR11,ChakrabartiAHS03} can be seen as a special case of our results.

Chatterjee et al.~\cite{ChatterjeeD12} have studied $\omega$-regular energy games through energy parity games.
They have shown that the decision problem is in $\text{NP} \cap \text{coNP}$ and presented a recursive algorithm in exponential time.
The work~\cite{ChatterjeeD12} has also shown that winning strategies with a finite memory of an exponential size are sufficient.
We slightly improve the memory upper bound obtained in~\cite{ChatterjeeD12} (see Sect.~\ref{sec:sufficientbound}, Cor.~\ref{cor:parityEnergyMemory}). 

Moreover, it was shown in~\cite{ChatterjeeD12} that the decision problem of mean-payoff parity games~\cite{ChatterjeeHJ05} can be reduced to that of energy parity games.
Consequently, energy $\mu$-calculus can also solve the decision problem of $\omega$-regular mean-payoff games by applying the reduction of~\cite{ChatterjeeD12} and using our results. 

The work~\cite{BouyerFLMS08} has introduced bounded variants of energy games. Among these variants is the lower-weak-upper-bound problem,
which we refer to as the bounded energy objective. The work~\cite{BouyerFLMS08} has also established a sufficiently large upper bound that
enables the solution of (unbounded) energy games. This bound has been used in~\cite{BrimC12} to solve energy games.

Moreover, since energy games may be seen as energy parity games~\cite{ChatterjeeD12} with a single priority,
in fact, we obtain the sufficient bound of~\cite{BouyerFLMS08} by invoking Lem.~\ref{lem:lemma-6-revised} for the special case where $d=1$.
To the best of our knowledge, our work is the first to generalize~\cite{BouyerFLMS08} 
by introducing sufficient bounds that enable the solution of energy games with any $\omega$-regular winning condition.


Velner et al.~\cite{VelnerC0HRR15} have studied the complexity of multi-dimensional energy and mean-payoff games where the weights are integer vectors.
They have shown that the decision problem of multi-dimensional energy games 
is $\text{coNP}$-complete and finite memory strategies are sufficient for winning. Fahrenberg et al.~\cite{FahrenbergJLS11}
have studied variants of multi-dimensional energy games with both lower and upper bounds.

Finally, Chatterjee et al.~\cite{ChatterjeeRR14} have established that strategies with an exponential memory are necessary and sufficient
for multi-dimensional energy parity games. Furthermore, they have presented
an exponential fixed-point algorithm to compute such strategies.

\subsubsection{The $\mu$-Calculus and Symbolic Algorithms}\label{sec:related:muCalculus}
Besides model checking (see, e.g.,~\cite{bradfieldmu}), the modal $\mu$-calculus has been used to solve $\omega$-regular
games (e.g.,~\cite{BJP+12,BruseFL14,EmersonJ91,KonighoferHB13,Walukiewicz1996}), as well as to synthesize winning strategies (e.g., in GR(1)~\cite{BJP+12}
and parity~\cite{BruseFL14} games).

Multi-valued or quantitative extensions of the
$\mu$-calculus have been suggested for verification of multi-valued or quantitative transition systems (e.g.,~\cite{deAlfaro,BrunsG04,Fischer2010,GrumbergLLS05,GrumbergLLS07}).
Nevertheless, such extensions have also been introduced to solve, e.g.,
probabilistic and concurrent games~\cite{AlfaroM04}, and games with imperfect information~\cite{RaskinCDH07}.
The translation of $\omega$-regular conditions to the $\mu$-calculus for the purpose
of solving the corresponding games has been studied in~\cite{AlfaroHM01} w.r.t. both Boolean and quantitative semantics. 
We apply this approach to energy games.

The semantics of energy $\mu$-calculus exploits the monotonicity of the energy objective as it maps states to the minimal winning initial credits.
It is inspired by the antichain representation used by the algorithm of~\cite{ChatterjeeRR14}, which solves multi-dimensional energy games.
Essentially, symbolic antichain representations exploit monotonicity properties
to succinctly represent the sets that the symbolic algorithm manipulates. 
The use of antichains to obtain performance improvements has been implemented for various applications, such as model checking (e.g.,~\cite{DoyenR10,WulfDHR06,WulfDMR08}),
games with imperfect information~\cite{BerwangerCWDH10,RaskinCDH07}, and LTL synthesis (e.g.,~\cite{BohyBFR13,FiliotJR11}).

The semantics of energy $\mu$-calculus prescribes symbolic algorithms that manipulate energy functions.
Therefore, implementations of the energy $\mu$-calculus should be based on symbolic data structures, and in particular,
on those that encode multi-valued functions. Such a notable data structure are Algebraic Decision Diagrams~\cite{BaharFGHMPS97,FujitaMY97} (ADDs), which generalize Binary Decision Diagrams (BDDs)~\cite{Bryant86}.

The use of ADDs to encode real-valued matrices for the analysis of probabilistic models, such as Markov chains, has been studied extensively, e.g.,~\cite{AlfaroKNPS00,BaierCHKR97,HermannsKNPS03,KwiatkowskaNP04,KwiatkowskaNP11}.
However, ADDs have only recently been studied in the context of game solving.
The work~\cite{MaozPR16} has presented an ADD-based, symbolic fixed-point algorithm for energy games and evaluated its performance.
In fact, this algorithm implements the energy $\mu$-calculus formula ${\nu X(\circlediamond_{\op{E}} X)}$ that we have considered in Sect.~\ref{sec:related:energyGames}.
The evaluation in~\cite{MaozPR16} showed that the ADD-based algorithm outperformed an alternative, BDD-based algorithm in terms of scalability.
Moreover, the work~\cite{BustanKV04} presented a symbolic ADD-based version of the well-known small progress measures (explicit) algorithm~\cite{Jurdzinski00} for parity games.
The algorithm of~\cite{BustanKV04} has recently been implemented and evaluated in~\cite{StasioMV18}.

\section{Preliminaries}
\label{sec:preliminaries}

Throughout this paper, for $a, b \in \mathbb{Z}\cup\{+\infty\}$, $[a,b]$ denotes the
set $\{z \in \mathbb{Z}\cup\{+\infty\}\mid a \leq z \leq b\}$. 
For a set of Boolean variables $\mathcal V$, a \emph{state}, $s\in 2^\mathcal V$, is a truth assignment to $\mathcal V$, an \emph{assertion} $\phi$ is a propositional formula over $\mathcal V$, $s\models \phi$ denotes that $s$ satisfies $\phi$, and $\mathcal V'$ denotes the set $\{v'\mid v\in\mathcal V\}$ of \emph{primed} variables.
We denote by $p(s)\in 2^\mathcal{V'}$ the \emph{primed version} of the state $s\in{2^\mathcal{V}}$, obtained by replacing each
$v\in{s}$ with $v'\in \mathcal{V}'$.
For $\mathcal V=\bigcup_{i=1}^k\mathcal V_i$ and truth assignments $s_i\in 2^{\mathcal V_i}$, we use $(s_1,\ldots,s_k)$ as an abbreviation for $s_1\cup{\ldots}\cup{s_k}$. Thus, we may replace expressions, e.g., $s\in 2^\mathcal V$, $s\models \varphi$, $p(s)$, and $f(s)$ with $(s_1,\dots s_k)\in 2^\mathcal V$, $(s_1,\dots,s_k)\models \varphi$, $p(s_1,\dots,s_k)$, and $f(s_1,\dots,s_k)$, respectively.
We denote by $s|_{\mathcal{Z}}$ the \emph{projection} of
$s\in{2^{\mathcal{V}}}$ to $\mathcal{Z}\subseteq{\mathcal{V}}$, i.e.,
$s|_{\mathcal{Z}}:= s \cap \mathcal{Z}$.

\subsection{Games, Game Structures, and Strategies}\label{sec:games}

We consider an infinite game played between an
environment player (\envbig) and
a system player (\sysbig) on a finite directed graph
as they move along its transitions. In each round of the game, the environment plays
first by choosing a valid input, and the system plays
second by choosing a valid output.
The goal of the system is to satisfy
the winning condition, regardless of the actions of the environment.

Formally, a game is symbolically represented by a \emph{game
structure} (GS)
$G := \langle \mathcal{V},\mathcal{X}, \mathcal{Y}, \\\rho^e,\rho^s, \varphi \rangle$~\cite{BJP+12,PitermanPS06}
that consists of the following components:

\begin{itemize}
  \item $\mathcal{V} = \{v_1, \ldots, v_n\}$: A finite set of Boolean variables. 
  \item $\mathcal{X}\subseteq{\mathcal{V}}$: A set of \emph{input variables}
  controlled by the \emph{environment} player (\envbig).
  \item $\mathcal{Y}={\mathcal{V}\setminus{\mathcal{X}}}$: A set of \emph{output
  variables} controlled by the \emph{system} player (\sysbig).
  \item $\rho^e$: An assertion over $\mathcal{V}\cup \mathcal{X'}$ that defines the environment's transitions.
   The environment uses $\rho^e$ to relate a state over $\mathcal{V}$ to \emph{possible next inputs} over $\mathcal X'$.
  \item $\rho^s$: An assertion over $\mathcal{V}\cup \mathcal V' = \mathcal{V}\cup \mathcal{X'}\cup \mathcal{Y'}$ that defines the system's transitions.
  The system uses $\rho^s$ to relate a state over $\mathcal{V}$ and an input over $\mathcal{X}'$ to \emph{possible next outputs} over $\mathcal Y'$.
  \item $\varphi$: The winning condition of the system.
\end{itemize}

\noindent We consider \emph{$\omega$-regular GSs}, i.e., GSs with 
$\omega$-regular winning conditions $\varphi$.
A state $t$ is a \emph{successor} of $s$ if $(s,
p(t)) \models \rho^e \wedge \rho^s$. The rounds of a game on $G$ form a sequence
of states $\sigma=s_0s_1\ldots$ called a \emph{play}, which satisfies 
the following conditions: (1) \emph{Consecution}: for each $i\geq0$, $s_{i+1}$
is a successor of $s_i$. (2) \emph{Maximality}: if $\sigma$ is finite, then either it
  ends with a \emph{deadlock for the environment}: $\sigma=s_0\ldots s_k$, and there is no
  input value $s_{\mathcal{X}}\in{2^\mathcal{X}}$ such that $(s_k,
  p(s_{\mathcal{X}}))\models {\rho^e}$, or it ends with a \emph{deadlock for the system}: $\sigma=s_0\ldots
  s_k s_{\mathcal{X}}$ where $s_{\mathcal{X}}\in{2^\mathcal{X}}$,
  $(s_k, p(s_{\mathcal{X}}))\models \rho^e$, and there is no output
  $s_{\mathcal{Y}}\in{2^\mathcal{Y}}$ such that
  $(s_k, p(s_{\mathcal{X}}), p(s_{\mathcal{Y}}))\models {\rho^s}$.

We denote by $\textsf{Plays}(G)$ the set of all $G$ plays. A play
$\sigma = s_{0}\ldots \in \textsf{Plays}(G)$ is \emph{from} $S \subseteq
2^{\mathcal{V}}$ if $s_{0}\in S$. A play $\sigma\in\textsf{Plays}(G)$ \emph{wins for the system} if either
$\sigma$ is finite and ends with a deadlock for the environment,
or $\sigma$ is infinite and satisfies the winning condition $\varphi$.
We denote by $\textsf{Plays}(G,\varphi)$ the set of all plays that win for
the system. If $\sigma\not\in \textsf{Plays}(G, \varphi)$, we say that $\sigma$ wins for the environment.

A \emph{strategy} for the system player is a partial function $g_{\sys}:
(2^\mathcal{V})^{+}2^{\mathcal{X}}\rightarrow 2^{\mathcal{Y}}$. It satisfies that
for every prefix $\sigma=s_0\ldots s_k \in (2^\mathcal{V})^{+}$ and
$s_{\mathcal{X}}\in{2^\mathcal{X}}$ such that $(s_k, p(s_{\mathcal{X}}))\models
\rho^e$, if $g_{\sys}$ is defined for $\sigma s_{\mathcal{X}}$, then $(s_k,
p(s_{\mathcal{X}}), p(g_{\sys}(\sigma s_{\mathcal{X}})))\models
{\rho^s}$.

Let $g_{\sys}$ be a strategy for the system, and $\sigma = s_{0}s_{1}\ldots \in
\textsf{Plays}(G)$. The prefix $s_{0}\ldots s_{k}$ of
$\sigma$ is \emph{consistent} with $g_{\sys}$ if for each $0 \leq i < k$, $g_{\sys}$ is
defined at $s_0\ldots s_i s_{i+1}|_{\mathcal{X}}$, and $g_{\sys}(s_0\ldots s_i s_{i+1}|_{\mathcal{X}}) = s_{i+1}|_{\mathcal{Y}}$.
We say that $\sigma$ is consistent with $g_{\sys}$ if all of its prefixes are
consistent with $g_{\sys}$.

The strategy $g_{\sys}$ is \emph{from $S \subseteq 2^{\mathcal{V}}$} if it is
defined (1) for every prefix $s_{0}\ldots s_{j} \in
(2^\mathcal{V})^{+}$ of a play from $S$, consistent with
$g_{\sys}$, and (2) for every input $s_{\mathcal{X}} \in 2^{\mathcal{X}}$ such that
$(s_j,p(s_{\mathcal{X}}))\models \rho^e$, and $(s_j, p(s_{\mathcal{X}}))$ is not a deadlock for the system. In case
$S= \{s\}$ for $s \in 2^{\mathcal{V}}$, we will simply write $s$. We dually define strategies and consistent plays
for the environment player.

A strategy $g_{\alpha}$ \emph{wins} for player $\alpha \in \playerset$
from $s\in 2^{\mathcal{V}}$, if it is a strategy for $\alpha$ from $s$,
and all plays from $s$ that are consistent with $g_{\alpha}$
win for $\alpha$. The assertion $W_{\alpha}$ describes the set of \emph{winning states}, i.e., from which there exists a
winning strategy for player $\alpha$. We may use the assertion
$W_{\alpha}$ interchangeably with the set $\{s\in 2^{\mathcal{V}}\mid s\models W_{\alpha}\}$.

\subsection{Weighted Game Structures and Energy Objectives}\label{sec:combinedEngObj}

We now define the energy objective. Our definition is based on both the \emph{lower-weak-upper-bound} and the \emph{lower-bound} problems introduced by Bouyer et al.~\cite{BouyerFLMS08}, while it uses
a slightly different formulation adapted for GSs.

A finite \emph{weighted game structure} (WGS)
${G^{w}} := \langle \mathcal{V}, \mathcal{X}, \mathcal{Y}, \rho^e,\rho^s, \varphi, w^s \rangle$
is a GS extended with a partial \emph{weight function} $w^s:{2^{\mathcal{V}\cup \mathcal{V}'}} \rightarrow {\mathbb{Z}}$,
defined for $\rho^s$ transitions.
Intuitively, $w^s$ describes the amount by which system's actions 
reclaim or consume a constrained resource, which we refer to as \emph{energy}.

Let $G^{w}$ be a WGS, $\sigma=s_{0}s_{1}\ldots\in\textsf{Plays}(G^{w})$ be a
$G^{w}$ play, and $\sigma[0\ldots
k]:=s_0\ldots s_k$ be a prefix of $\sigma$ for $k \in \mathbb{N}$. Given a
(finite) \emph{upper bound} $c\in{\mathbb{N}}$, and an \emph{initial credit} $c_{0} \in
[0,c]$,
the \emph{energy level under} $c$ of $\sigma[0\ldots
k]$, denoted by $\textsf{EL}_{c}(G^{w},c_0, \sigma[0\ldots k])$, is
the sum of $c_{0}$ and the weights that $w^{s}$ assigns to the transitions of
$\sigma[0\ldots k]$, such that whenever it exceeds the upper bound $c$, it is
truncated to $c$. Formally, $\textsf{EL}_{c}(G^{w}, c_0, \sigma[0\ldots k]) :=
r_{k}$, where $r_0 := c_0$ and for each $i\in [1,k]$, $r_{i} := \min\lbrack c,
r_{i-1}+w^s(s_{i-1},p(s_{i}))\rbrack$.
In Sect.~\ref{sec:sufficientbound}, we also consider the (unbounded) energy
level of $\sigma[0\ldots k]$, $\textsf{EL}_{+\infty}(G^{w},c_0, \sigma[0\ldots
k])$ where $c = +\infty$ and $c_{0} \in \mathbb{N}$.
Note that in this special case, it is simply the sum of $c_{0}$ and the weights along $\sigma[0\ldots
k]$, i.e., $\textsf{EL}_{+\infty}(G^{w},c_0, \sigma[0\ldots k]) = c_{0}
+ \sum_{i=1}^{k} w^s(s_{i-1},p(s_{i}))$.

A WGS $G^{w}$ represents a game with both \emph{qualitative}
and \emph{quantitative} winning conditions. The former is specified by
the $\omega$-regular condition $\varphi$, and the latter is the energy objective that requires to keep
the energy levels of all plays' prefixes, non-negative. Formally, given
an upper bound $c\in{\mathbb{N}}\cup \{+\infty\}$ and an initial credit
$c_{0} \not = +\infty$, $c_{0} \leq c$, the \emph{energy objective w.r.t. $c$
for $c_{0}$} is
$\textsf{E}_{c}(G^{w}, c_{0}):=\{\sigma \in \textsf{Plays}(G^{w})\mid\forall j \geq 0:
\textsf{EL}_{c}(G^{w}, c_{0}, \sigma[0\ldots j])\geq{0}\}$, and
we say that a play $\sigma \in
\textsf{Plays}(G^{w})$ \emph{wins the energy objective
w.r.t. $c$ for $c_{0}$} if $\sigma \in
\textsf{E}_{c}(G^{w}, c_{0})$. Thus, $\sigma$
\emph{wins the $\varphi$-energy objective w.r.t. $c$ for $c_{0}$}
iff $\sigma \in \textsf{Plays}(G^{w}, \varphi)\cap \textsf{E}_{c}(G^{w},
c_{0})$.

Let $c\in{\mathbb{N}}\cup \{+\infty\}$ be an upper bound.
A strategy $g$ for the system (resp.
environment) \emph{wins} \emph{from} $s\in 2^{\mathcal{V}}$ \emph{w.r.t.
$c$} \emph{for an initial credit} $c_{0} \not = +\infty$, $c_{0} \leq c$, if it
is a strategy for the system (resp. environment) from $s$, and all plays
$\sigma$ that are from $s$ and consistent with $g$, win (resp. do not win)
the $\varphi$-energy objective w.r.t. $c$ for $c_{0}$.
A state $s\in 2^{\mathcal{V}}$
\emph{wins} for the system (resp. environment) \emph{w.r.t. $c$ for an initial
credit $c_{0}$}, if there exists a strategy that wins for the system (resp.
environment) from $s$ w.r.t. $c$ for $c_{0}$.
We say that $s\in 2^{\mathcal{V}}$ \emph{wins for the system
w.r.t. $c$}, if it wins for the system w.r.t. $c$ for \emph{some} initial credit
$c_{0}$. Otherwise, if $s$ wins for the environment w.r.t. $c$ for
\emph{all} initial credits $c_{0} \not = +\infty$, $c_{0} \leq c$, we say that
it wins for the environment. Accordingly, we denote by $W_{\alpha}(c)$ the set
of states that win for player $\alpha \in \playerset$ w.r.t. $c$.

Further, note that the energy objective is \emph{monotone}
w.r.t. both the initial credit and the
bound. That is, for all upper bounds $c, c^{h} \in\mathbb{N}$ and initial
credits $c_{0} \in [0,c]$, $c^{h}_{0} \in [0,c^{h}]$ such that
$c \leq c^{h}$ and $c_{0} \leq c^{h}_{0}$: $\textsf{E}_{c}(G^{w},
c_{0}) \subseteq \textsf{E}_{c^{h}}(G^{w}, c^{h}_{0})$, $\textsf{E}_{+\infty}(G^{w},
c_{0}) \subseteq \textsf{E}_{+\infty}(G^{w}, c^{h}_{0})$ and $\textsf{E}_{c}(G^{w},
c_{0}) \subseteq \textsf{E}_{+\infty}(G^{w}, c_{0})$. This gives rise to
consider in Sect.~\ref{sec:solvingEnergyGames} the
\emph{optimal} (i.e., minimal) initial credit and a sufficiently large upper
bound for which the system wins.

\subsection{\texorpdfstring{$\mu$-Calculus Over Game Structures}{mu-Calculus Over Game Structures}}\label{sec:propMuCalculus}

We consider the logic of the modal $\mu$-calculus \cite{Kozen} over GSs, 
and repeat its definition from~\cite{BJP+12} below.
It will be useful
in Sect.~\ref{sec:EngMuCalc} where we introduce a multi-valued extension thereof.

\begin{defi}[$\mu$-calculus: syntax]\label{def:prop_mu_calculus_grammar}
Let $\mathcal{V}$ be a set of Boolean variables,
and let $\mathit{Var} = \{X, Y, \ldots\}$ be a set of relational variables.
The formulas of $\mu$-calculus (in positive form) are built as
follows:
\[\psi::=~v~|~\neg{v}~|~X~|~
\psi\vee\psi~|~\psi\wedge\psi~|~\circlediamond\psi|~\circlebox\psi~|~\mu{X}\psi|~\nu{X}\psi\]
where $v\in\mathcal{V}$, $X\in\mathit{Var}$, and $\mu$ and $\nu$ denote the least and the greatest fixed-point operators, respectively.
\end{defi}

We denote by $\mathcal{L}_{\mu}$ the set of all formulas generated
by the grammar of Def.~\ref{def:prop_mu_calculus_grammar}.
We further denote by $\sysFormulas$ (resp. $\envFormulas$) the subset of
$\mathcal{L}_{\mu}$ that consists of all formulas in which the modal operator $\circlebox$ (resp. $\circlediamond$)
does \emph{not} occur. We will refer to $\sysFormulas$
(resp. $\envFormulas$) formulas as $\sysmu$ (resp. $\envmu$) formulas.

In this paper, we may refer to the \emph{alternation depth}~\cite{EmersonL86, Niwinski86} of a formula $\psi \in \mathcal{L}_{\mu}$, i.e.,
the number of alternations between interdependent, nested least and greatest
fixed-point operators in $\psi$. For the formal definition, see, e.g.,~\cite[Chapter~10]{2001automata}.

\begin{defi}[$\mu$-calculus: semantics]\label{def:prop_mu_calculus_semantics}
We inductively define the set $\llbracket \psi\rrbracket^{G}_{\mathcal{E}}$ of
states in which $\psi\in \mathcal{L}_{\mu}$ is true w.r.t. a finite GS, $\GS$, and a valuation
$\mathcal{E}: \mathit{Var} \rightarrow (2^{\mathcal{V}} \rightarrow \{0,1\})$, as follows:\footnote{If all of the relational variables in $\psi$ are bound by fixed-point
operators, i.e., $\psi$ is a closed formula, we may omit $\mathcal{E}$ from the
semantic brackets.}

\begin{itemize}
  \setlength\itemsep{0.5em}
  \item For $v\in\mathcal{V}$,
  $\llbracket{v}\rrbracket^{G}_{\mathcal{E}} = \{s\in 2^{\mathcal{V}}~|~ s\models v\}$;
  $\llbracket{\neg{v}}\rrbracket^{G}_{\mathcal{E}} = \{s\in 2^{\mathcal{V}}~|~
  s\not\models v\}$.
  
  \item For $X\in{\mathit{Var}}$, $\llbracket{X}\rrbracket^{G}_{\mathcal{E}} =
  \mathcal{E}(X)$.
  
  \item $\llbracket{\phi_1\vee{\phi_2}}\rrbracket^{G}_{\mathcal{E}} =
  \llbracket{\phi_1}\rrbracket^{G}_{\mathcal{E}}\cup
  \llbracket{\phi_2}\rrbracket^{G}_{\mathcal{E}}$; 
  $\llbracket{\phi_1\wedge{\phi_2}}\rrbracket^{G}_{\mathcal{E}} =
  \llbracket{\phi_1}\rrbracket^{G}_{\mathcal{E}}\cap
  \llbracket{\phi_2}\rrbracket^{G}_{\mathcal{E}}$.
  
  \item
  $\llbracket{\circlediamond\phi}\rrbracket^{G}_{\mathcal{E}} =
  \left\{
  s\in 2^{\mathcal{V}} \Bigl\vert
  \begin{aligned}
&\forall s_{\mathcal{X}}\in 2^{\mathcal{X}},
(s,p(s_{\mathcal{X}}))\models \rho^e \Rightarrow \exists
s_{\mathcal{Y}}\in{2^{\mathcal{Y}}} \text{ such that }\\
&(s, p(s_{\mathcal{X}}), p(s_{\mathcal{Y}}))\models \rho^s
\text{ and } (s_{\mathcal{X}}, s_{\mathcal{Y}})\in
\llbracket{\phi}\rrbracket^{G}_{\mathcal{E}}
\end{aligned}
\right\}$.
\item $\llbracket{\circlebox\phi}\rrbracket^{G}_{\mathcal{E}} =
\left\{
s\in 2^{\mathcal{V}} \Bigl\vert
\begin{aligned}
&\exists s_{\mathcal{X}}\in 2^{\mathcal{X}} \text{ such that }
(s,p(s_{\mathcal{X}}))\models \rho^e \text{ and } \forall
s_{\mathcal{Y}}\in{2^{\mathcal{Y}}},\\
&(s, p(s_{\mathcal{X}}), p(s_{\mathcal{Y}}))\models \rho^s
\Rightarrow (s_{\mathcal{X}}, s_{\mathcal{Y}})\in
\llbracket{\phi}\rrbracket^{G}_{\mathcal{E}}
\end{aligned}
\right\}$.

\item
$\llbracket\twolinescurly{\mu}{\nu} X\phi\rrbracket^{G}_{\mathcal{E}} =
\twolinescurly{\bigcup_{i}}{\bigcap_{i}}{S_i}$, where
$\twolinescurly{S_0=\emptyset}{S_0=2^{\mathcal{V}}}$,
$S_{i+1}= 
\llbracket{\phi}\rrbracket^{G}_{\mathcal{E}[X\mapsto{S_i}]}$, and
$\mathcal{E}[X\mapsto{S}]$ denotes\\[0.5em]
the valuation which is like $\mathcal{E}$ except that it maps $X$ to $S$.
\end{itemize}
\end{defi}

\noindent Note that Def.~\ref{def:prop_mu_calculus_semantics} relates to
game solving rather than to model checking
(cf.~\cite{EmersonJS01,EmersonL86,Schneider2004}).
That is,
the classical predecessor operators from~\cite{Kozen} are replaced with the
\emph{controllable predecessor} operators:
$\mathit{Cpre}_{\sys},~\mathit{Cpre}_{\env}:2^{2^{\mathcal{V}}}\rightarrow{2^{2^{\mathcal{V}}}}$.
The set $\llbracket{\circlediamond\phi}\rrbracket^{G}_{\mathcal{E}} =
\mathit{Cpre}_{\sys}(\llbracket\phi\rrbracket^{G}_{\mathcal{E}})$ consists of all states
from which the system can force the environment in a single step to reach a state in the set
$\llbracket{\phi}\rrbracket^{G}_{\mathcal{E}}$, and dually,
$\llbracket{\circlebox\phi}\rrbracket^{G}_{\mathcal{E}}=
\mathit{Cpre}_{\env}(\llbracket\phi\rrbracket^{G}_{\mathcal{E}})$ consists of all states
from which the environment can force the system in a single step to reach a state in
$\llbracket{\phi}\rrbracket^{G}_{\mathcal{E}}$.

De Alfaro et al.~\cite{AlfaroHM01} have shown that
$\omega$-regular GSs can be solved by evaluating closed $\mathcal{L}_{\mu}$ formulas.
That is, for every $\omega$-regular winning condition $\varphi$, there
is a closed $\sysmu$ (resp. $\envmu$) formula $\psi_{\varphi} \in \sysFormulas$
($\psi_{\neg\varphi} \in \envFormulas$) that for all GSs, $G$, computes the set of states
that win for the system (environment) player, i.e., $W_{\sys} = \llbracket{\psi_{\varphi}}\rrbracket^{G}$
($W_{\env} = \llbracket{\psi_{\neg\varphi}}\rrbracket^{G}$).
We say that $\psi_{\varphi} \in \sysFormulas$ \emph{matches} $\varphi$
if for all GSs $G$, $W_{\sys} = \llbracket{\psi_{\varphi}}\rrbracket^{G}$, and
dually for $\psi_{\neg\varphi} \in \envFormulas$.
%

\section{\texorpdfstring{Energy $\mu$-Calculus Over Weighted Game Structures}{Energy
mu-Calculus Over Weighted Game Structures}}\label{sec:EngMuCalc}

This section introduces \emph{energy $\mu$-calculus}, 
a multi-valued extension of the $\mu$-calculus~\cite{Kozen} over GSs~\cite{BJP+12,EmersonJ91}.
First, Sect.~\ref{sec:EngMuCalcSyntaxSemantics} presents the syntax and semantics thereof.
It identifies two dual syntactic fragments analogous to $\sysmu$ and $\envmu$ formulas from Sect.~\ref{sec:propMuCalculus},
and presents their semantics separately in Sect.~\ref{sec:EngMuCalcSysSemantics} and Sect.~\ref{sec:EngMuCalcEnvSemantics}.
Second, Sect.~\ref{sec:EngMuCalcCorrect} shows that the semantics of each fragment encodes the energy objective w.r.t. finite upper bounds.
The proofs for the theorems, propositions, and lemmas of this section appear in Appx.~\ref{app:energyMuCalc:proofs}.
\subsection{\texorpdfstring{Energy $\mu$-Calculus: Syntax and Semantics}{Energy mu-Calculus: Syntax and Semantics}}\label{sec:EngMuCalcSyntaxSemantics}

Let $\EFormulas$ denote the set of formulas generated by the following grammar:

\begin{defi}[Energy $\mu$-calculus: syntax]\label{def:EngMuCalcGrammar}
Let $\mathcal{V}$ be a set of Boolean variables,
and let $\mathit{Var} = \{X, Y, \ldots\}$ be a set of relational variables.
The syntax of energy $\mu$-calculus (in positive form) is as follows:
\[\psi::=~v~|~\neg{v}~|~X~|~\psi\vee\psi~|~\psi\wedge\psi~|~
\circlediamond_{\op{E}}\psi|~\circlebox_{\op{E}}\psi~|~\mu{X}\psi|~\nu{X}\psi\]
where $v \in \mathcal V$ and $X \in \mathit{Var}$.
\end{defi}

We denote by $\sysEFormulas$ (resp. $\envEFormulas$) the subset of
$\EFormulas$ that consists of all formulas in which $\circlebox_{\op{E}}$ (resp. $\circlediamond_{\op{E}}$)
does \emph{not} occur. We refer to $\sysEFormulas$
(resp. $\envEFormulas$) formulas as \emph{$\sysEmu$} (resp. \emph{$\envEmu$}) formulas.
Further, let $\psi^\op{E}\in \mathcal{L}_{e\mu}$ denote
the energy $\mu$-calculus formula obtained from $\psi\in \mathcal{L}_{\mu}$
by replacing all occurrences of $\circlediamond$ and $\circlebox$ with $\circlediamond_{\op{E}}$ and $\circlebox_{\op{E}}$, respectively.

\subsubsection{\texorpdfstring{$\sysCapEnergy$ $\mu$-Calculus}{sys-Energy mu-Calculus}}\label{sec:EngMuCalcSysSemantics}

The value of a \emph{$\sysEmu$} formula $\psi^\op{E}\in \sysEFormulas$, which
we formally define below, is a function that maps each state of the underlying WGS to the
\emph{minimum} initial credit for which that state wins for the system w.r.t.
a finite upper bound, provided that
$\psi\in \sysFormulas$ matches the underlying winning condition (see Thm.~\ref{thm:sysEngMuCalcCorrectness}).
Accordingly, we define the semantics of $\sysEmu$ formulas w.r.t. a finite upper bound $c\in \mathbb{N}$
and a WGS $G^{w} = \langle G, w^s\rangle$, and use $\GEc$ as a shorthand for the tuple $\langle G^{w}, c \rangle$. 

For $c\in\mathbb{N}$, we respectively define the finite sets
$\mathit{E(c)} := [0,c] \cup \{+\infty\}$ and $\mathit{EF(c)} := \mathit{E(c)}^{2^{\mathcal{V}}}$
of initial credits up to $c$ and \emph{energy functions} from states to $\mathit{E(c)}$.
The semantics' definition of $\sysEmu$ formulas makes use of the \emph{energy controllable predecessor}
operator $\ECpre:\mathit{EF(c)}\rightarrow{\mathit{EF(c)}}$, which we define below in Def.~\ref{def:ECpre},
and corresponds to the classical $\mathit{Cpre_{\sys}}$ operator of Def.~\ref{def:prop_mu_calculus_semantics}.
Informally, for all $f\in \mathit{EF(c)}$, $\ECpre(f)$ denotes the
energy function that maps each state $s \in 2^{\mathcal{V}}$ to the minimum initial credit
sufficient for the system to force the environment to move in a single step from $s$ to some successor $t$ with an energy level at least $f(t)$.


\begin{defi}[Energy controllable predecessor operator]\label{def:ECpre}
For all WGSs $\langle G, w^s\rangle$, upper bounds $c\in \mathbb{N}$, energy functions $f\in \mathit{EF(c)}$, and states $s \in 2^{\mathcal{V}}$,
\begin{align*}
&{\ECpre(f)(s)} := {\max\limits_{s_{\mathcal{X}}\in{2^\mathcal{X}}}
\lbrack\min\limits_{s_{\mathcal{Y}}\in{2^\mathcal{Y}}}\textsf{EC}_c((s,p(s_{\mathcal{X}},s_{\mathcal{Y}})),
f(s_{\mathcal{X}},s_{\mathcal{Y}}))\rbrack}\\&
\text{where $\textsf{EC}_c:2^{\mathcal{V}\cup\mathcal{V'}}\times{\mathit{E(c)}}\rightarrow{\mathit{E(c)}}$ and for all
$s \in 2^{\mathcal{V}}$, $s' \in 2^{\mathcal{V'}}$, and $e \in \mathit{E(c)}$,}\\&
{\textsf{EC}_c((s, s'),e)} =
\begin{cases}
0,\ &\mbox{if $(s,s')\not\models{\rho^e}$}\\
+\infty,\ &\mbox{if $e=+\infty$ or $(s,s') \models {\rho^e \wedge {\neg\rho^s}}$}\\
+\infty,\ &\mbox{if $e - w^{s}(s,s') > c$}\\
\max{\lbrack0, e - w^{s}(s,s')\rbrack},\
& \mbox{otherwise}
\end{cases}&
\end{align*}
\end{defi}

\noindent In Def.~\ref{def:ECpre}, $\ECpre$ uses the function $\textsf{EC}_c:2^{\mathcal{V}\cup\mathcal{V'}}\times{\mathit{E(c)}}\rightarrow{\mathit{E(c)}}$.
Intuitively, $\textsf{EC}_c((s,s'),e)$ is the minimum initial credit sufficient for the system to traverse
the transition $(s,s')$, provided that $e$ is the minimum initial credit required to proceed from $s'$.
Specifically, if $(s,s')$ is invalid for the environment (i.e., $(s,s')\not\models{\rho^e}$), then the initial credit $0$ is sufficient,
and if $\textsf{EC}_c((s,s'),e) = +\infty$, there is no initial credit $c_{0} \leq c$ sufficient for
traversing $(s,s')$. The latter
holds when either $(s,s')$ is only valid for the environment (i.e., $(s,s')\models{\rho^e\wedge{\neg\rho^s}}$),
there is no initial credit sufficient to proceed from $s'$ (i.e., $e=+\infty$), or
the minimum initial credit required to traverse $(s,s')$ exceeds the upper bound (i.e., $e - w^{s}(s,s') > c$).

Def.~\ref{def:sysEngMuCalcSemantics} formally defines the value of $\psi \in \sysEFormulas$
w.r.t. WGSs, finite upper bounds, 
and \emph{valuations over $\mathit{EF(c)}$} that map each relational variable in $\mathit{Var}$
to an energy function in $\mathit{EF(c)}$.

\begin{defi}[$\sysEnergy$ $\mu$-calculus: semantics]\label{def:sysEngMuCalcSemantics}
The semantics $\llbracket \psi\rrbracket^{\GEc}_{\eValuation}$ of
$\psi\in{\sysEFormulas}$ w.r.t. a WGS 
$G^{w} = \langle \mathcal{V},\mathcal{X}, \mathcal{Y}, \rho^e,\rho^s, \varphi, w^s\rangle$, a finite upper bound $c \in \mathbb{N}$, and a valuation $\eValuation: \mathit{Var} \rightarrow \mathit{EF(c)}$ over $\mathit{EF(c)}$, is inductively defined for all states
$s\in{2^{\mathcal{V}}}$, as follows:\footnote{We may drop the valuation $\eValuation$ 
from the semantic brackets for closed formulas.}

\begin{itemize}
\setlength\itemsep{0.5em}
\item For $~v\in\mathcal{V}$,\begin{minipage}{0.5\textwidth}
\begin{align*}&\llbracket{v}\rrbracket^{\GEc}_{\eValuation}(s) =
\begin{cases}
0, & \text{ if } s\vDash{v} \\
+\infty, & \text{ if } s\nvDash{v}
\end{cases}.\\&
\llbracket{\neg{v}}\rrbracket^{\GEc}_{\eValuation}(s) =
\begin{cases}
+\infty, & \text{ if } s\vDash{v} \\
0, & \text{ if } s\nvDash{v}
\end{cases}.&
\end{align*}\end{minipage}
\item For $~X\in{\mathit{Var}}$, $\llbracket{X}\rrbracket^{\GEc}_{\eValuation}(s) = \eValuation(X)(s)$.
\item $\llbracket{\phi_1\vee{\phi_2}}\rrbracket^{\GEc}_{\eValuation}(s) =
\min(\llbracket{\phi_1}\rrbracket^{\GEc}_{\eValuation},
\llbracket{\phi_2}\rrbracket^{\GEc}_{\eValuation})(s)$.
\item $\llbracket{\phi_1\wedge{\phi_2}}\rrbracket^{\GEc}_{\eValuation}(s) =
\max(\llbracket{\phi_1}\rrbracket^{\GEc}_{\eValuation},
\llbracket{\phi_2}\rrbracket^{\GEc}_{\eValuation})(s)$.
\item $\llbracket{\circlediamond_{\op{E}}\phi}\rrbracket^{\GEc}_{\eValuation}(s) = \ECpre(\llbracket\phi\rrbracket^{\GEc}_{\eValuation})(s)$.
\item \parbox[t]{\linewidth}{$\llbracket\twolinescurly{\mu}{\nu} X\phi\rrbracket^{\GEc}_{\eValuation}(s) =
\twolinescurly{\mathit{lfp}}{\mathit{gfp}} (\lambda{f}.\llbracket\phi\rrbracket^{\GEc}_{\eValuation[X\mapsto{f}]})(s) =
\twolinescurly{\min\limits_{i}}{\max\limits_{i}}\lbrack{h_i}\rbrack(s)$,\\
where $\twolinescurly{h_0=f_{+\infty}}{h_0=f_{0}}$ and ${h_{i+1}} = {\llbracket{\phi}\rrbracket^{\GEc}_{\eValuation[X\mapsto{h_i}]}}$.}
\end{itemize}
\end{defi}

\noindent In Def.~\ref{def:sysEngMuCalcSemantics}, $\mathit{lfp}(g)$ and $\mathit{gfp}(g)$ respectively denote the least
and greatest fixed points of $g:\mathit{EF(c)}\rightarrow{\mathit{EF(c)}}$, whose existence will be proved later in this subsection.
For $x \in \{+\infty, 0\}$, $f_{x}$ denotes the constant energy function that maps all states to 
$x$,
and $\eValuation[X\mapsto{f}]$ denotes the valuation which is like $\eValuation$ except that it maps $X$ to $f \in \mathit{EF(c)}$.
Intuitively, in all states $s$ that satisfy an assertion $\psi\in \sysEFormulas$,
the value of $\psi$ is $0$, which is the minimum initial credit sufficient
for the system to enforce $\psi$ from $s$, and enforcing $\phi_{1} \wedge \phi_{2} \in \sysEFormulas$
from a state requires the maximum of the values of $\phi_{1}$ and $\phi_{2}$ (dually for $\vee$ and minimum).
This intuition is translated in Def.~\ref{def:sysEngMuCalcSemantics} to the use of pointwise $\min$ and $\max$ operations that
are respectively the join and meet operations of the \emph{energy function lattice},
$\mathit{EFL(c)}:=\langle \mathit{EF(c)}, \min, \max, f_{+\infty}, f_{0} \rangle$, which replaces the powerset lattice of Def.~\ref{def:prop_mu_calculus_semantics}, and
$f_{+\infty}$ and $f_{0}$ are its bottom and top elements, respectively.
%
%

We also characterize $\mathit{EFL(c)}$ as a partially ordered set
by augmenting
$\mathit{E(c)}$
with the linear order $\preceq$ such that for all $x,y\in{\mathit{E(c)}}
: x\preceq{y}$ iff $x \geq y$, and defining the pointwise partial order
$\preceq$ on $\mathit{EF(c)}$, such that for all $f, g \in \mathit{EF(c)}$:
\[f\preceq{g} \text{ iff } {f} = {\max(f,g)} \text{ iff for all } s \in {2^{\mathcal{V}}}:
f(s)\preceq{g(s)}.\]

Def.~\ref{def:sysEngMuCalcSemantics} uses the \emph{dual} $\min$ and $\max$ operations and thus
the \emph{inverse} order of $\leq$, which reflects the notion that the less required initial credit,
the better for the system player. This design choice maintains correspondence between the values of $\psi\in \sysFormulas$ and $\psi^\op{E}\in \sysEFormulas$
(see Lem.~\ref{lem:sysEngMuCalc}). Importantly, it keeps the classification of $\mu$ and $\nu$ formulas as liveness and safety properties~\cite{Bradfield}.
As an example, for $p \in \mathcal{V}$, consider the $\mu$-formula $\psi_{\diamond{p}} := \mu X (p \vee \circlediamond X)$ that
solves the $p$-states \emph{reachability} game~\cite{2001automata,Thomas95}.
If we used the order $\leq$, we would need to take the $\nu$-formula, $\nu X (p \wedge \circlediamond_{\op{E}} X)$,
instead of $\psi_{\diamond{p}}^\op{E}$, to solve the corresponding reachability energy game, while $\nu X (p \wedge \circlediamond X)$, in fact,
solves the dual $p$-states \emph{safety} game. Therefore, instead of $x \geq y$, we write $x \preceq y$ to reflect that, although a smaller integer, $y$ is an element greater than $x$.

Since $\mathit{EFL(c)}$ is a finite, complete lattice, it follows from the Knaster-Tarski fixed-point theorem~\cite{Tarski1955} that
monotonicity of $\lambda{f}.\llbracket\psi\rrbracket^{\GEc}_{\eValuation[X\mapsto{f}]}$ w.r.t. $\preceq$ for all $\psi \in \sysEFormulas$
guarantees the existence of its extremal fixed points,
each of which can be computed using fixed-point iteration (as in Def.~\ref{def:sysEngMuCalcSemantics}) that stabilizes at
most after $2^{|\mathcal{V}|}\cdot (c+1)$ iterations.
We claim that Def.~\ref{def:sysEngMuCalcSemantics} is $\preceq$-monotone due to the following:
(1) Def.~\ref{def:EngMuCalcGrammar} is in positive form where 
negation only applies to the Boolean variables
$\mathcal{V}$;
(2) as shown by Prop.~\ref{prop:ECpreMonotone}, the $\mathit{ECpre_{\sys}}$ operator
is $\preceq$-monotone; (3) monotonicity is closed under function composition,
meet and join operations, and the fixed-point operators are monotone (for detailed proofs, see, e.g.,
Lem.~3.16,
and Lem.~3.17 in~\cite{Schneider2004}).


\begin{prop}\label{prop:ECpreMonotone}
$\ECpre:\mathit{EF(c)}\rightarrow{\mathit{EF(c)}}$
from Def.~\ref{def:ECpre} is $\preceq$-monotone. That is,
for all $f,g \in \mathit{EF(c)}$: if $f\preceq{g}$
then $\ECpre(f) \preceq \ECpre(g)$.
\end{prop}

\subsubsection{\texorpdfstring{$\envCapEnergy$ $\mu$-Calculus}{env-Energy mu-Calculus}}\label{sec:EngMuCalcEnvSemantics}
So far, we have considered $\omega$-regular energy games from
the system player's perspective who aims 
to minimize the required initial credit.
We now consider the dual perspective of the environment player,
and encode it in the semantics of \emph{$\envEmu$} formulas.
Informally, given $\psi\in \envFormulas$ that matches the winning condition,
the value of $\psi^\op{E}\in \envEFormulas$
in a state of the underlying WGS
corresponds to the \emph{maximum} initial credit for which
that state wins for the environment w.r.t.
a finite upper bound (see Thm.~\ref{thm:envEngMuCalcCorrectness}).

The formal semantics of all $\psi\in \envEFormulas$ in which
$\circlebox_{\op{E}}$ does \emph{not} occur is the same
as that defined previously in Def.~\ref{def:sysEngMuCalcSemantics}
for $\sysEmu$ formulas. We now treat the remaining case
of all $\psi \in \envEFormulas$ such that
$\psi = \circlebox_{\op{E}}\phi$ and $\phi \in \envEFormulas$.

In order to obtain
a \emph{duality} between the controllable predecessor operators of the two players, as exists
in the Boolean powerset semantics of Def.~\ref{def:prop_mu_calculus_semantics},
we first
augment
the
energy function lattice
$\mathit{EFL(c)}$ with a pointwise unary \emph{negation} operation ${\sim} : {{\mathit{E(c)}} \rightarrow {\mathit{E(c)}}}$,
such that for every $x\in{\mathit{E(c)}}$,
\begin{align}\label{eq:negation}
&{\sim x} =
\begin{cases}
+\infty,& \text{if  $x = 0$}\\
0,& \text{if $x = +\infty$}\\
c+1-x,& \text{otherwise}
\end{cases}&
\end{align}

\noindent Lem.~\ref{lem:deMorganAlgebra} shows that $\sim$ is an involution that satisfies De Morgan laws.

\begin{lem}\label{lem:deMorganAlgebra}
$\langle \mathit{EFL(c)},\sim \rangle$ is a De Morgan algebra.
\end{lem}

We denote by $\ECpreEnv$ the dual operator of ${\ECpre}:{{\mathit{EF(c)}}\rightarrow{{\mathit{EF(c)}}}}$
from Def.~\ref{def:ECpre}, i.e., such that for all $f \in \mathit{EF(c)}$, ${\ECpreEnv(f)} = {{\sim \ECpre({\sim f})}}$, 
and complete the definition of $\llbracket\psi\rrbracket^{\GEc}_{\eValuation}$ for all $\psi\in \envEFormulas$
by defining ${\llbracket{\circlebox_{\op{E}}\phi}\rrbracket^{\GEc}_{\eValuation}} :=
{\ECpreEnv(\llbracket\phi\rrbracket^{\GEc}_{\eValuation})}$
for all $\phi\in \envEFormulas$.
We provide the explicit definition of $\ECpreEnv$ in Appx.~\ref{app:energyMuCalc:extDefinitions}.

It follows from Prop.~\ref{prop:ECpreMonotone}
that $\ECpreEnv$ is also $\preceq$-monotone.
Hence, we conclude that the semantics of both
$\sysEmu$ and $\envEmu$ formulas is $\preceq$-monotone, and
consequently well-defined. Further, for all $\psi \in \EFormulas$,
computing $\llbracket\psi\rrbracket^{\GEc}_{\eValuation}$ using optimizations from~\cite{BrowneCJLM97} requires
$O(n^{\lfloor d/2 \rfloor +1 })$ iterations, where $n=2^{|\mathcal{V}|}\cdot (c+1)$
is the height of $\mathit{EFL(c)}$ and $d$ is the alternation depth of $\psi$.

As a final remark, we observe that Lem.~\ref{lem:deMorganAlgebra} allows us to
add \emph{negation} to the logic of energy $\mu$-calculus, namely to define:
${\llbracket \neg\psi \rrbracket^{G^w(c)}_\eValuation} = {{\sim}{\llbracket \psi \rrbracket^{G^w(c)}_\eValuation}}$.
That is permitted since Lem.~\ref{lem:deMorganAlgebra} implies the correctness of the following well-known equations (see, e.g., Lem.~2.48 and Lem.~3.13 in~\cite{Schneider2004}):
\begin{align}
\label{eq:negation1}&{\llbracket \neg\neg\psi \rrbracket^{G^w(c)}_\eValuation} = {\llbracket \psi \rrbracket^{G^w(c)}_\eValuation}.\\
\label{eq:negation2}&{\llbracket \neg(\psi \wedge (\text{resp. } \vee)~\xi) \rrbracket^{G^w(c)}_\eValuation} = \llbracket \neg\psi \vee (\text{resp. } \wedge)~\neg\xi \rrbracket^{G^w(c)}_\eValuation.\\
\label{eq:negation3}&{\llbracket \neg(\circlediamond_{\op{E}} \text{(resp. $\circlebox_{\op{E}}$)}~\psi) \rrbracket^{G^w(c)}_\mathcal D} = 
{\llbracket \circlebox_{\op{E}} \text{(resp. $\circlediamond_{\op{E}}$)}~\neg\psi \rrbracket^{G^w(c)}_\mathcal D}.\\
\label{eq:negation4}&{\llbracket \neg(\mu\text{(resp. $\nu$)}X~\psi(X)) \rrbracket^{G^w(c)}_\eValuation} = {\llbracket \nu\text{($\text{resp. }\mu$)}X~(\neg\psi(\neg X))\rrbracket^{G^w(c)}_\eValuation}.&
\end{align}


\noindent However, in order to keep the semantics monotone, all sub-formulas of the form $\mu X \phi$ or $\nu X \phi$ must satisfy that
all free occurrences of $X$ in $\phi$ fall under an even number of negations.

\subsection{\texorpdfstring{Energy $\mu$-Calculus: Correctness}{Energy mu-Calculus: Correctness}}\label{sec:EngMuCalcCorrect}

Let $\varphi$ be an $\omega$-regular condition,
and let $\psi_{\varphi}\in \sysFormulas$ and $\psi_{\neg\varphi}\in \envFormulas$
be closed formulas that match $\varphi$.
Let $G^{w} = (G, w^s)$ be a WGS with $\varphi$ as its winning condition, and let
$c \in \mathbb{N}$ be a finite upper bound.

In this section we prove the following theorems. 

\begin{thm}[$\sysEnergy$ $\mu$-calculus: correctness]\label{thm:sysEngMuCalcCorrectness}
For all states $s\in{2^{\mathcal{V}}}$,
if $\llbracket{\psi^{\op{E}}_{\varphi}\rrbracket^{\GEc}}(s) \not = +\infty$ then $\llbracket{\psi^{\op{E}}_{\varphi}\rrbracket^{\GEc}}(s)$ is the
minimum initial credit for which the system wins from $s$ w.r.t. $c$ in
$G^{w}$. Otherwise, $s$ does not win for the system w.r.t. $c$.
\end{thm}

\begin{thm}[$\envEnergy$ $\mu$-calculus: correctness]\label{thm:envEngMuCalcCorrectness}
For all states $s\in{2^{\mathcal{V}}}$,
if $\llbracket{\psi^{\op{E}}_{\neg\varphi}\rrbracket^{\GEc}}(s) \not
= +\infty$ then $c - \llbracket{\psi^{\op{E}}_{\neg\varphi}\rrbracket^{\GEc}}(s)$ is the
maximum initial credit for which the environment wins from $s$ w.r.t. $c$ in $G^{w}$.
Otherwise, $s$ does not win for the environment w.r.t. $c$ for all initial credits $c_{0}\in [0,c]$.
\end{thm}

The proofs of the above theorems rely on
an alternative solution to $\omega$-regular energy games via a reduction
to classical $\omega$-regular games.
Below, Def.~\ref{def:naiveReduction} defines this reduction, which is inspired by the reduction of energy to safety games from~\cite{BrimCDGR11}
and encodes the energy objective by adding additional safety constraints, all of which are defined over new
system controlled variables, to the system's transitions.

Notice that the reduction in Def.~\ref{def:naiveReduction} is presented here only as part of the correctness proof.
It is \emph{not} part of our energy $\mu$-calculus based algorithms that solve $\omega$-regular energy games.
We describe these algorithms later in Sect.~\ref{sec:solvingEnergyGames}.


\begin{defi}[Reduction: WGS to GS]
\begin{itemize}
  \item[]
  \item \emph{Input}: $\WGS$ and $c\in{\mathbb{N}}$.
  \item \emph{Output}: The GS $G^{*} = \langle \mathcal{V}^{*}, \mathcal{X},
  \mathcal{Y}^{*}, \rho^{e},\rho^{s*}, \varphi \rangle$ where
  \begin{enumerate}
    \item $\mathcal{V}^{*} := \mathcal{X}\cup{\mathcal{Y}^{*}}$.
    \item \label{def:naiveReduction:yDomSet} $\mathcal{Y}^{*} := \mathcal{Y}\cup
    yDom$ where $yDom:=\{y_{0}, \ldots, y_{\lfloor\log(c)\rfloor}\}$ encodes the domain $[0,c]$
    of a new system variable $y$.
    \item \label{def:naiveReduction:newRhoS} 
    For all $s_1, s_2 \in 2^{\mathcal{V}}$ and $c_{1}, c_{2} \in [0,c]$:
    $((s_1, c_1), p(s_2, c_2))\models \rho^{s*}$
    iff $(s_1, p(s_2))\models \rho^{s}$ and
    $c_1 + w^{s}(s_1, p(s_2)) \geq c_2$, where $c_{1}, c_{2}$ are
    used interchangeably with their binary encodings over the variables $yDom$.
  \end{enumerate}
\end{itemize}\label{def:naiveReduction}
\end{defi}

\noindent Intuitively, the reduction of Def.~\ref{def:naiveReduction} constructs
the GS $G^{*}$ that differs from $G^{w}$ in the following attributes:
(1) A \emph{blown-up state space} due to
an additional system controlled variable $y$ that keeps track of the
initial credit or energy level under $c$ in every state;
(2) \emph{additional constraints} to $\rho^{s}$, which ensure
that the non-negative value of the new variable $y$ at each state $s^{*}\in
2^{\mathcal{V}^*}$, during a winning play $\sigma^{*} =
s^{*}_{0}s^{*}_{1}\ldots s^{*} \ldots\in\textsf{Plays}(G^{*})$ in
$G^{*}$, is a lower bound of the energy level
of the prefix
$s^{*}_{0}|_{\mathcal{V}}\ldots s^{*}|_{\mathcal{V}}$ in $G^{w}$.

\begin{thm}[Correctness of Def.~\ref{def:naiveReduction}]\label{thm:reductionCorrectness}
For all upper bounds $c\in{\mathbb{N}}$, WGSs $G^{w}$, initial
credits $c_{0}\in [0,c]$, and $G^{w}$ states $s\in{2^{\mathcal{V}}}$:
the system (resp. environment) wins from $s$ w.r.t. $c$ for $c_{0}$ if and only if the system (resp. environment) wins from
$(s, c_{0})\in 2^{\mathcal{V}^{*}}$ in the GS $G^{*}$
constructed by Def.~\ref{def:naiveReduction} from $c$ and
$G^{w}$.
\end{thm}

The following lemmas show that the semantics of energy $\mu$-calculus
w.r.t. $\GEc$ succinctly represents that of the $\mu$-calculus w.r.t. the GS $G^{*}$,
constructed by Def.~\ref{def:naiveReduction}.
Specifically, Lem.~\ref{lem:sysEngMuCalc} (resp. Lem.~\ref{lem:envEngMuCalc}) relates to all $\sysmu$ ($\envmu$) formulas $\psi$,
and relies on the property of the set $\llbracket{\psi}\rrbracket^{G^{*}}_{\mathcal{E}}$ being
$\leq$-upward ($\leq$-downward) closed w.r.t. the initial credits' values, which are encoded
over the variable $y$.

\begin{lem}\label{lem:sysEngMuCalc}
Let $\psi\in\sysFormulas$ where all Boolean variables are in $\mathcal{V}$.
Let $\eValuation:\mathit{Var}\rightarrow \mathit{EF(c)}$, $\mathcal{E}:
\mathit{Var}\rightarrow (2^{\mathcal{V^{*}}} \rightarrow \{0,1\})$ be valuations such
that for all $X \in \mathit{Var}$, $s \in 2^{\mathcal{V}}$, and $\mathit{val} \in[0,c]$:
$\mathit{val} \preceq \eValuation(X)(s)$ if and only if $(s, \mathit{val}) \in
\mathcal{E}(X)$.
Then, for all $s \in 2^{\mathcal{V}}$ and $\mathit{val} \in[0,c]$:\\
$\mathit{val} \preceq\llbracket{\psi^{\op{E}}}\rrbracket^{\GEc}_{\eValuation}(s)$ if and only if $(s, \mathit{val})\in\llbracket{\psi}\rrbracket^{G^{*}}_{\mathcal{E}}$.
\end{lem}

\begin{lem}\label{lem:envEngMuCalc}
Let $\psi\in\envFormulas$ where all Boolean variables are in $\mathcal{V}$. Let
$\eValuation:\mathit{Var}\rightarrow \mathit{EF(c)}$, $\mathcal{E}:
\mathit{Var}\rightarrow (2^{\mathcal{V^{*}}} \rightarrow \{0,1\})$ be
valuations such that for all $X \in \mathit{Var}$, $s \in 2^{\mathcal{V}}$, and
$\mathit{val} \in[0,c]$: $\mathit{val} \preceq \eValuation(X)(s)$ if and only if $(s, c - \mathit{val}) \in
\mathcal{E}(X)$.
Then, for all $s \in 2^{\mathcal{V}}$ and $\mathit{val} \in[0,c]$:\\
$\mathit{val} \preceq\llbracket{\psi^{\op{E}}}\rrbracket^{\GEc}_{\eValuation}(s)$ if and only if
$(s, c - \mathit{val})\in\llbracket{\psi}\rrbracket^{G^{*}}_{\mathcal{E}}$.
\end{lem} 


That concludes this subsection as it holds that Thm.~\ref{thm:sysEngMuCalcCorrectness} (resp. Thm.~\ref{thm:envEngMuCalcCorrectness})
is a corollary of both Thm.~\ref{thm:reductionCorrectness} and Lem.~\ref{lem:sysEngMuCalc} (resp. Lem.~\ref{lem:envEngMuCalc}).



\section{{\texorpdfstring{Solving Energy Games via Energy $\mu$-Calculus}{Solving Energy Games via Energy mu-Calculus}}}
\label{sec:solvingEnergyGames}

In this section, we show how to use energy $\mu$-calculus to solve $\omega$-regular energy games.
Formally, given a WGS $\WGS$, an upper bound $c\in{\mathbb{N}}\cup \{+\infty\}$,
and a state $s\in 2^{\mathcal{V}}$ in $G^{w}$, we aim to use the results of Sect.~\ref{sec:EngMuCalc} for solving the following problems:

\begin{enumerate}[label=\textbf{P\arabic*},ref=P\arabic*]
  \item
  \label{problem:decision} \emph{The decision problem}:
  Checks whether $s$ wins for the system w.r.t. $c$.
  \item
  \label{problem:optimal} \emph{The minimum credit problem}: Asks
  what is the \emph{minimum} initial credit, $c^{\min}_{0} \not = +\infty$, $c^{\min}_{0}\leq c$,
  for which $s$ wins for the system w.r.t. $c$.
\end{enumerate}

\noindent Sect.~\ref{sec:solvingBEGames} considers these problems when there is a finite upper bound $c\in{\mathbb{N}}$ on the energy levels, while
Sect.~\ref{sec:sufficientbound} treats the unbounded case, namely when $c=+\infty$.
We present an extended version of Sect.~\ref{sec:sufficientbound:WGStoParityEnergyGames} in Appx.~\ref{app:sufficient-bound-proof}.


\subsection{Solving Energy Games with Finite Upper Bounds}\label{sec:solvingBEGames}



%
Let $\varphi$ be an $\omega$-regular condition,
and let $\psi_{\varphi}\in \sysFormulas$ and $\psi_{\neg\varphi}\in \envFormulas$
be closed formulas that match $\varphi$.
Let $G^{w}=(G, w^s)$ be a WGS whose winning condition is $\varphi$, let $c \in \mathbb{N}$
be a finite upper bound, and let $G^{*}$ be the GS constructed by Def.~\ref{def:naiveReduction}.

Eq.~\ref{eq:sysWinningStates} (resp. Eq.~\ref{eq:envWinningStates}) below
follows from Thm.~\ref{thm:reductionCorrectness} and Thm.~\ref{thm:sysEngMuCalcCorrectness}
(resp. Thm.~\ref{thm:envEngMuCalcCorrectness}),
and describes how to compute the set of states that win for the system (resp. environment) player in $G^{w}$ w.r.t. $c$.
\begin{equation}\label{eq:sysWinningStates}
\begin{split}
W_{\sys}(c) &= \{s\in2^{\mathcal{V}} \mid \exists c_{0} \in [0,c] :
(s, c_{0})\in \llbracket{\psi_{\varphi}}\rrbracket^{G^{*}} \}\\  &= \{s\in2^{\mathcal{V}} \mid 
\llbracket{\psi^{\op{E}}_{\varphi}}\rrbracket^{\GEc} (s) \not =  +\infty\}.
\end{split}
\end{equation}

\begin{equation}
\label{eq:envWinningStates}
\begin{split}
W_{\env}(c)  &= \{s\in2^{\mathcal{V}} \mid \forall c_{0} \in [0,c] :
(s, c_{0})\in \llbracket{\psi_{\neg\varphi}}\rrbracket^{G^{*}} \} \\ &= \{s\in2^{\mathcal{V}} \mid 
\llbracket{\psi^{\op{E}}_{\neg\varphi}}\rrbracket^{\GEc} (s)  = 0\}.
\end{split}
\end{equation}

Therefore, given a state $s\in 2^{\mathcal{V}}$, solving the decision problem (\ref{problem:decision}) for $c$ amounts to checking
whether $\llbracket{\psi^{\op{E}}_{\varphi}}\rrbracket^{\GEc}(s)\not = +\infty$. As
we have that $\llbracket{\psi_{\varphi}}\rrbracket^{G^{*}}\cup\llbracket{\psi_{\neg \varphi}}\rrbracket^{G^{*}}
= 2^{\mathcal{V}^{*}}$ due to determinacy of $\omega$-regular games, it follows from
Thm.~\ref{thm:reductionCorrectness} that $\omega$-regular energy games are determined w.r.t. finite upper bounds, i.e.,
$W_{\sys}(c) \cup W_{\env}(c) = 2^{\mathcal{V}}$. Thus, alternatively, we
can solve \ref{problem:decision} by checking whether $\llbracket{\psi^{\op{E}}_{\neg\varphi}}\rrbracket^{\GEc}(s) \not = 0$.

As a side note, following a reasoning similar to the above, determinacy also holds w.r.t. the bound of $+\infty$. In sketch, this claim is argued as follows.
We define a reduction which is like Def.~\ref{def:naiveReduction} but does not use a finite upper bound.
That is, the modified reduction adds a new system controlled variable $y$ whose domain is $\mathbb{N}$ (rather than $[0,c]$);
thus, it constructs a GS $G^{*}$ whose state space is infinite.
The variable $y$ now keeps track of the unbounded energy level in every state.
Then, the determinacy of $G^{*}$ implies the determinacy of the WGS $G^{w}$ w.r.t. $+\infty$.
Consequently, we conclude that \emph{$\omega$-regular energy games are determined}.

A corollary of Thm.~\ref{thm:sysEngMuCalcCorrectness} is that
we can solve the minimum credit problem (\ref{problem:optimal}) by simply returning $\llbracket{\psi^{\op{E}}_{\varphi}}\rrbracket^{\GEc}(s)$.
However, the determinacy of $\omega$-regular energy games together with Thm.~\ref{thm:sysEngMuCalcCorrectness} and Thm.~\ref{thm:envEngMuCalcCorrectness},
imply that we can also solve \ref{problem:optimal} by computing $\llbracket{\psi^{\op{E}}_{\neg\varphi}}\rrbracket^{\GEc}$ and return ${\sim}{\llbracket{\psi^{\op{E}}_{\neg\varphi}}\rrbracket^{\GEc}}$, as follows:

\begin{enumerate}
\item If $\llbracket{\psi^{\op{E}}_{\neg\varphi}}\rrbracket^{\GEc}(s) = 0$, return ``$s$ does not win
for the system w.r.t. $c$'' (i.e., return $+\infty$).
\item If $\llbracket{\psi^{\op{E}}_{\neg\varphi}}\rrbracket^{\GEc}(s) = +\infty$, return 0.
\item Otherwise, return $c + 1 - \llbracket{\psi^{\op{E}}_{\neg\varphi}}\rrbracket^{\GEc}(s)$.
\end{enumerate}

\noindent Finally, we stress that the semantics of an energy $\mu$-calculus formula, which we have inductively defined in Sect.~\ref{sec:EngMuCalcSyntaxSemantics},
immediately prescribes an algorithm to compute the energy function characterized by this formula (cf. Eq.~\ref{eq:buchiExample:energyMuCalcFormula} and Alg.~\ref{alg:buchiEnergy}).
Therefore, seeing $\llbracket{\psi^{\op{E}}_{\varphi}}\rrbracket^{\GEc}$ and $\llbracket{\psi^{\op{E}}_{\neg\varphi}}\rrbracket^{\GEc}$
each as a symbolic algorithm, the above, in fact, describes \emph{algorithms} to solve problems \ref{problem:decision} and \ref{problem:optimal}.
A straightforward implementation of $\llbracket{\psi^{\op{E}}_{\varphi}}\rrbracket^{\GEc}$ or $\llbracket{\psi^{\op{E}}_{\neg\varphi}}\rrbracket^{\GEc}$
gives an algorithm that computes the desired energy function in $O((|2^\mathcal V|(c+1))^q)$ symbolic steps, where $q$ is the largest number of nested fixed-point operators in the energy $\mu$-calculus formula.
Nevertheless, using the well-known techniques proposed in~\cite{BrowneCJLM97} and~\cite{EmersonL86}, we can reduce this time complexity to $O((|2^\mathcal V|(c+1))^{\lfloor d/2 \rfloor+1})$
symbolic steps, where $d$ is the alternation depth of the formula (hence $d\leq q$).



\subsection{A Sufficient Upper Bound}\label{sec:sufficientbound}

We have shown in Sect.~\ref{sec:solvingBEGames} how energy $\mu$-calculus can be
used to solve problems \ref{problem:decision} and \ref{problem:optimal} when
there is a finite upper bound $c \in \mathbb{N}$ on the energy levels.
However, in some cases, such a finite bound is unknown a priori, so one may
wish to find a complete upper bound, i.e., a sufficiently large bound whose increase
would not introduce additional winning states for the system player.

In this section, we show how to compute such a bound.
Moreover, based on the main result of this section (see
Thm.~\ref{Thm:a-sufficient-bound}), \emph{we now solve problems \ref{problem:decision}
and \ref{problem:optimal} also for the case where $c=+\infty$}, which we have left
unresolved in the previous section.
The complete bound we present depends on the size of the game's state space
($N$), the maximal absolute weight of a transition in the game ($K$), and the
length ($m$) and the alternation depth ($d$)~\cite{EmersonL86,Niwinski86} of the $\mu$-calculus formula that matches the winning condition. 

\begin{thm}\label{Thm:a-sufficient-bound}
Let $\WGS$ be a WGS, $N=|2^{\mathcal{V}}|$, and let $K$ 
  be the maximal transition weight in $G^w$, in absolute value. Take $\psi\in \sysFormulas$,
  a closed $\sysmu$ formula that matches $\varphi$, and let $m$ be its length
  and $d$ its alternation depth. Then, if the system wins from a state $s$ w.r.t. 
  $+\infty$ for an initial credit $c_0$, then it also wins from $s$  
 w.r.t. $(d+1)((N^2+N)m-1)K$ 
for an initial credit 
$\min\{c_0,((N^2+N)m-1)K\}$.
\end{thm}

We devote the remainder of this section to proving
Thm.~\ref{Thm:a-sufficient-bound}. The crux of the proof is to reduce the game
into an \emph{energy parity game}~\cite{ChatterjeeD12,EmersonJ91,mostowski1985}, thus
we turn to discuss these games extensively. Note that the reduction to energy parity games is presented here only as part of the proof.
It is \emph{not} part of our algorithm for solving $\omega$-regular energy games. We show connections between solving
the bounded and unbounded energy parity objective that allow us to prove
Thm.~\ref{Thm:a-sufficient-bound}.

\begin{defi}[Energy parity game]\label{def:energyParityGame} An \emph{energy parity game} is a tuple ${G} = \langle({V}={V_0\cup V_1},E),\prio,w\rangle$ that consists
of the following components:

\begin{enumerate}
  \item
  A directed graph $(V=V_0\cup V_1,E)$ where $V_0,V_1$ partition
  $V$ into $\text{player}_0$ states and $\text{player}_1$ states, respectively.
  \item
  A priority function $\prio:V\into \mathbb{N}$.
  \item
  A weight function $w:E\into \mathbb{Z}$.
\end{enumerate}

\end{defi}

Let $G$ be an energy parity game as in Def.~\ref{def:energyParityGame}.
If the weight function is omitted, $G$ is simply said to be a \emph{parity game}.
Plays and strategies in $G$ are defined in a way similar to those in a WGS (see Sect.~\ref{sec:preliminaries}). The only difference is that in $G$, the players do not necessarily take steps in an
alternating manner; $\text{player}_i$ chooses the next successor whenever the
play reaches a $V_i$-state. Thus, a play is a path that is either infinite, or
ends in a deadlock, i.e., a state with no outgoing edges.

Given a play, $\sigma$, $\mathit{inf}(\sigma)\subseteq V$ is the set of all
states that appear infinitely often in $\sigma$. For $c\in
\mathbb{N}\cup\{+\infty\}$, $c_0\in[0,c]\cap \mathbb N$, and a path $\sigma$
in $G$ of length at least $k+1$, we use the notion of $\textsf{EL}_{c}(G,c_0,
\sigma[0\ldots k])$, as was defined in Sect.~\ref{sec:combinedEngObj}.
 In words, $\textsf{EL}_{c}(G,c_0, \sigma[0\ldots k])$ is the energy level accumulated
 so far according to $w$, when the play starts with the initial credit $c_0$, and $c$
 is the energy upper bound.
 
 A play $\sigma$ wins for $\text{player}_0$ w.r.t. $c$ for an initial credit
 $c_0$ if the following holds: (1) for every finite prefix of $\sigma$, $\sigma'$, 
$\textsf{EL}_{c}(G,c_0, \sigma')\geq 0$; (2) if $\sigma$ is infinite,
then $\min\{\prio(v):v\in \mathit{inf}(\sigma)\}$ is even, and
(3) if $\sigma$ is finite, it ends in a deadlock for $\text{player}_1$.
Otherwise, $\sigma$ wins for $\text{player}_1$. We refer to (1) 
as the \emph{energy objective}, while requirements (2) and (3) form the 
\emph{parity objective}. Hence, in the case of a parity game, a play wins 
for $\text{player}_0$ if the parity objective is achieved. 

As for WGSs, we adapt the notion of $W_\alpha(c)$ to denote the winning
region of $\text{player}_\alpha$ w.r.t. the upper bound $c\in
\mathbb{N}\cup\{+\infty\}$ in a given energy parity game. When necessary, we may also write $W^G_\alpha(c)$ to
clarify that we relate to the winning region of $\text{player}_\alpha$ in the
energy parity game, $G$. If $G$ is a parity game, the energy upper bound is
neglected, and thus we just write $W_\alpha$ or $W^G_\alpha$.


The following lemma, which establishes the first step towards proving Thm.~\ref{Thm:a-sufficient-bound}, is an adaptation of Lem.~6 in~\cite{ChatterjeeD12}:

\begin{lem}\label{lem:lemma-6-revised}
Let $G$ be an energy parity game (as defined in Def.~\ref{def:energyParityGame}) with $n$ states, $d$ different priorities 
and maximal absolute value of the weights, $K$. 
If $\text{player}_0$ has a winning strategy from a state $s$ w.r.t. $+\infty$ 
for an initial credit $c_0$, then she has a winning strategy w.r.t. $d(n-1)K$ 
for an initial credit $\min\{c_0,(n-1)K\}$.
\end{lem}

Consider an energy parity game $G$, as in Lem.~\ref{lem:lemma-6-revised}, and let $s$ be a state of $G$ that wins for $\text{player}_{0}$ w.r.t. $+\infty$ for an initial credit $c_0$.
While Lem.~6 of~\cite{ChatterjeeD12} shows that $\text{player}_{0}$ wins from $s$ w.r.t. $+\infty$ for the initial credit $(n-1)K$ with a strategy which has a memory of size $dnK$,
Lem.~\ref{lem:lemma-6-revised} shows that $\text{player}_0$ wins from $s$ w.r.t. the finite upper bound $d(n-1)K$
without the need to increase her initial credit, $c_{0}$. Both lemmas are proved by induction. However, in contrast to~\cite{ChatterjeeD12}, which proves Lem.~6 by induction
on $d$, we prove Lem.~\ref{lem:lemma-6-revised} by induction on ${n}+{d}$. This allows us to apply the induction hypothesis in more cases and, consequently, avoid the 
use of recursion, as opposed to~\cite{ChatterjeeD12}.
We provide the full proof for Lem.~\ref{lem:lemma-6-revised} in Appx.~\ref{app:lemma-6-revised-proof}.

Also, as a side note, Lem.~\ref{lem:lemma-6-revised} implies the following corollary which slightly improves the first result listed in~\cite{ChatterjeeD12}. Moreover, this corollary establishes a link between
the upper bound on the energy level accumulation and the upper bound on the strategy's memory size.

\begin{cor}
  Let $G$ be an energy parity game with $n$ states, $d$ different priorities 
  and maximal absolute value of the weights, $K$. If $\text{player}_{0}$ wins from a state $s$ w.r.t. $+\infty$ for an initial credit $c_{0}$ in $G$, then she has a strategy that wins for the initial credit $c_0$ and
 has a memory of size $d(n-1)K+1$.
 \label{cor:parityEnergyMemory}
\end{cor}


We obtain Cor.~\ref{cor:parityEnergyMemory} by applying Lem.~\ref{lem:lemma-6-revised} to $G$,
obtaining that $\text{player}_0$ wins from $s$ w.r.t. $d(n-1)K$,
and observing that if $\text{player}_0$ wins from $s$ w.r.t. $c\in \mathbb N$,
then she has a winning strategy from $s$ w.r.t. $+\infty$ with a memory of size $c+1$. The proof sketch for the latter claim is as follows. We construct from $G$ and $c$ a new parity game $G^{p}$. 
The states of $G^{p}$ have the form $(s,c'_{0})$ where $s$ is a state of $G$ and $c'_{0}\in[0,c] \cup \{+\infty\}$ is the accumulated energy level under $c$. States of the form $(s, +\infty)$ are deadlocks for $\text{player}_{0}$
and correspond to violation of the energy objective. 
The edges of $G^{p}$ are taken from $G$ and update the energy component accordingly.
Then, a memoryless strategy in $G^{p}$, which exists due to memoryless determinacy of parity games~\cite{EmersonJ91,Zielonka98},
can be lifted to a strategy in $G$ with a memory of size $(c+1)$ to keep track of the energy level under $c$.

The second step towards proving Thm.~\ref{Thm:a-sufficient-bound}, involves showing that if the system player can win from a state of a WGS w.r.t. $+\infty$, then it can also win from that state w.r.t. some finite upper bound.
This is formally stated by the next lemma.





\begin{lem}\label{lem:from-infinite-to-finite}
  Let $\WGS$ be a WGS and assume 
  that the system player has a winning strategy from a state $s\in 2^\mathcal{V}$ w.r.t. 
  $+\infty$ for an initial credit $c_0$. Then, for some finite upper bound $c \in \mathbb{N}$, the system 
  has a winning strategy from $s$ w.r.t. $c$ for an initial credit
  $\min\{c_{0},c\}$.  
\end{lem}


\begin{proof}
 To prove this claim, we use the notion of a \emph{deterministic parity automaton}~\cite{Buchi60,mostowski1985}.

\begin{defi}[Deterministic parity automaton]
A deterministic parity automaton is a tuple ${\mathcal{A}}={\langle Q,\Sigma,\delta,q_0,\prioA\rangle}$ where $Q$ is a finite set of states, $\Sigma$ is an alphabet, $\delta:Q\times\Sigma\rightarrow Q$ is the transition function, $q_0\in Q$ is the initial state, and $\prioA:Q\rightarrow \mathbb{N}$ is the priority function.  
\end{defi}

An $\omega$-word $\sigma_0\sigma_1\sigma_2\cdots\in \Sigma^\omega$ is accepted by a deterministic parity automaton $\mathcal{A}=\langle Q,\Sigma,\delta,q_0,\prioA\rangle$, if there is an infinite sequence of states $q'_0,q'_1,q'_2,\dots$ such that $q'_0=q_0$, $\delta(q'_j,\sigma_i)=q'_{i+1}$, and $\underset{i\rightarrow +\infty}{\lim}\min\{\prioA(q_l):l\geq i\}$ is even. In words, using its transition function, the automaton $\mathcal{A}$ reads an $\omega$-word, starting from $q_0$, and if the minimal priority traversed infinitely often is even, $\mathcal{A}$ accepts the word.
The set of all $\omega$-words accepted by the automaton is denoted by $L(\mathcal{A})$.
It is known that for every $\omega$-regular language $L$ there exists a deterministic parity automaton $\mathcal{A}$ with $L=L(\mathcal{A})$~\cite{2001automata,piterman2006,safra88}.

Let $\mathcal{A}=\langle Q,2^\mathcal{V},\delta,q_0,\prioA\rangle$ be a deterministic parity automaton with $L(\mathcal{A})=L(\varphi)$, where $L(\varphi)$ is the set of all $\omega$-words $\sigma \in (2^{\mathcal V})^{\omega}$ for which the $\omega$-regular winning condition of $G^{w}$, $\varphi$, holds. 

We define an energy parity game $G_\mathcal{A}=\langle (V,E),\mathit{prio},w \rangle$, as follows: 
  
  \begin{enumerate}
    \item $\text{player}_1$ states are all states of the form $(s,q)$ 
    where $s\in 2^\mathcal{V}$ and $q$ is a state of $\mathcal{A}$.
    
    \item $\text{player}_0$ states are all states of the form $(s,u,q)$ where 
    $s\in 2^{\mathcal V}$, $q$ is a state of $\mathcal{A}$, and $u\in 2^\mathcal{X}$.
    
    \item There is a transition from a state $(s,q)$ to a state $(s,u,q)$ if 
    $(s,p(u))\models \rho^e$. The weight of such a transition is $0$.
    
    \item There is a transition from a state $(s,u,q)$ to a state $(t,q')$ if 
    $u=t|_\mathcal{X}$, $(s,p(t))\models \rho^s$, and $q'=\delta(q,t)$. 
    The weight of such a transition is $w^s(s,p(t))$.
    
    
    \item $\prio(s,u,q)=\prio(s,q)=\prioA(q)$.
  \end{enumerate}
  
  \noindent It is not difficult to see that for every state $t\in 2^\mathcal{V}$, upper bound 
  $d\in\mathbb{N}\cup \{+\infty\}$, and initial credit $d_0\in \mathbb{N}$, $d_0\leq d$, the
  system wins in $G^w$ from $t$ w.r.t. $d$ for $d_0$, 
  if and only if $\text{player}_0$ 
  wins in $G_{\mathcal A}$ from $(t,q_t=\delta(q_0,t))$ w.r.t. $d$ for $d_0$. 
  
  Now, let $s \in 2^{\mathcal V}$ be a state in $G^{w}$ from which the system wins w.r.t. $+\infty$ for an initial credit
  $c_0\in \mathbb{N}$. Therefore, for $q_s=\delta(q_0,s)$, $\text{player}_0$ wins in $G_{\mathcal A}$ from $(s,q_s)$ w.r.t. 
  $+\infty$ for the initial credit $c_0$. By Lem.~\ref{lem:lemma-6-revised}, 
  for some finite upper bound $c\in\mathbb{N}$, $\text{player}_0$ wins from $(s,q_s)$ w.r.t. $c$ for the initial credit $\min\{c_{0},c\}$. Hence, the system wins in $G^w$ from $s$ w.r.t. $c$ for the initial credit $\min\{c_{0},c\}$, as required.
\end{proof}

\subsubsection{Reducing Weighted Game Structures to Energy Parity Games}\label{sec:sufficientbound:WGStoParityEnergyGames}

So far, we have shown in Lem.~\ref{lem:lemma-6-revised} that if $\text{player}_0$ can win an energy
parity game w.r.t. $+\infty$, then she can also win w.r.t. some finite upper bound
that depends on the size of the game. We have concluded in Lem.~\ref{lem:from-infinite-to-finite} that the same holds for any WGS,
but we have not yet achieved the desired upper bound, which is specified in
Thm.~\ref{Thm:a-sufficient-bound}.

In the following, we prove the sufficiency of
the upper bound $(d+1)((N^2+N)m-1)K$ for winning, in case winning is
possible. The idea is to reduce $\omega$-regular energy games (WGSs) to energy parity games without using the explicit construction from
the proof of Lem.~\ref{lem:from-infinite-to-finite}.
Instead, we provide a construction that uses the energy $\mu$-calculus formula that solves the game.
That is mostly useful in cases where the $\mu$-calculus formula is relatively small,
e.g., reachability, safety, B\"{u}chi, co-B\"{u}chi, GR(1)~\cite{BJP+12}, etc.
We present the construction guidelines.
For the full details, we refer the reader to Appx.~\ref{app:sufficient-bound-proof}.

Consider a WGS, $G^w$, where $K$ is the maximal absolute value of the weights in
$G^w$. Let $\psi\in\sysFormulas$ be a closed $\sysmu$ formula that
matches the winning condition of $G^w$, $\varphi$, and let $m$ be its length and
$d$ its alternation depth. For a natural number $c\geq K$, we construct a
parity game in several steps as elaborated below. We remark that the actual
construction (see Appx.~\ref{app:sufficient-bound-proof}) is, in some places, slightly different than the one described here. That is
because we choose to omit some technical details which we believe to only distract and conceal the essence of the construction.

\begin{description}
\item[{\bfseries Step 1}] Let $G_c$ be the graph defined in Def.~\ref{def:naiveReduction} (appears in Def.~\ref{def:naiveReduction} as $G^{*}$).
Hence, the states of $G_c$ are of the form: $(s,c_0)$ where $s$ 
is a state of $G^{w}$ and $c_0 \in [0,c]$. Recall that by Lem.~\ref{lem:sysEngMuCalc},
${c_0 \geq \llbracket \psi^\op{E}\rrbracket^{\GEc}(s)}\text{ iff }{((s,c_0)\in\llbracket \psi \rrbracket^{G_{c}})}$.

\item[{\bfseries Step 2}] We apply the seminal \emph{model checking game construction}~\cite{EmersonJ91} to obtain a parity game $G_c\times \psi$, which has at most $d+1$ different priorities. 
The states of $G_c\times \psi$ are of the form $((s,c_0),\xi)$ where $\xi$ is a
sub-formula of $\psi$. By~\cite{EmersonJ91}, $(s,c_0)\in \llbracket \psi \rrbracket^{G_c}$ iff $((s,c_0),\psi)\in W^{G_c\times\psi}_0$.

\item[{\bfseries Step 3}] The next step is to add a weight function $w$ to $G_{c}\times \psi$, namely to transform $G_c\times\psi$ into an energy parity game (as defined in Def.~\ref{def:energyParityGame}). Some of the transitions in $G_c\times \psi$ simulate transitions of $G_c$,
which correspond to transitions of $G^w$. The weight of such a transition $T = ((s_1,c_1),\xi_1),((s_2,c_2),\xi_2)$ is
inherited from the transition of $G^w$ that $T$ simulates, i.e., $T$ is assigned the weight $w^s(s_1,p(s_2))$. The weight of all other transitions is $0$. However, the construction of $G_c$ ensures that every transition $T = ((s_1,c_1),\xi_1),((s_2,c_2),\xi_2)$ satisfies that ${c_1 + w(T)} \geq {c_2}$
(cf. the definition of $\rho^{s*}$ in Def.~\ref{def:naiveReduction}). Hence, in any play that starts from $((s,c_0), \psi)$ with an initial credit $c_0$, the energy level always remains non-negative. Consequently, the additional energy objective is merely artificial and does not prevent $\text{player}_0$ from winning; the result is that a state $s$ wins for the system in $G^w$
w.r.t. $c$ for an initial credit $c_0$, iff $((s,c_0),\psi)\in W^{G_c\times
  \psi}_0(c)$.



\item[{\bfseries Step 4}] The final step is to eliminate the energy component from the
states of $G_{c}\times \psi$. States of the form $((s,c_0),\xi)$ are replaced with
a single state $(s,\xi)$. Thus, each state of the obtained energy parity game,
$\faktor{G_c\times \psi}{c}$, matches a set of states in $G_c\times \psi$. A path
in $\faktor{G_c \times \psi}{c}$ can be lifted to a path in $G_c\times \psi$.
Thus, we have that $(s,\psi)$ wins 
for $\text{player}_0$ in $\faktor{G_c\times \psi}{c}$ w.r.t. $c$ for an initial credit $c_0$ iff $\text{player}_0$ wins from $((s,c_0),\psi)$ in $G_c\times \psi$ w.r.t. $c$ for an initial credit $c_0$ iff the system player wins from $s$ in $G^w$ w.r.t. $c$ for $c_0$.
  
\end{description}


\noindent The key idea behind this construction is that the upper bound $c$ does \emph{not} play a
role in the resulting game, $\faktor{G_c\times\psi}{c}$. In fact, we get that
for any two finite upper bounds, $c,c'\geq K$, $\faktor{G_c\times\psi}{c}
=\faktor{G_{c'}\times\psi}{c'}$, so we may denote this graph by a single name,
say $\tilde{G}$. By Lem.~\ref{lem:lemma-6-revised}, if a state of $\tilde{G}$
wins for $\text{player}_0$ w.r.t. $+\infty$, it also wins w.r.t. $b=
(d+1)((N^2+N)m-1)K$, as $(N^2+N)m$ is the number of states of $\tilde{G}$. Therefore, if a state $s$ wins for the system in $G^w$ w.r.t. some
finite upper bound $c$, then $s$ also wins w.r.t. $b$. This consequence, together with Lem.~\ref{lem:from-infinite-to-finite}, completes the proof of
Thm.~\ref{Thm:a-sufficient-bound}.

\emph{Notice that this result establishes energy $\mu$-calculus algorithms for problems \ref{problem:decision} and
\ref{problem:optimal} when the bound is $+\infty$:}

\begin{itemize}
  \item Checking if $\llbracket\psi^\op{E}\rrbracket^{G^{w}(b)}(s)\neq +\infty$ solves the decision problem (\ref{problem:decision}).
\item Returning $\llbracket \psi^\op{E}\rrbracket^{G^{w}(b)}(s)$ solves the minimum credit problem (\ref{problem:optimal}).
\end{itemize}

\noindent The sufficient bound $b=(d+1)((N^2+N)m-1)K$, which we have just obtained, applies to WGSs with \emph{any} $\omega$-regular winning conditions.
Nevertheless, in some cases, this bound is not tight, as we demonstrate below.

Consider a $\sysmu$ formula $\psi_B$ that matches a B\"{u}chi winning condition $\varphi_B$ (cf. Eq.~\ref{eq:buchiExample:muCalcFormula}), and let $m_B$ be its length.
As the alternation depth of $\psi_B$ is 2, Thm.~\ref{Thm:a-sufficient-bound} implies the sufficiency of the bound $b_{B}=3((N^2+N)m_B-1)K$.
Interestingly, however, in this special case, we argue that the bound $b_{B}$ is not tight, and it can be replaced with a \emph{lower} one, specifically with
$b^{\mathit{low}}_{B} = 2({N^2+N-1})K$.

To obtain $b^{\mathit{low}}_{B}$, we reduce a WGS $G^{w}$ whose winning condition is $\varphi_B$ to an energy parity game $G^{ep}$ with two priorities
and at most $N^2 + N$ states; then, we invoke Lem.~\ref{lem:lemma-6-revised} on $G^{ep}$.

The crux of this reduction is that it constructs $G^{ep}$
without using $\psi_B$, as opposed to the above construction. 
The reduction sees the B\"{u}chi condition as a parity condition with two priorities (w.l.o.g. $0$ and $1$),
and constructs $G^{ep}$, which is simply an explicit representation of the symbolic WGS $G^w$, as follows:   

\begin{itemize}
     \item $\text{player}_1$ states are all states $s \in 2^{\mathcal V}$.
    \item $\text{player}_0$ states are all pairs $(s,u)$ where $s \in 2^{\mathcal V}$ is a $G^{w}$-state and $u \in 2^{\mathcal X}$ is an assignment to the input variables.
  \item There is a transition from $s$ to $(s,u)$ if $(s,p(u))\models \rho^e$; such a transition corresponds to an environment's step and its weight is 0.
  \item There is a transition from $(s,u)$ to $t \in 2^{\mathcal V}$ if $u=t|_\mathcal{X}$ and $(s,p(t))\models \rho^s$; such a transition corresponds
  to a system's step and its weight is $w^s(s,p(t))$.
  \item The priorities of all the states $s \in 2^{\mathcal V}$ in $G^{ep}$ remain the same as those in $G^{w}$, while every state $(s,u)$ in $G^{ep}$ is assigned the same priority as $s$.
\end{itemize}

\noindent Finally, we observe that, in fact, the construction of $G^{ep}$ does not only apply to B\"{u}chi winning conditions, but to
parity conditions. That is, this construction
reduces a parity WGS with $d$ priorities to an energy parity game $G^{ep}$ with $d$ priorities and at most $N^2+N$ states.
Therefore, we also conclude the sufficiency of the tighter bound $d(N^2+N-1)K$ for parity WGSs with $d$ different priorities.

\section{Conclusion}
\label{sec:conclusion}

We have introduced energy $\mu$-calculus, a multi-valued extension of the game $\mu$-calculus~\cite{EmersonJ91} over symbolic game structures~\cite{BJP+12} that serves as
a symbolic framework for solving $\omega$-regular energy games.
Existing, well-known game $\mu$-calculus formulas $\psi$ that solve $\omega$-regular games
can be seamlessly reused as energy $\mu$-calculus formulas $\psi^\op{E}$ to solve corresponding energy augmented games (see Thm.~\ref{thm:sysEngMuCalcCorrectness} and Thm.~\ref{thm:envEngMuCalcCorrectness}).
The semantics of $\psi^\op{E}$ immediately prescribes a symbolic algorithm to solve the underlying $\omega$-regular energy games (cf. Alg.~\ref{alg:buchiEnergy}).

The semantics of energy $\mu$-calculus is defined w.r.t. finite upper bounds. Nevertheless,
we have shown that energy $\mu$-calculus solves both the decision and the minimum credit problems (i.e., problems~\ref{problem:decision} and~\ref{problem:optimal} in Sect.~\ref{sec:solvingEnergyGames}), also with an unbounded energy level accumulation.
We have obtained this result by showing that every $\omega$-regular winning condition admits a sufficiently large upper bound
under which the bounded energy level accumulation coincides with the unbounded one. Moreover,
importantly, although it is finite, the sufficient bound still enables the system player to win without increasing the initial credit.

We have introduced a sufficient bound that depends on the size of the state space, the maximal absolute weight, and the length and the alternation depth of the game $\mu$-calculus formula that solves the $\omega$-regular game (see Thm.~\ref{Thm:a-sufficient-bound}).
To prove this bound, we have reduced $\omega$-regular energy games over symbolic weighted game structures, to energy parity games~\cite{ChatterjeeD12}.
This reduction, which applies a construction that uses the $\mu$-calculus formula that solves the game,
establishes a connection to the sufficient bound that we have obtained for energy parity games (see Lem.~\ref{lem:lemma-6-revised}).


\subsection*{Future Work} The game $\mu$-calculus has not only been used to compute the sets of winning
states, but to also synthesize winning strategies; see,
e.g.,~\cite{BJP+12,BruseFL14,2001automata,KonighoferHB13}. 
Thus, in addition to solving the decision and the minimum credit problems, we believe that
energy $\mu$-calculus can augment $\mu$-calculus-based strategy synthesis with
energy. That is, we conjecture that finite memory winning strategies may be extracted from the intermediate energy functions of the fixed-point iterations.




\section*{Acknowledgment}
 This project has received funding from the European Research Council (ERC)
under the European Union's Horizon 2020 research and innovation programme (grant
agreement No 638049, SYNTECH).

\bibliographystyle{alpha}
\bibliography{doc}

\appendix

\section{\texorpdfstring{Energy $\mu$-Calculus Over Weighted Game Structures}{Energy mu-Calculus Over Weighted Game Structures}}\label{app:energyMuCalc}

\subsection{\texorpdfstring{$\envCapEnergy$ $\mu$-Calculus: Full Definition}{env-Energy mu-Calculus: Full Definition}}\label{app:energyMuCalc:extDefinitions}

In this appendix, we provide the full definition of the semantics of $\envEnergy$ $\mu$-calculus, as a supplementary to Sect.~\ref{sec:EngMuCalcEnvSemantics}.

\begin{defi}[Dual energy controllable predecessor operator]\label{def:dualECpre}
	For all WGSs $\langle G, w^s\rangle$, upper bounds $c\in \mathbb{N}$, energy functions $f\in \mathit{EF(c)}$,
	and states $s \in 2^{\mathcal{V}}$:
	 \begin{align*}
  &{\ECpreEnv(f)(s)} := {\min\limits_{s_{\mathcal{X}}\in{2^\mathcal{X}}}
    \lbrack\max\limits_{s_{\mathcal{Y}}\in{2^\mathcal{Y}}}\overline{\textsf{EC}}_c((s, p(s_{\mathcal{X}},s_{\mathcal{Y}})),
    f(s_{\mathcal{X}},s_{\mathcal{Y}}))\rbrack}\\&
  \text{where ${\overline{\textsf{EC}}_c} : {{{2^{\mathcal{V} \cup \mathcal{V'}}}\times{{\mathit{E(c)}}}}\rightarrow{\mathit{E(c)}}}$
  and for all $s \in 2^{\mathcal V}$, $s' \in 2^{\mathcal V'}$, and $e \in \mathit{E(c)}$,}\\&
  {\overline{\textsf{EC}}_c((s,s'),e)} =
  \begin{cases}
  +\infty,\ &\mbox{if $(s,s')\not\models{\rho^e}$}\\
  0,\ &\mbox{if $e=0$ or $(s,s')\models{\rho^e\wedge{\neg\rho^s}}$}\\
  0,\ &\mbox{if $e=+\infty$ and $w^{s}(s,s')+c < 0$}\\
  +\infty,\ &\mbox{if $e=+\infty$ and $w^{s}(s,s')\geq 0$}\\
  c+1+w^{s}(s,s'),\ &\mbox{if $e=+\infty$} \\
  0,\ &\mbox{if $e + w^{s}(s,s') \leq 0$}\\
  +\infty, &\mbox{if $e + w^{s}(s,s') > c$}\\
  e + w^{s}(s,s'),\ &\mbox{otherwise}
  \end{cases}&
   \end{align*}
\end{defi}

\begin{defi}[$\envEnergy$ $\mu$-calculus: semantics]\label{def:envEngMuCalcSemantics}
	The semantics $\llbracket \psi\rrbracket^{\GEc}_{\eValuation}$ of
	$\psi\in{\envEFormulas}$ w.r.t. a finite WGS
	$\WGS$, a finite upper bound $c \in \mathbb{N}$,
	and a valuation $\eValuation: \mathit{Var} \rightarrow \mathit{EF(c)}$ over $\mathit{EF(c)}$, is inductively defined for all states
	$s\in{2^{\mathcal{V}}}$, as follows:
	
	\begin{itemize}
		\setlength\itemsep{0.5em}
		\item For $v\in\mathcal{V}$, $\llbracket{v}\rrbracket^{\GEc}_{\eValuation}(s) =
		\begin{cases}
		0, & \text{ if } s\vDash{v} \\
		+\infty, & \text{ if } s\nvDash{v}
		\end{cases};
		\llbracket{\neg{v}}\rrbracket^{\GEc}_{\eValuation}(s) =
		\begin{cases}
		+\infty, & \text{ if } s\vDash{v} \\
		0, & \text{ if } s\nvDash{v}
		\end{cases}
		$
		\item For $X\in{\mathit{Var}}$, $\llbracket{X}\rrbracket^{\GEc}_{\eValuation}(s) = \eValuation(X)(s)$.
		\item $\llbracket{\phi_1\vee{\phi_2}}\rrbracket^{\GEc}_{\eValuation}(s) =
		\min(\llbracket{\phi_1}\rrbracket^{\GEc}_{\eValuation},
		\llbracket{\phi_2}\rrbracket^{\GEc}_{\eValuation})(s)$.
		\item $\llbracket{\phi_1\wedge{\phi_2}}\rrbracket^{\GEc}_{\eValuation}(s) =
		\max(\llbracket{\phi_1}\rrbracket^{\GEc}_{\eValuation},
		\llbracket{\phi_2}\rrbracket^{\GEc}_{\eValuation})(s)$.
		\item $\llbracket{\circlebox_{\op{E}}\phi}\rrbracket^{\GEc}_{\eValuation}(s) =
		\ECpreEnv(\llbracket\phi\rrbracket^{\GEc}_{\eValuation})(s)$.
		\item $\llbracket\twolinescurly{\mu}{\nu} X\phi\rrbracket^{\GEc}_{\eValuation}(s) =
		\twolinescurly{\mathit{lfp}}{\mathit{gfp}} (\lambda{f}.\llbracket\phi\rrbracket^{\GEc}_{\eValuation[X\mapsto{f}]})(s) =
		\twolinescurly{\min\limits_{i}}{\max\limits_{i}}\lbrack{h_i}\rbrack(s)$,\\
		where $\twolinescurly{h_0=f_{+\infty}}{h_0=f_{0}}$ and
		$h_{i+1}=
		\llbracket{\phi}\rrbracket^{\GEc}_{\eValuation[X\mapsto{h_i}]}$.
	\end{itemize}
\end{defi}

\noindent The next lemma proves the correctness of the operator defined in Def.~\ref{def:dualECpre}.

\begin{lem}[Correctness of Def.~\ref{def:dualECpre}]\label{lemma:dualECpre}
	The operator $\ECpreEnv:\mathit{EF(c)}\rightarrow{\mathit{EF(c)}}$ from Def.~\ref{def:dualECpre}
	is the dual of $\ECpre$ from Def.~\ref{def:ECpre}.
	That is, for all $f \in \mathit{EF(c)}$, ${\sim \ECpre(f)}={\ECpreEnv({\sim f})}$ where $\sim$ is the pointwise negation operation
	defined in Eq.~\ref{eq:negation}.
\end{lem}
\begin{proof}\hfill

	\begin{enumerate}
		\item\label{proof:dualECprePart1}
		First, we show that for all WGSs $G^{w} = \langle G , w^{s} \rangle$, upper bounds $c\in \mathbb{N}$, $s \in 2^{\mathcal{V}}$, $s_{\mathcal{X}}\in{2^\mathcal{X}}$,
		$s_{\mathcal{Y}}\in{2^\mathcal{Y}}$, and $f \in \mathit{EF(c)}$:
		\[{\textsf{EC}_c((s,(s_{\mathcal{X'}},s_{\mathcal{Y'}})),
			f(s_{\mathcal{X}},s_{\mathcal{Y}}))}=
		{\sim\overline{\textsf{EC}}_c((s,(s_{\mathcal{X'}},s_{\mathcal{Y'}})),
			\sim f(s_{\mathcal{X}},s_{\mathcal{Y}}))},\]
		where $s_{\mathcal{X'}}$ and $s_{\mathcal{Y'}}$ denote
		$p(s_{\mathcal{X}})$ and $p(s_{\mathcal{Y}})$, respectively. Also, let
		$e:=f(s_{\mathcal{X}},s_{\mathcal{Y}})
		\text{ and } w := w^{s}(s,(s_{\mathcal{X'}},s_{\mathcal{Y'}}))$.
		\begin{itemize}
			\item If $(s,(s_{\mathcal{X'}},s_{\mathcal{Y'}}))\not\models{\rho^e}$, then according to Def.~\ref{def:ECpre} and Def.~\ref{def:dualECpre}, we have that
			${\textsf{EC}_c((s,(s_{\mathcal{X'}},s_{\mathcal{Y'}})), e)} = {0}$ and
			${\sim{\overline{\textsf{EC}}_c((s,(s_{\mathcal{X'}},s_{\mathcal{Y'}})),\sim{e})}} = {\sim+\infty} = {0}$.
			\item Otherwise, if $e=+\infty$ (iff $\sim e=0$) or $(s,(s_{\mathcal{X'}},s_{\mathcal{Y'}}))\models{\rho^e\wedge{\neg\rho^s}}$,
			we have that ${\textsf{EC}_c((s,(s_{\mathcal{X'}},s_{\mathcal{Y'}})), e)} = {+\infty} = {\sim\overline{\textsf{EC}}_c((s,(s_{\mathcal{X'}},s_{\mathcal{Y'}})), \sim{e})} = {\sim{0}}$.
			\item If $(s,(s_{\mathcal{X'}},s_{\mathcal{Y'}}))\models{\rho^e\wedge{\rho^s}}$ and $e=0$ (iff
			$\sim e=+\infty$), we have one of the following cases: (1) if $0-w > c$ (iff $w + c < 0$), then ${\textsf{EC}_c((s,(s_{\mathcal{X'}},s_{\mathcal{Y'}})), 0)} = {+\infty}$ and
			${\sim\overline{\textsf{EC}}_c((s,(s_{\mathcal{X'}},s_{\mathcal{Y'}})), +\infty)} = {\sim{0}} = {+\infty}$;
			(2) if $0-w\leq0$ (iff $w\geq{0}$),
			then ${\textsf{EC}_c((s,(s_{\mathcal{X'}},s_{\mathcal{Y'}})), 0)} = {0} =
			{\sim\overline{\textsf{EC}}_c((s,(s_{\mathcal{X'}},s_{\mathcal{Y'}})),+\infty)} = {\sim{+\infty}}$; (3) if $0<0-w\leq{c}$, we have that
			$1\leq c+1+w\leq c$, and thus, by the definition of $\sim$,
			${\sim \overline{\textsf{EC}}_c((s,(s_{\mathcal{X'}},s_{\mathcal{Y'}})), +\infty)} = {\sim{(c+1+w)}} = {(c+1)-(c+1+w)} = {-w} = {\textsf{EC}_c((s,(s_{\mathcal{X'}},s_{\mathcal{Y'}})), 0)}$.
			\item In the remaining case we have that $(s,(s_{\mathcal{X'}},s_{\mathcal{Y'}}))\models{\rho^e\wedge{\rho^s}}$,
			$e\not=0$, and $e\not = +\infty$, thus $\sim{e}=c+1-e$; (1) if ${e-w} > {c}$, then
			${\sim{e}+w} \leq {0}$, thus
			${\sim\overline{\textsf{EC}}_c((s,(s_{\mathcal{X'}},s_{\mathcal{Y'}})),\sim{e})} = {\sim{0}} = {+\infty} =
			{\textsf{EC}_c((s,(s_{\mathcal{X'}},s_{\mathcal{Y'}})), e)}$; (2) if
			${e-w} \leq {0}$, then ${\sim{e}+w} > {c}$, thus
			${\sim\overline{\textsf{EC}}_c((s,(s_{\mathcal{X'}},s_{\mathcal{Y'}})),
				\sim{e})} = {\sim{+\infty}} = {0} = {\textsf{EC}_c((s,(s_{\mathcal{X'}},s_{\mathcal{Y'}})), e)}$;
			(3) if ${0} < {e-w} \leq {c}$, then $0 < {\sim{e} + w \leq{c}}$, and thus
			${\sim\overline{\textsf{EC}}_c((s,(s_{\mathcal{X'}},s_{\mathcal{Y'}})), \sim{e})} = {\sim{(\sim e + w)}} = \linebreak(c+1) - (\sim e + w) = {e-w} =
			\textsf{EC}_c((s,(s_{\mathcal{X'}},s_{\mathcal{Y'}})), e)$.
		\end{itemize}
		\item From~\ref{proof:dualECprePart1} together with Lem.~\ref{lem:deMorganAlgebra}, it follows that
		for all $f \in \mathit{EF(c)}$ and $s \in 2^{\mathcal{V}}$:\\ ${\sim{\ECpre(f)(s)}} =
		{\sim{\max\limits_{s_{\mathcal{X}}\in{2^\mathcal{X}}}
				\lbrack\min\limits_{s_{\mathcal{Y}}\in{2^\mathcal{Y}}}\textsf{EC}_c((s,p(s_{\mathcal{X}}, s_{\mathcal{Y}})),
				f(s_{\mathcal{X}},s_{\mathcal{Y}}))\rbrack}}
		\\= {\sim{\max\limits_{s_{\mathcal{X}}\in{2^\mathcal{X}}}
				\lbrack\min\limits_{s_{\mathcal{Y}}\in{2^\mathcal{Y}}}\sim\overline{\textsf{EC}}_c((s,p(s_{\mathcal{X}}, s_{\mathcal{Y}})),
				\sim f(s_{\mathcal{X}},s_{\mathcal{Y}}))\rbrack}}
		\\= {\sim{\max\limits_{s_{\mathcal{X}}\in{2^\mathcal{X}}}
				\lbrack\sim\max\limits_{s_{\mathcal{Y}}\in{2^\mathcal{Y}}}\overline{\textsf{EC}}_c((s,p(s_{\mathcal{X}}, s_{\mathcal{Y}})),
				\sim {f(s_{\mathcal{X}},s_{\mathcal{Y}}))}\rbrack}}
		\\= {\sim\sim\min\limits_{s_{\mathcal{X}}\in{2^\mathcal{X}}}
			\lbrack\max\limits_{s_{\mathcal{Y}}\in{2^\mathcal{Y}}}\overline{\textsf{EC}}_c((s,p(s_{\mathcal{X}}, s_{\mathcal{Y}})),
			\sim f(s_{\mathcal{X}},s_{\mathcal{Y}}))\rbrack}
		\\= {\min\limits_{s_{\mathcal{X}}\in{2^\mathcal{X}}}
			\lbrack\max\limits_{s_{\mathcal{Y}}\in{2^\mathcal{Y}}}\overline{\textsf{EC}}_c((s,p(s_{\mathcal{X}}, s_{\mathcal{Y}})),
			\sim f(s_{\mathcal{X}},s_{\mathcal{Y}}))\rbrack}
		={\ECpreEnv(\sim f)(s)}$.\qedhere
	\end{enumerate}
\end{proof}

\newpage
\subsection{Proofs}\label{app:energyMuCalc:proofs}

\subsubsection{Proofs of Sect.~\ref{sec:EngMuCalcSyntaxSemantics}}

\begin{proof}[Proof of Prop.~\ref{prop:ECpreMonotone}]\hfill

	Let $c \in \mathbb{N}$ be an upper bound, and let $f,g \in \mathit{EF(c)}$ such that ${f}\preceq{g}$.
	
	First, given $s \in 2^{\mathcal{V}}$, $s_{\mathcal{X}}\in{2^\mathcal{X}}$, and
	$s_{\mathcal{Y}}\in{2^\mathcal{Y}}$, we show that
	${\textsf{E}(f)} \preceq {\textsf{E}(g)}$ where for all $h \in \mathit{EF(c)}$,
	${\textsf{E}(h)} := {\textsf{EC}_c((s,p(s_{\mathcal{X}},s_{\mathcal{Y}})),
		h(s_{\mathcal{X}},s_{\mathcal{Y}}))}$.
	If $(s,p(s_{\mathcal{X}},s_{\mathcal{Y}})) \not \models \rho^{e}$, then $\textsf{E}(f) =
	\textsf{E}(g) = 0$. Otherwise, if
	$f(s_{\mathcal{X}},s_{\mathcal{Y}}) = +\infty$ or $(s,
	p(s_{\mathcal{X}},s_{\mathcal{Y}})) \models \rho^{e} \wedge \neg \rho^{s}$, it follows
	that $\textsf{E}(f) = +\infty \preceq \textsf{E}(g)$.
	Otherwise, we have that $(s,p(s_{\mathcal{X}},s_{\mathcal{Y}}))
	\models \rho^{e} \wedge \rho^{s}$ and $f(s_{\mathcal{X}},s_{\mathcal{Y}}) \not = +\infty$, thus
	$g(s_{\mathcal{X}},s_{\mathcal{Y}}) \not = +\infty$
	and $g(s_{\mathcal{X}},s_{\mathcal{Y}}) \leq
	f(s_{\mathcal{X}},s_{\mathcal{Y}})$. In this case, (1) if $f(s_{\mathcal{X}},s_{\mathcal{Y}}) -
	w^{s}(s,p(s_{\mathcal{X}},s_{\mathcal{Y}})) > c$, then $\textsf{E}(f) = +\infty
	\preceq \textsf{E}(g)$; and (2) if $g(s_{\mathcal{X}},s_{\mathcal{Y}}) -
	w^{s}(s,p(s_{\mathcal{X}},s_{\mathcal{Y}})) \leq f(s_{\mathcal{X}},s_{\mathcal{Y}}) -
	w^{s}(s,p(s_{\mathcal{X}},s_{\mathcal{Y}})) \leq c$, then $\textsf{E}(f) = \max[0, f(s_{\mathcal{X}},s_{\mathcal{Y}}) -
	w^{s}(s,p(s_{\mathcal{X}},s_{\mathcal{Y}}))] \preceq \max[0, g(s_{\mathcal{X}},s_{\mathcal{Y}}) - w^{s}(s,p(s_{\mathcal{X}},s_{\mathcal{Y}}))] = \textsf{E}(g)$.
	
	Second, we show that $\ECpre(f) \preceq \ECpre(g)$.
	Let $s\in 2^{\mathcal{V}}$ be a state, and let us show that $\mathit{ECpre_{\sys}}(f)(s) \preceq \mathit{ECpre_{\sys}}(g)(s)$.
	Note that by the above, we have that for every
	$s_{\mathcal{X}}\in{2^\mathcal{X}}$ and
	$s_{\mathcal{Y}}\in{2^\mathcal{Y}}$:
	${\textsf{EC}_c((s,p(s_{\mathcal{X}},s_{\mathcal{Y}})),
		f(s_{\mathcal{X}},s_{\mathcal{Y}}))} \preceq {\textsf{EC}_c((s,
		p(s_{\mathcal{X}},s_{\mathcal{Y}})),
		g(s_{\mathcal{X}},s_{\mathcal{Y}}))}$.
	Since $\min$ and $\max$ are monotone w.r.t. $\preceq$ in all of their arguments, it follows
	that ${\mathit{ECpre_{\sys}}(f)(s)} \preceq {\mathit{ECpre_{\sys}}(g)(s)}$.
\end{proof}

\begin{proof}[Proof of Lem.~\ref{lem:deMorganAlgebra}]\hfill

	It holds that ${\mathit{EFL(c)}} = {\langle \mathit{EF(c)}, \min, \max, f_{+\infty}, f_{0} \rangle}$ is a bounded distributive lattice. Thus, it remains to show that the unary operator $\sim$ satisfies the following axioms:
	(1) for all $x \in {\mathit{E(c)}}$: ${\sim\sim{x}} = {x}$ (\emph{involution}), and (2) for all $x,y\in{\mathit{E(c)}}$: ${\sim(\max(x,y))} =
	{\min(\sim{x},\sim{y})}$ (\emph{De Morgan's laws}).\footnote{Note
		that the axiom: $\forall x,y\in{\mathit{E(c)}}: {\sim(\min(x,y))} =
		{\max(\sim{x},\sim{y})}$ also holds, since (1) and (2) imply that ${\sim(\min(x,y))} = {\sim{\min(\sim\sim{x},
				\sim\sim{y}))}} =
		{\sim{(\sim(\max(\sim{x}, \sim{y})))}} = {\max(\sim{x},{\sim{y}})}$.}
	\begin{enumerate}
		\item[(1)] For $x=+\infty$, we have that
		${\sim\sim{x}} = {\sim{0}} = {+\infty} = {x}$, and similarly for ${x}={0}$.
		For $x \in{\mathit{E(c)}}{\setminus{\{+\infty,0\}}}$, it holds that ${\sim{\sim{x}}} = {c+1-\sim{x}} = {c+1-(c+1-x)} = {x}$.
		\item[(2)] Let $x,y\in{\mathit{E(c)}}$:
		\begin{itemize}
			\item if $x=+\infty$ or $y=+\infty$, then ${\sim(\max(x,y))} =
			{\sim{+\infty}} = {0} = {\min(\sim{x},\sim{y})}$;
			\item if $x=0$ and $y \not=+\infty$, then  
			${\sim(\max(x,y))} = {\sim{y}} = {\min(+\infty,\sim{y})} =
			{\min(\sim{x},\sim{y})}$ (and similarly if $x\not=+\infty$ and $y=0$);
			\item if $x,y\in{\mathit{E(c)}}{\setminus{\{+\infty,0\}}}$, then it holds that
			$\max(x,y)\in{\mathit{E(c)}}{\setminus{\{+\infty,0\}}}$, and thus ${\sim(\max(x,y))} = {c+1-\max(x,y)} = {c+1 + \min
				(-x, -y)} = \min(c+1-x, c+1-y) = \min(\sim{x}, \sim{y})$.\qedhere
		\end{itemize}
	\end{enumerate}
\end{proof}

\subsubsection{Proofs of Sect.~\ref{sec:EngMuCalcCorrect}}\label{app:energyMuCalc:proofs:correctness}


\begin{proof}[Proof of Thm.~\ref{thm:reductionCorrectness}]\hfill

	First, we prove for the system player. Assume that $g^*$ is a winning strategy for the system from $(s,c_0)$ in $G^*$. We define a strategy $g$ for the system in $G^w$, as follows.
	For $k \geq 0$, take a prefix $s_{0},\dots,s_k,s_\mathcal{X}\in (2^\mathcal V)^{+}2^\mathcal X$
	such that $s_{0}=s$ and $(s_{k}, p(s_{\mathcal{X}}))\models \rho^{e}$.
	If there are $c_1,\dots,c_k \in ([0,c])^{k}$ such that $(s,c_0),(s_1,c_1),\dots,(s_k,c_k) \in (2^{\mathcal{V}^{*}})^{+}$ is consistent with $g^{*}$, choose such a prefix,
	let $(s_\mathcal Y,c_{k+1})=g^*( (s,c_0),(s_1,c_1),\dots,(s_k,c_k),s_\mathcal X )$, and define $g(s,s_1,\dots,s_k,s_\mathcal{X})=s_\mathcal Y$.
	By the construction of $G^*$, it is not difficult to prove by induction on $k$ that

	\begin{itemize}
		\item $(s_{k},p(s_\mathcal X,s_\mathcal Y))\models \rho^s$.
		\item $s_{0},\ldots,s_{k},(s_\mathcal X,s_\mathcal Y)$ is the projection to $\mathcal V$ of the unique prefix of a play in $G^{*}$,
		$(s_{0},c_{0}),\ldots,\linebreak(s_{k}, c_{k}),(s_\mathcal X,s_\mathcal Y, c_{k+1})$, consistent with $g^{*}$.
		\item ${\sf EL}_c(G^w,c_0,(s_{0},\ldots,s_{k},(s_\mathcal X,s_\mathcal Y)))\geq c_{k+1}$.
	\end{itemize}
	Hence, since $g^*$ wins for the system from $(s,c_0)$, we get that $g$ wins for the system from $s$ w.r.t. $c$ for the initial credit $c_0$.
	
	%
	For the other direction, assume that $g$ is a winning strategy for the system from $s$ w.r.t. $c$ for an initial credit $c_0$, in $G^w$. We define a strategy $g^*$ for the system in $G^*$ from $(s,c_0)$, as follows:
	\begin{itemize}
		\item $g^*=\bigcup_{i=1}^{+\infty}g^*_i$ where $g_i^*:(2^{\mathcal{V}^*})^i 2^\mathcal X\rightarrow 2^{\mathcal{Y}^{*}}$.
		\item For $\sigma= s, s_\mathcal X\in 2^{\mathcal V} 2^\mathcal X$ such that $(s,p(s_\mathcal X))\models \rho^e$, let $g(\sigma)=s_\mathcal Y$. Then, we set $g^*_1((s,c_0),s_\mathcal X) =(s_\mathcal Y,\min\{c,c_0+w^s(s,p(s_\mathcal X,s_\mathcal Y))\})$.
		
		\item Assume that $g_i^*$ has been defined. Take $\sigma=(s,c_0),(s_1,c_1),\dots,(s_{i},c_{i}), s_\mathcal X\in (2^{\mathcal V^*})^{i+1} 2^\mathcal X$ such that $(s,c_0),(s_1,c_1),\dots,(s_{i},c_{i})$ is consistent with $g_i^*$ and $(s_{i},p(s_\mathcal X))\models \rho^e$. Let $g(s,s_1,\dots,s_{i},s_\mathcal X)=s_\mathcal Y$. Then, we set $g^*_{i+1}(\sigma) =(s_\mathcal Y,\min\{c,c_{i}+w^s(s_{i},p(s_\mathcal X,s_\mathcal Y))\})$.
	\end{itemize}
	By applying an induction on $i$, we get that if $(s,c_0),(s_1,c_1),\dots,(s_{i},c_{i})$ is consistent with $g_i^*$ and $s_\mathcal X \in 2^{\mathcal{X}}$, the following holds:
	\begin{enumerate}
		\item $(s,s_1,\dots,s_{i})$ is consistent with $g$.
		
		\item If ${(s_{i},p(s_\mathcal X))\models \rho^e}$, then for ${s_\mathcal Y=g(s,s_1,\dots s_{i},s_\mathcal X)}$,\\${g^*_{i+1}}({(s,c_0)},{(s_1,c_1)},\dots, {(s_{i},  c_{i})},s_\mathcal X) = (s_\mathcal Y,{\sf EL}_c(G^w,c_0,(s,s_1,\dots,s_{i},(s_\mathcal X,s_\mathcal Y)))$.
	\end{enumerate}
	Since $g$ is winning for the system, it follows that $g_{i+1}^*$ is well-defined (and hence $g^*$ is well-defined), and $g^*$ is winning for the system from $(s,c_0)$.

	
	We turn now to prove the claim for the environment. First, assume that $g$ is a winning strategy for the environment from $s$ in $G^w$ w.r.t. $c$ for an initial credit $c_0$. Then, the system cannot win from $(s,c_0)$ in $G^*$. Since $\omega$-regular games are determined, the environment has a winning strategy from $(s,c_0)$ in $G^*$.
	
	For the other direction, assume that $g^*$ is a winning strategy for the environment in $G^*$ from $(s,c_0)$, and we construct a winning strategy $g$ for the environment in $G^w$ from $s$ w.r.t. $c$ for an initial credit $c_0$. We define $g=\bigcup_{i=1}^{+\infty} g_i$ where $g_i:(2^\mathcal V)^i\into 2^\mathcal X$, such that for $i>1$, the following holds:

	
	
	\begin{quote}
		$g_i(s,s_1,\cdots,s_{i-1})$ is defined iff  $(s,c_0),(s_1,c_1),  \dots,(s_{i-1},c_{i-1})$ is consistent with $g^*$ where, for $l\in\{1,\dots,i-1\}$, $c_l={\sf EL}_c(G^w,c_0,(s,s_1, \dots,s_l))$. Furthermore, in this case, $g_i(s,s_1,\dots,s_{i-1})=g^*( (s,c_0),(s_1,c_1),\dots,(s_{i-1},c_{i-1}))$.
	\end{quote}

	\noindent The construction of $g$ is by induction on $i$ as follows:
	\begin{itemize}
		
		\item $g_1(s)=g^*(s,c_0)$.
		
		\item Assume that $g_i$ has been defined, and we aim to define $g_{i+1}$. Let $g_i(s,s_1,\dots,s_{i-1})=s_\mathcal X$. Take $s_\mathcal Y\in 2^\mathcal Y$ such that $(s_{i-1},p(s_\mathcal X,s_\mathcal Y))\models \rho^s$. 

		\begin{itemize}
			\item First, if ${\sf EL}_c(G^w,c_0,(s,s_1,\dots,s_{i-1},(s_\mathcal X,s_\mathcal Y)))<0$, the environment wins in this particular choice of the system and there is no need to define $g_{i+1}( s,s_1,\dots,s_{i-1},(s_\mathcal X,s_\mathcal Y) )$.
			
			\item Otherwise, for $l\in\{1,\dots,{i-1}\}$, let $c_l={\sf EL}_c(G^w,c_0,(s,s_1,\dots,s_{l}))$. Therefore, by the induction hypothesis, we have that $(s,c_0),(s_1,c_1),\dots,(s_{i-1},c_{i-1})$ is consistent with $g^*$ and $g^*( (s,c_0),(s_1,c_1),\dots,(s_{i-1},c_{i-1}) )=s_\mathcal X$. Let $s_i=(s_\mathcal X,s_\mathcal Y)$ and $c_i={\sf EL}_c(G^w, c_0, (s,s_1,\dots,s_i) )=\min\{c,c_{i-1}+w^s(s_{i-1},p(s_i))\}\geq 0$. By the construction of $G^*$ (Def.~\ref{def:naiveReduction}), $( (s_{i-1},c_{i-1}), p(s_i,c_i) )\models \rho^{s*}$; we define $g_{i+1}(s,s_1,\dots,s_i)=s_\mathcal X'$ where $s_\mathcal X'=g^*((s,c_0),(s_1,c_1),\dots,(s_i,c_i))$.
		\end{itemize}

	\end{itemize}
	
	\noindent Now, we prove that $g$ indeed wins for the environment. Consider a play $\sigma$ from $s$, consistent with $g$. If $\sigma$ ends in a deadlock for the system, or the energy level decreases in $\sigma$ to a negative value, the environment wins. If $\sigma$ is infinite and the energy level remains non-negative along $\sigma$, write $\sigma=s,s_1,s_2,\dots$ and observe that $\sigma^*=(s,c_0),(s_1,c_1),(s_2,c_2),\dots$ is consistent with $g^*$ where $c_l={\sf EL}_c(G^w,c_0,(s,s_1,\dots,s_l))$. Hence, by assumption, $\sigma^{*}$ does not satisfy $\varphi$. This implies that $\sigma$ does not satisfy $\varphi$, and consequently, wins for the environment.
	
	It is left to show that $\sigma$ does not end in a deadlock for the environment. Assume otherwise, and take a prefix, $\sigma=s,s_1,\dots,s_k$, that is consistent with $g$ and reaches such a deadlock. Hence, $(s,c_0),(s_1,c_1),\dots,(s_k,c_k)$ is consistent with $g^*$, where $c_0,\dots,c_k$ are defined as before. Since there is no $s_\mathcal X\in 2^\mathcal X$ such that $(s_k,p(s_\mathcal X))\models \rho^e$, $(s_k,c_k)$ is a deadlock for the environment in $G^*$. This, of course, contradicts the assumption that $g^*$ is a strategy that wins for the environment in $G^*$. 
\end{proof}


\begin{proof}[Proof of Lem.~\ref{lem:sysEngMuCalc}]\hfill

	The proof is by induction on the structure of $\psi \in \sysFormulas$.
	
	\begin{itemize}
		
		\item $\psi = v$ for $v \in\mathcal{V}$:
		$\mathit{val}
		\preceq\llbracket{v}\rrbracket^{\GEc}_{\eValuation}(s)$ iff
		$\llbracket{v}\rrbracket^{\GEc}_{\eValuation}(s) = 0$ iff $s \models v$ iff 
		$(s, \mathit{val}) \models v$ iff $(s,
		\mathit{val})\in\llbracket{v}\rrbracket^{G^{*}}_{\mathcal{E}}$.
		\item $\psi = \neg v$ for $v \in\mathcal{V}$: $\mathit{val}
		\preceq\llbracket{\neg v}\rrbracket^{\GEc}_{\eValuation}(s)$ iff
		$\llbracket{\neg v}\rrbracket^{\GEc}_{\eValuation}(s) = 0$ iff $(s,
		\mathit{val})\!\in\!\llbracket{
			\neg v}\rrbracket^{G^{*}}_{\mathcal{E}}$.
		\item $\psi =  X$: $\mathit{val}
		\preceq\llbracket{X}\rrbracket^{\GEc}_{\eValuation}(s)$ iff
		$\mathit{val} \preceq \eValuation(X)(s)$ iff$_{\text{(premise)}}$ $(s, \mathit{val})
		\in \mathcal{E}(X)$
		iff $(s,\mathit{val})\in\llbracket{X}\rrbracket^{G^{*}}_{\mathcal{E}}$.
		\item $\psi = \phi_{1} \vee \phi_{2}$: $\mathit{val}
		\preceq\llbracket{\phi_{1}^{\op{E}} \vee
			\phi_{2}^{\op{E}}}\rrbracket^{\GEc}_{\eValuation}(s)$ iff $\mathit{val}
		\preceq
		\min(\llbracket{\phi_{1}^{\op{E}}}\rrbracket^{\GEc}_{\eValuation}(s),\llbracket{\phi_{2}^{\op{E}}}\rrbracket^{\GEc}_{\eValuation}(s))$
		iff $\mathit{val}
		\preceq \llbracket{\phi_{1}^{\op{E}}}\rrbracket^{\GEc}_{\eValuation}(s)$ or $\mathit{val}
		\preceq \llbracket{\phi_{2}^{\op{E}}}\rrbracket^{\GEc}_{\eValuation}(s)$
		iff$_{\text{(i.h.)}}$ $(s, \mathit{val})\in\llbracket{
			\phi_{1}}\rrbracket^{G^{*}}_{\mathcal{E}}$ or $(s, \mathit{val})\in\llbracket{
			\phi_{2}}\rrbracket^{G^{*}}_{\mathcal{E}}$ iff $(s,
		\mathit{val})\in\llbracket{
			\phi_{1} \vee \phi_{2}}\rrbracket^{G^{*}}_{\mathcal{E}}$.
		\item $\psi = \phi_{1} \wedge \phi_{2}$: $\mathit{val}
		\preceq\llbracket{\phi_{1}^{\op{E}} \wedge
			\phi_{2}^{\op{E}}}\rrbracket^{\GEc}_{\eValuation}(s)$\\
		iff $\mathit{val}\preceq \max(\llbracket{\phi_{1}^{\op{E}}}\rrbracket^{\GEc}_{\eValuation}(s),
		\llbracket{\phi_{2}^{\op{E}}}\rrbracket^{\GEc}_{\eValuation}(s))$\\
		iff $\mathit{val}\preceq \llbracket{\phi_{1}^{\op{E}}}\rrbracket^{\GEc}_{\eValuation}(s)$ and $\mathit{val}
		\preceq \llbracket{\phi_{2}^{\op{E}}}\rrbracket^{\GEc}_{\eValuation}(s)$
		iff$_{\text{(i.h.)}}$ $(s, \mathit{val})\in\llbracket{ \phi_{1}}\rrbracket^{G^{*}}_{\mathcal{E}}$ and $(s,
		\mathit{val})\in\llbracket{
			\phi_{2}}\rrbracket^{G^{*}}_{\mathcal{E}}$ iff $(s,
		\mathit{val})\in\llbracket{
			\phi_{1} \wedge \phi_{2}}\rrbracket^{G^{*}}_{\mathcal{E}}$.
		\item $\psi = \circlediamond\phi$:
		\begin{equation*}
		\mathit{val}\preceq\llbracket{\circlediamond_{\op{E}}\phi^{\op{E}}}\rrbracket^{\GEc}_{\eValuation}(s)
		\end{equation*}
		iff
		\begin{equation*}
		\mathit{val}\preceq \max\limits_{s_{\mathcal{X}}\in{2^\mathcal{X}}}
		\lbrack\min\limits_{s_{\mathcal{Y}}\in{2^\mathcal{Y}}}\textsf{EC}_c((s,p(s_{\mathcal{X}},s_{\mathcal{Y}})),
		\llbracket\phi^{\op{E}}\rrbracket^{\GEc}_{\eValuation}(s_{\mathcal{X}},s_{\mathcal{Y}}))\rbrack
		\end{equation*}
		iff
		\begin{equation*}
		\forall s_{\mathcal{X}}\in 2^{\mathcal{X}} \exists s_{\mathcal{Y}}\in
		2^{\mathcal{Y}}: \mathit{val}
		\preceq \textsf{EC}_c((s,p(s_{\mathcal{X}},s_{\mathcal{Y}})),
		\llbracket\phi^{\op{E}}\rrbracket^{\GEc}_{\eValuation}(s_{\mathcal{X}},s_{\mathcal{Y}}))
		\end{equation*}
		iff$_{\text{Def.~\ref{def:ECpre}}}$
		\begin{equation*}
		\begin{split}
		\forall s_{\mathcal{X}}\in 2^{\mathcal{X}} \exists s_{\mathcal{Y}}\in
		2^{\mathcal{Y}}:(s,p(s_{\mathcal{X}})) \models \rho^{e} &\Rightarrow \big[
		(s,p(s_{\mathcal{X}},s_{\mathcal{Y}}))\models \rho^{s} \text{ and }\\
		&\llbracket\phi^{\op{E}}\rrbracket^{\GEc}_{\eValuation}(s_{\mathcal{X}},s_{\mathcal{Y}})
		\not = +\infty \text{ and}\\
		\mathit{val} \preceq \max \lbrack 0,&
		\llbracket\phi^{\op{E}}\rrbracket^{\GEc}_{\eValuation}(s_{\mathcal{X}},s_{\mathcal{Y}})
		- w^{s}(s, p(s_{\mathcal{X}},s_{\mathcal{Y}}))\rbrack\big ]
		\end{split}
		\end{equation*}
		iff
		\begin{equation*}
		\begin{split}
		\forall
		s_{\mathcal{X}}\in 2^{\mathcal{X}} \exists s_{\mathcal{Y}}\in 2^{\mathcal{Y}} :
		(s,p(s_{\mathcal{X}})) \models \rho^{e} \Rightarrow \big[ (s,p(s_{\mathcal{X}},s_{\mathcal{Y}}))\models \rho^{s} &\text{ and }\\
		\min[c, \mathit{val} + w^{s}(s, p(s_{\mathcal{X}},s_{\mathcal{Y}}))] \preceq
		\llbracket\phi^{\op{E}}\rrbracket^{\GEc}_{\eValuation}(s_{\mathcal{X}},s_{\mathcal{Y}})\big]
		\end{split}
		\end{equation*}
		iff$_{(\bigstar)}$
		\begin{equation*}
		\begin{split}
		\forall s_{\mathcal{X}}\in 2^{\mathcal{X}} \exists s_{\mathcal{Y^{*}}}\in 2^{\mathcal{Y^{*}}} :
		(s, \mathit{val},p(s_{\mathcal{X}})) \models \rho^{e} &\Rightarrow\\
		\big[(s,\mathit{val},p(s_{\mathcal{X}}, s_{\mathcal{Y^{*}}}))\models
		\rho^{s*} \text{ and } (s_{\mathcal{X}},s_{\mathcal{Y^{*}}}) \in
		\llbracket{\phi}\rrbracket^{G^{*}}_{\mathcal{E}}\big]
		\end{split}
		\end{equation*}
		iff
		\begin{equation*}
		(s,\mathit{val})\in\llbracket{\circlediamond\phi}\rrbracket^{G^{*}}_{\mathcal{E}}.
		\end{equation*}
		
		$(\bigstar)$:
		\begin{itemize}
			\item \emph{``if'':} Assume that for all $s_{\mathcal{X}}\in
			2^{\mathcal{X}}$ such that $(s, \mathit{val},p(s_{\mathcal{X}}))
			\models \rho^{e}$, there exists $s_{\mathcal{Y^{*}}}\in
			2^{\mathcal{Y^{*}}}$ such that $(s,\mathit{val},p(s_{\mathcal{X}},s_{\mathcal{Y^{*}}}))\models \rho^{s*} \text{ and } (s_{\mathcal{X}},s_{\mathcal{Y^{*}}})
			\in \llbracket{\phi}\rrbracket^{G^{*}}_{\mathcal{E}}$. Let $s_{\mathcal{X}}\in
			2^{\mathcal{X}}$ such that $(s,p(s_{\mathcal{X}})) \models \rho^{e}$. As
			$(s, \mathit{val},p(s_{\mathcal{X}})) \models \rho^{e}$,
			the premise implies that there exists $s_{\mathcal{Y^{*}}}\in
			2^{\mathcal{Y^{*}}}$ such that
			$(s,p(s_{\mathcal{X}},s_{\mathcal{Y}}))\models
			\rho^{s}$, $\mathit{val} + w^{s}(s,p(s_{\mathcal{X}},s_{\mathcal{Y}})) \geq \mathit{\mathit{val}'}$
			, and $(s_{\mathcal{X}},s_{\mathcal{Y^{*}}})\in
			\llbracket{\phi}\rrbracket^{G^{*}}_{\mathcal{E}}$, where
			$s_{\mathcal{Y}} = s_{\mathcal{Y^{*}}}|_{\mathcal{Y}}$, $\mathit{\mathit{val}'}\in [0,c]$, and $\mathit{\mathit{val}'} = s_{\mathcal{Y^{*}}}|_{yDom}$.
			Therefore, $\min[c, \mathit{val} + w^{s}(s,p(s_{\mathcal{X}},s_{\mathcal{Y}}))] \preceq \mathit{\mathit{val'}}$, and it follows from
			the induction hypothesis that $\mathit{val'} \preceq \llbracket\phi^{\op{E}}\rrbracket^{\GEc}_{\eValuation}(s_{\mathcal{X}},s_{\mathcal{Y}})$.
			This implies that
			$\min[c,\mathit{val} + w^{s}(s,p(s_{\mathcal{X}},s_{\mathcal{Y}}))]\!\preceq\!\llbracket\phi^{\op{E}}\rrbracket^{\GEc}_{\eValuation}(s_{\mathcal{X}},s_{\mathcal{Y}})$, as required.
			
			
			\item \emph{``only if'':} Assume that for all $s_{\mathcal{X}}\in
			2^{\mathcal{X}}$ such that $(s,p(s_{\mathcal{X}})) \models \rho^{e}$,
			there exists $s_{\mathcal{Y}}\in 2^{\mathcal{Y}}$ such that
			$(s,p(s_{\mathcal{X}},s_{\mathcal{Y}}))\!\models\!\rho^{s}$ and
			$\min[c,\mathit{val} + w^{s}(s,p(s_{\mathcal{X}},s_{\mathcal{Y}}))] \preceq
			\llbracket\phi^{\op{E}}\rrbracket^{\GEc}_{\eValuation}(s_{\mathcal{X}},s_{\mathcal{Y}})$.
			
			Let $s_{\mathcal{X}}\in 2^{\mathcal{X}}$ such that $(s,
			\mathit{val},p(s_{\mathcal{X}})) \models \rho^{e}$. As 
			it also holds that $(s,p(s_{\mathcal{X}})) \models \rho^{e}$,
			it follows from the premise and the induction hypothesis that there exists $s_{\mathcal{Y}}\in 2^{\mathcal{Y}}$ such
			that $(s,p(s_{\mathcal{X}},s_{\mathcal{Y}}))\models \rho^{s}$ and
			$(s_{\mathcal{X}},s_{\mathcal{Y^{*}}}) \in
			\llbracket{\phi}\rrbracket^{G^{*}}_{\mathcal{E}}$ where
			$s_{\mathcal{Y^{*}}} = (s_{\mathcal{Y}}, \min[c, \mathit{val} +
			w^{s}(s, p(s_{\mathcal{X}},s_{\mathcal{Y}}))])$.
			Since $\mathit{val} + w^{s}(s, p(s_{\mathcal{X}},s_{\mathcal{Y}})) \geq \min[c,\mathit{val} +
			w^{s}(s, p(s_{\mathcal{X}},s_{\mathcal{Y}}))]$, it also holds that $(s,\mathit{val},p(s_{\mathcal{X}},s_{\mathcal{Y^{*}}}))\models
			\rho^{s*}$, as required.
		\end{itemize}
		\item $\psi = \mu X \phi$ (We only show the proof for $\mu$ as the proof for $\nu$ is similar):
		
		Note that (1) $\mathit{val} \preceq \llbracket{\mu X \phi^{\op{E}}}\rrbracket^{\GEc}_{\eValuation}(s)$
		iff$_{(\text{Def.~\ref{def:sysEngMuCalcSemantics}})}$ $\mathit{val} \preceq \min\limits_{i}\lbrack{h_i}\rbrack(s)$ where
		$h_0=f_{+\infty}$ and $h_{i+1}=
		\llbracket{\phi^{\op{E}}}\rrbracket^{\GEc}_{\eValuation[X\mapsto{h_i}]}$; and
		(2) $(s,\mathit{val}) \in \llbracket{\mu
			X \phi}\rrbracket^{G^{*}}_{\mathcal{E}}$ iff$_{(\text{Def.~\ref{def:prop_mu_calculus_semantics}})}$
		$(s,\mathit{val})\in\bigcup\limits_{i}{S_i}$ where $S_0=\emptyset$ and $S_{i+1}= \llbracket{\phi}\rrbracket^{G^{*}}_{\mathcal{E}[X\mapsto{S_i}]}$.
		
		We show by induction that for all $i \in \mathbb{N}$, $\mathit{val} \preceq h_{i}(s)$ iff
		$(s,\mathit{val})\in S_{i}$:
		
		
		\begin{itemize}
			\item \emph{Basis:} It holds that $\mathit{val} \not\preceq +\infty = h_{0}(s)$ and
			$(s,\mathit{val})\not\in \emptyset = S_{0}$.
			\item \emph{Step:} 
			%
			It follows from the premise and the induction hypothesis (over $h_{i}$ and $S_{i}$) that for all $Y \in \mathit{Var}$, for
			all states $s'' \in 2^{\mathcal{V}}$, and for all $\mathit{val''} \in [0,c]$: $\mathit{val''} \preceq \eValuation[X\mapsto{h_i}](Y)(s'')$ iff
			$(s'',\mathit{val''})\in \mathcal{E}[X\mapsto{S_i}](Y)$. Thus, the structural induction hypothesis ensures that
			${{\mathit{val}} \preceq {h_{i+1}(s)}}=
			{\llbracket{\phi^{\op{E}}}\rrbracket^{\GEc}_{\eValuation[X\mapsto{h_i}]}(s)}$ iff
			${{(s,\mathit{val})}\in {S_{i+1}}} = 
			{\llbracket{\phi}\rrbracket^{G^{*}}_{\mathcal{E}[X\mapsto{S_i}]}}$.
		\end{itemize}
		This implies that $\mathit{val}\preceq\min\limits_{i}\lbrack{h_i}\rbrack(s)$ iff $(s,\mathit{val})\in\bigcup\limits_{i}{S_i}$, as required.\qedhere
	\end{itemize}
\end{proof}


\begin{proof}[Proof of Lem.~\ref{lem:envEngMuCalc}]\hfill

	This lemma can be proved by structural induction as the former lemma,
	but it is simpler (and shorter) to obtain this result from Lem.~\ref{lem:sysEngMuCalc},
	using the semantics of negation. As it is omitted from Def.~\ref{def:prop_mu_calculus_grammar} and Def.~\ref{def:prop_mu_calculus_semantics},
	we mention that for a $\mu$-calculus formula $\phi$, the semantics of its
	negation is defined by $\llbracket \neg \phi \rrbracket^G_\mathcal E:=2^\mathcal V\setminus \llbracket \phi \rrbracket^G_\mathcal E$.
	We further require that if $\phi = \mu X\phi_1(X)$ or
	$\phi = \nu X\phi_1(X)$, then $\phi_1$ is syntactically monotone in $X$, i.e.,
	all free occurrences of $X$ in $\phi_1$ fall under
	an even number of negations.
	Recall that we defined the semantics of negation of energy $\mu$-calculus formulas in Sect.~\ref{sec:EngMuCalcEnvSemantics}
	and required the same syntactic restrictions.
	Moreover, notice that the following well-known equations hold for all $\mu$-calculus formulas~\cite{bradfieldmu,Schneider2004}:
	\begin{align}
&{\llbracket \neg\neg\phi \rrbracket^G_\mathcal{E}}={\llbracket \phi \rrbracket^G_\mathcal{E}}.\\
&{\llbracket \neg(\phi \wedge (\text{resp. } \vee) \xi) \rrbracket^G_\mathcal{E}} = {\llbracket (\neg\phi \vee (\text{resp. } \wedge) \neg\xi) \rrbracket^G_\mathcal{E}}.\\
&{\llbracket \neg\circlediamond \text{(resp. $\circlebox$)}\phi \rrbracket^G_\mathcal{E}} = {\llbracket \circlebox \text{(resp. $\circlediamond$)}(\neg\phi) \rrbracket^G_\mathcal{E}}.\\
&{\llbracket \neg\mu\text{(resp. $\nu$)} X\phi(X) \rrbracket^G_\mathcal{E}}  = {\linebreak \llbracket \nu\text{($\text{resp. }\mu$)}X \neg\phi(\neg X)
  \rrbracket^G_\mathcal{E}}.&
\end{align}

	
	\noindent By Lem.~\ref{lem:deMorganAlgebra}, we have
	that the same equations analogously hold for energy $\mu$-calculus formulas (see Eq.~\ref{eq:negation1}-\ref{eq:negation4}). We will refer to
	these equations as the \emph{negation laws}.

	Let $\mathcal{G}_{\sys}$ and $\mathcal{G}_{\env}$ (resp. $\mathcal{G}^{\op{E}}_{\sys}$
	and $\mathcal{G}^{\op{E}}_{\env}$) respectively denote the grammars that generate
	the formulas $\sysFormulas$ and $\envFormulas$ as defined Sect.~\ref{sec:propMuCalculus}
	(resp. $\sysEFormulas$ and $\envEFormulas$ as defined in Sect.~\ref{sec:EngMuCalcSyntaxSemantics}).
	
	Let $\sysFormulas^{\neg}$ and ${\envFormulas}^{\neg}$ (resp. $\sysEFormulas^{\neg}$ and $\envEFormulas^{\neg}$)
	be the sets of all formulas generated
	by $\mathcal{G}_{\sys}$ and $\mathcal{G}_{\env}$ (resp. $\mathcal{G}^{\op{E}}_{\sys}$
	and $\mathcal{G}^{\op{E}}_{\env}$), respectively, together with the new rule of
	$\neg X$ for $X \in \mathit{Var}$, 
	and in which all sub-formulas of the form $\mu X \phi$ or $\nu X \phi$ satisfy that 
	all free occurrences of $X$ in $\phi$ are un-negated\footnote{This guarantees that the extremal fixed-points exist, and thus the semantics is well-defined~\cite[Chapter~3]{Schneider2004}.}.
	Note that $\sysFormulas^{\neg} \setminus \sysFormulas$,
	$\envFormulas^{\neg} \setminus \envFormulas$, $\sysEFormulas^{\neg} \setminus \sysEFormulas$,
	and $\envEFormulas^{\neg} \setminus \envEFormulas$, each consists of formulas in which negated free
	variables, $\neg X$, occur. 
	
	We will make use of the following claim which can be proved by a standard structural induction on $\phi$.
	\begin{clm}
		\label{lem:envToSysNegation}
		Let $\phi \in \envFormulas^{\neg}$ (resp. $\phi \in \envEFormulas^{\neg}$) and let $\eta$ 
		be final result of the application of the negation laws to $\neg \phi$.
		Then, $\eta \in \sysFormulas^{\neg}$ (resp. $\eta \in \sysEFormulas^{\neg}$) and for all 
		variables $X \in \mathit{Var}$: 
		\begin{enumerate}
			\item $X$ occurs free in $\phi$ iff $X$ occurs free in $\eta$.
			\item if $X$ occurs free in $\phi$, then
			\begin{itemize}
				\item if all free occurrences of $X$ in $\phi$ are un-negated, then all
				free occurrences of $X$ in $\eta$ are negated.
				\item if all free occurrences of $X$ in $\phi$ are negated, then all
				free occurrences of $X$ in $\eta$ are un-negated.
			\end{itemize}
		\end{enumerate}
	\end{clm}

	For $\phi \in \sysFormulas^{\neg}$ (resp. $\phi \in \sysEFormulas^{\neg}$)
	we denote by
	$\phi_{+} \in \sysFormulas$ ($\phi_{+} \in \sysEFormulas$)
	the formula obtained from $\phi$ by replacing all occurrences of negated relational
	variables,
	$\neg X$, with their un-negated form, $X$. We will use the following claim, which states a relation between
	the semantics of $\phi \in \sysFormulas^{\neg}$ and $\phi_{+} \in \sysFormulas$. As in the case of the former claim, this claim can also be proved by standard structural induction thus we leave the details to the reader.
	
	\begin{clm}
		 \label{lem:negVarNegVal}
		Let $\eta \in \sysFormulas^{\neg}$ (resp. $\eta \in \sysEFormulas^{\neg}$)
		where each variable that occurs free in $\eta$ either only has negated
		occurrences or only has un-negated occurrences.
		Let $\mathcal{E} : \mathit{Var}\rightarrow (2^{\mathcal{V}} \rightarrow \{0,1\})$
		(resp. $\eValuation:\mathit{Var}\rightarrow \mathit{EF(c)}$)
		be a valuation, and let $\overline{\mathcal{E}}$ (resp. $\overline{\eValuation}$) denote the valuation such that for all free
		variables $X$ that occur negated in $\eta$, $\overline{\mathcal{E}}(X) = 2^{\mathcal{V}} \setminus \mathcal{E}(X)$
		($\overline{\eValuation}(X) = \sim \eValuation(X)$), and for all
		free variables $X$ that occur un-negated in $\eta$,
		$\overline{\mathcal{E}}(X) = \mathcal{E}(X)$ ($\overline{\eValuation}(X) = \eValuation(X)$).
		Then, $\llbracket \eta \rrbracket^{G}_\mathcal{E} = \llbracket \eta_{+} \rrbracket^{G}_{\overline{\mathcal{E}}}$
		(resp. $\llbracket \eta \rrbracket^{\GEc}_{\eValuation} = \llbracket \eta_{+} \rrbracket^{\GEc}_{\overline{\eValuation}}$).
	\end{clm}

	Now, let $\psi \in \envFormulas \subseteq \envFormulas^{\neg}$ that satisfies the premise.
	Let $\eta$ be the formula obtained by applying the negation laws to 
	$\neg \psi$, and consequently,
	we have that
	$\llbracket \psi^\op{E}\rrbracket^{G^w(c)}_\mathcal D =
	\llbracket \neg\neg\psi^\op{E}\rrbracket^{G^w(c)}_\mathcal D
	= \sim\llbracket \eta^\op{E}\rrbracket^{G^w(c)}_\mathcal D$
	and $\llbracket \psi \rrbracket^{G^*}_\mathcal E=2^{\mathcal{V}^{*}}\setminus \llbracket \eta \rrbracket^{G^*}_\mathcal E$.
	Since it follows from Claim~\ref{lem:envToSysNegation} that $\eta \in \sysFormulas^{\neg}$ and
	all free variables in $\eta$ occur negated (as all free variables in $\psi$ occur un-negated),
	by Claim~\ref{lem:negVarNegVal}, we have that
	$\llbracket \eta^\op{E}\rrbracket^{\GEc}_{\eValuation} =
	\llbracket \eta_{+}^\op{E}\rrbracket^{\GEc}_{\overline{\eValuation}}$ and 
	$\llbracket \eta \rrbracket^{G^*}_{\mathcal{E}} =
	\llbracket \eta_{+} \rrbracket^{G^*}_{\overline{\mathcal{E}}}$.
	Therefore, we obtain the following two equations:
	\begin{equation}
	\label{eq:phi-and-neg-phi}
	\llbracket \psi^\op E\rrbracket^{G^w(c)}_\mathcal D = \sim\llbracket \eta_{+}^\op{E}\rrbracket^{G^w(c)}_{\overline{\eValuation}}.
	\end{equation}
	\begin{equation}
	\label{eq:partition-of-G*}
	\text{For all states, $s \in 2^{\mathcal{V}}$, and $c_0\in[0,c]$: }
	(s,c_0)\in \llbracket \psi \rrbracket^{G^*}_\mathcal E \Leftrightarrow (s,c_0)\notin \llbracket \eta_{+} \rrbracket^{G^*}_{\overline{\mathcal{E}}}.
	\end{equation}
	
	\noindent It further follows that for all free variables $X$ that occur in $\eta_{+}$ (resp. $\eta^{\op{E}}_{+}$),
	$\overline{\mathcal{E}}(X) = 2^{\mathcal{V}^{*}} \setminus \mathcal{E}(X)$ (resp. $\overline{\eValuation}(X) = \sim \eValuation(X)$).
	We claim that the following equation holds for all free variables $X$ that occur in $\eta_{+}$, states $s \in 2^{\mathcal{V}}$, and
	$val \in [0,c]$:
	\begin{equation}
	\label{eq:premiseSysLemma}
	\mathit{val}\preceq\overline{\eValuation}(X)(s) \Leftrightarrow (s,val) \in \overline{\mathcal{E}}(X).
	\end{equation}
	
	The proof of Eq.~\ref{eq:premiseSysLemma} is as follows.
	Take a free variable $X$ that occurs in $\eta_{+}$, and a state $s\in 2^\mathcal V$.
	First, assume that $\eValuation(X)(s) = +\infty$. Thus, by the premise, for all $val \in [0,c]$:
	$(s,c-\mathit{val}) \notin {\mathcal{E}}(X)$. This implies that for all $\mathit{val} \in [0,c]$:
	$(s,\mathit{val}) \in \overline{\mathcal{E}}(X)$. By the assumption, $\overline{\eValuation}(X)(s)
	= 0$, and hence, $\forall \mathit{val}\in [0,c]:
	(\mathit{val}\preceq \overline{\eValuation}(X)(s) \wedge (s,\mathit{val}) \in \overline{\mathcal{E}}(X))$, which implies\
	Eq.~\ref{eq:premiseSysLemma}. Second, assume that $\eValuation(X)(s) = 0$. Then, analogously
	to the first case, we obtain that $\forall \mathit{val}\in [0,c]:
	(\mathit{val}\npreceq \overline{\eValuation}(X)(s) \wedge (s,\mathit{val}) \notin \overline{\mathcal{E}}(X)$
	Third, we consider the remaining case where $\eValuation(X)(s) = m \in [1,c]$. Thus, $\overline{\eValuation}(X)(s) = c+1-m$,
	and by the premise,
	for all $\mathit{val}\in [0,c]:$ $\mathit{val}\in [m,c] \Leftrightarrow (s,c-\mathit{val}) \in {\mathcal{E}}(X)$. This implies that
	for all $\mathit{val}\in [0,c]:$ $\mathit{val}\in [0,m-1] \Leftrightarrow (s,c-\mathit{val}) \in \overline{{\mathcal{E}}}(X)$, iff 
	for all $\mathit{val}\in [0,c]:$ $\mathit{val}\in [c+1-m,c] \Leftrightarrow (s,\mathit{val}) \in \overline{{\mathcal{E}}}(X)$, iff
	for all $\mathit{val}\in [0,c]:$ $\mathit{val} \preceq \overline{\eValuation}(X)(s) \Leftrightarrow (s,\mathit{val}) \in \overline{{\mathcal{E}}}(X)$, which concludes
	the proof of Eq.~\ref{eq:premiseSysLemma}.
	
	Note that $\eta^{\op{E}}_{+}\in \sysEFormulas$ and $\eta_{+}\in \sysFormulas$, and by Eq.~\ref{eq:premiseSysLemma},
	the valuations $\overline{\eValuation}$ and $\overline{\mathcal{E}}$
	satisfy the premise of Lem.~\ref{lem:sysEngMuCalc}.
	Take $s\in 2^\mathcal V$ and, first, assume that $\llbracket \eta_{+}^\op E\rrbracket^{G^w(c)}_{\overline{\eValuation}}(s)=+\infty$. Thus, by Eq.~\ref{eq:phi-and-neg-phi}, $\llbracket \psi^\op{E}\rrbracket^{G^w(c)}_\mathcal D(s)=0$, 
	and by Lem.~\ref{lem:sysEngMuCalc} and Eq.~\ref{eq:partition-of-G*}, $\forall c_0\in [0,c]( (s,c_0)\in \llbracket \psi \rrbracket^{G^*}_\mathcal E )$. Therefore, $\forall \mathit{val}\in [0,c] ( \mathit{val}\preceq \llbracket \psi^\op{E}\rrbracket^{G^w(c)}_\mathcal D(s) \wedge (s,c-\mathit{val})\in \llbracket \psi \rrbracket^{G^*}_\mathcal E  )$, which implies the claim.
	In the case where $\llbracket \eta_{+}^\op E\rrbracket^{G^w(c)}_{\overline{\eValuation}}(s)=0$, analogous arguments reveal that $\forall \mathit{val}\in [0,c] ( \mathit{val}\npreceq \llbracket \psi^\op{E}\rrbracket^{G^w(c)}_\mathcal D(s) \wedge (s,c-\mathit{val})\notin \llbracket \psi \rrbracket^{G^*}_\mathcal E  )$.
	It is left to deal with the case where $\llbracket \eta_{+}^\op E\rrbracket^{G^w(c)}_{\overline{\eValuation}}(s)=m\in [1,c]$.
	In this case, by Lem.~\ref{lem:sysEngMuCalc} and Eq.~\ref{eq:partition-of-G*},
	$\forall \mathit{val}\in [0,c] (val<m \Leftrightarrow (s,val)\in \llbracket \psi \rrbracket^{G^*}_\mathcal E )$. Therefore, for all $\mathit{val}\in[0,c]$,
	\[\mathit{val}\preceq \llbracket \psi^\op{E}\rrbracket^{\GEc}_{\eValuation} \Leftrightarrow_{\text{Eq.~\ref{eq:phi-and-neg-phi}}}
	\mathit{val}\preceq \sim\llbracket \eta_{+}^\op{E}\rrbracket^{G^w(c)}_{\overline{\eValuation}} \Leftrightarrow 
	\mathit{val}\preceq c+1-m\]
	\[\Leftrightarrow
	c-val<m \Leftrightarrow 
	(s,c-\mathit{val})\in \llbracket \psi \rrbracket^{G^*}_\mathcal E.
	\qedhere\]
\end{proof}


\section{Proof of Lem.~\ref{lem:lemma-6-revised}}\label{app:lemma-6-revised-proof}
In this appendix, we present a full proof for Lem.~\ref{lem:lemma-6-revised}:
\begin{quote} \textbf{Lemma.~\ref{lem:lemma-6-revised}.}
	Let $G$ be an energy parity game (as defined in Def.~\ref{def:energyParityGame}) with $n$ states, $d$ different priorities 
	and maximal absolute value of the weights, $K$. 
	If $\text{player}_0$ has a winning strategy from a state $s$ w.r.t. $+\infty$ 
	for an initial credit $c_0$, then she has a winning strategy w.r.t. $d(n-1)K$ 
	for an initial credit $\min\{c_0,(n-1)K\}$.
\end{quote}

\noindent We divide the proof into two parts:
	
	\begin{enumerate}[label=\textbf{Part \arabic*.},ref={Part \arabic*},leftmargin=\parindent,itemindent=*]
		\item\label{lem:lemma-6-revised:Part1} Inspired by the proof of Lem.~6 in~\cite{ChatterjeeD12}, we show that $\text{player}_0$ has a strategy in $G$ that wins from $W_0(+\infty)$ w.r.t. $d(n-1)K$ for $(n-1)K$.
		\item\label{lem:lemma-6-revised:Part2} We show that if $c_0<(n-1)K$, a small modification to the strategy constructed in~\ref{lem:lemma-6-revised:Part1} allows $\text{player}_0$ to win for the initial credit $c_{0}$.
	\end{enumerate}

	
\noindent We start by stating five claims which will be useful for proving~\ref{lem:lemma-6-revised:Part1}.
	


		\begin{clm} \label{clm:lemma-6-revised:Claim1} Energy parity games are determined. That is, for an energy parity game, ${G^{ep}}={\langle({V^{ep}}={V_0^{ep}\cup V_1^{ep}},E^{ep}),\prio^{ep},w^{ep}\rangle}$, and $c\in\mathbb{N}\cup\{+\infty\}$, it holds that ${W^{G^{ep}}_0(c)} \cup {W^{G^{ep}}_1(c)}=V^{ep}$.
		\end{clm}
	
	
	In sketch, Claim~\ref{clm:lemma-6-revised:Claim1} is argued as follows. Given an energy parity game, \[{G^{ep}}={\langle({V^{ep}}={V_0^{ep}\cup V_1^{ep}},E^{ep}),\prio^{ep},w^{ep}\rangle},\] and an upper bound $c\in\mathbb{N}\cup\{+\infty\}$, we construct a parity game, \[{G^{p}}={\langle({V^{p}}={V^{p}_0\cup V^{p}_1},E^{p}),\prio^{p}\rangle},\] in a way similar to Def.~\ref{def:naiveReduction}. However, the construction takes into account that, in contrast to WGSs, the players do not necessarily take steps in an alternating manner and each state is controlled solely by one of the players.
	The states of $G^{p}$ have the form $(s,c'_{0})$ where $s$ is a state of $G^{ep}$ and $c'_{0}\in[0,c] \cup \{+\infty\}$ is the accumulated energy level under $c$. States of the form $(s, +\infty)$ are deadlocks for $\text{player}_{0}$
	and correspond to violation of the energy objective. 
	The edges of $G^{p}$ are taken from the original game $G^{ep}$, and update the energy component accordingly.
	Formally,
	
	\begin{itemize}
		\item $V^{p}_{0} = {A_{0}} \cup {\{(s, +\infty) \mid s \in V^{ep}\}}$ and $V^{p}_{1} = A_{1}$ where $A_{i} = \{(s, c'_{0}) \mid s \in V^{ep}_{i} \text{ and } {c'_{0}} \in {{[0,c]} \cap {\mathbb{N}}} \}$.
		\item For ${(s_{1},c_{1})} \in {V^{p}}$ and ${(s_{2},c_{2})} \in {V^{p}}$, $\big( (s_{1},c_{1}), (s_{2},c_{2}) \big) \in E^{p}$ if ${(s_{1},s_{2})} \in {E^{ep}}$, ${c_{1}} \not = {+\infty}$, and either ($\min\{c, c_{1} + w^{ep}(s_{1},s_{2})\} = c_{2} \geq 0$) or ($c_{1} + w^{ep}(s_{1},s_{2}) < 0$ and $c_{2} = +\infty$).
		\item For ${(s,c'_{0})} \in {V^{p}}$, $\prio^{p}(s,c'_{0}) = \prio^{ep}(s)$.
	\end{itemize}

	\noindent Note that $G^{p}$ has the same number of different priorities as $G^{ep}$ has, and in case $c=+\infty$, it has an infinite state space.
	It is not difficult to see that the winning region of $\text{player}_{i}$ in $G^{p}$ indicates which states win for her in $G^{ep}$
	and annotates these states with the winning initial credits. That is, for $s \in V^{ep}$, $c'_{0} \in {{[0,c]} \cap {\mathbb{N}}}$, and $i \in \{0,1\}$:
	${(s,c'_{0})} \in {W^{G^{p}}_{i}}$ iff $s$ wins in $G^{ep}$ for $\text{player}_{i}$ w.r.t. $c$ for the initial credit $c'_{0}$.
	By determinacy of parity games~\cite{martin1975borel,Zielonka98},
	it holds that ${W^{G^{p}}_{0}} \cup {W^{G^{p}}_{1}} = V^{p}$. This implies that ${W^{G^{ep}}_0(c)} \cup {W^{G^{ep}}_1(c)}=V^{ep}$, as required.

	The second claim  is an immediate consequence of Lem.~4 in~\cite{ChatterjeeD12}.
	
	\begin{clm} \label{clm:lemma-6-revised:Claim2} There is a strategy $\gfeStrategy$ for $\text{player}_0$ from $W_0(+\infty)$ such that
		every play $\sigma$, consistent with $\gfeStrategy$, 
		wins the energy objective w.r.t. $d(n-1)K$ for the initial credit $(n-1)K$, and either of the following occurs:
		\begin{itemize}
			\item $\sigma$ also wins the parity objective in $G$. Hence, $\sigma$ wins for $\text{player}_0$.
			\item The sum of the edges' weights traversed along $\sigma$ is unbounded above.
		\end{itemize}
		
		\noindent The strategy $\gfeStrategy$ is called \emph{good-for-energy}. 
	\end{clm}
	
	%
	%
	%
	
	The third claim is the following simple observation:
	
	\begin{clm} \label{clm:lemma-6-revised:Claim3} Let $c\in {\mathbb{N}}$ be a finite upper bound, let $c_{0} \in {[0,c]}$ be an initial credit, and let $\sigma$ be a play in $G$. Take $l,m\in\mathbb{N}$ such that ${c_{0} + m} \leq {c + l}$. Then, for every prefix of $\sigma$, $\sigma[0\ldots k]$,
		we have that ${\textsf{EL}_{c+l}(G,{c_0+m}, {\sigma[0\ldots k]})} \geq {{\textsf{EL}_{c}(G,c_0, {\sigma[0\ldots k]})} + {\min\{l,m\}}}$.
	\end{clm}
	
	Claim~\ref{clm:lemma-6-revised:Claim3} can be proved by standard induction on $k$, hence we leave the details to the reader.
	The fourth claim is an immediate corollary of Claim~\ref{clm:lemma-6-revised:Claim3}:

	

	
	\begin{clm} 
		\label{lem:lemma-6-revised:Claim4} Assume that $g$ is a strategy for $\text{player}_0$ in $G$ that wins from a set of states, $A \subseteq V$, w.r.t. $c\in \mathbb{N}$ for an initial credit $c_0 \in [0,c]$. Take $l,m\in\mathbb{N}$ such that ${c_{0} + m} \leq {c + l}$,
		and let $\sigma$ be a play from $A$, consistent with $g$. Then, for the initial credit $c_{0}+m$, the energy level under the upper bound $c+l$ never drops below $\min\{l,m\}$ along $\sigma$.
	\end{clm}


	The proof of~\ref{lem:lemma-6-revised:Part1} relies on the well-known notions of \emph{attractors}, \emph{traps}, and \emph{subgames}. Below,
	we repeat their definitions from~\cite[Chapter~6]{2001automata}.
	
	
	\begin{defi}[Attractors]\label{def:lem:lemma-6-revised:attractors}
		The \emph{$\text{player}_{i}$-attractor} $\mathit{Attr}_{i}(X) \subseteq V$ of a set of states $X \subseteq V$, is the set of all states from which $\text{player}_{i}$ has a strategy to force $\text{player}_{1-i}$
		to reach either a state in $X$ or a deadlock for $\text{player}_{1-i}$ in a finite number of steps.
		Formally, ${\mathit{Attr}_{i}(X)} = \bigcup_{j=0}^{+\infty}{A_{j}}$ where $A_{0} = X$ and $A_{j+1} = A_{j} \cup \{s \in V_{i} \mid \exists t\in V ({{(s,t)}\in {E}} \wedge {t\in{A_{j}}})\} \cup \{s \in V_{1-i} \mid \forall t \in V ({{(s,t)}\in {E}} \Rightarrow {t\in{A_{j}}})\}$.
	\end{defi}
	
	\begin{defi}[Traps]\label{def:lem:lemma-6-revised:traps}
		A \emph{$\text{player}_{i}$-trap} is a set of states, $U\subseteq V$, in which $\text{player}_{1-i}$ can trap $\text{player}_{i}$
		in the sense that all successors of $V_{i}$-states in $U$ belong to $U$ and every $V_{1-i}$-state in $U$ has a successor in $U$.
	\end{defi}

	\begin{defi}[Subgames]\label{def:lem:lemma-6-revised:subgames}
		Let $U \subseteq V$. The \emph{subgame} of $G$ induced by $U$ is ${G[U]} = \langle({V[U]}\linebreak={{{(V_0 \cap {U})}\cup {(V_1 \cap {U})}} }, {{E[U]} = E\cap (U \times U))},{\prio|_{U}},{w|_{E[U]}}\rangle$
		where $\prio|_{U}$ and $w|_{E[U]}$ are the restrictions of $\prio$ and $w$ to $U$ and $E[U]$, respectively. 
	\end{defi}
	
	Note that the complement of the $\text{player}_{i}$-attractor of a set $X \subseteq V$, $U = {V} \setminus {\mathit{Attr}_{i}(X)}$,
	is a $\text{player}_{i}$-trap. Also, since $\mathit{Attr}_{i}(W_{i}(+\infty)) = W_{i}(+\infty)$~\cite{ChatterjeeD12},
	it follows from Claim~\ref{clm:lemma-6-revised:Claim1} that $W_{i}(+\infty)$ is a $\text{player}_{1-i}$-trap. Lastly, we consider the following claim:
	
	
	\begin{clm}
		\label{lem:lemma-6-revised:Claim5} Let $G[{W_0(+\infty)}\setminus {\mathit{Attr}_{0}(X)}]$ be a subgame of $G$ where $\mathit{Attr}_{0}(X) \subseteq {W_0(+\infty)}$ is the $\text{player}_0$-attractor in $G[{W_0(+\infty)}]$ of some set, $X \subseteq {W_0(+\infty)}$.
		Then, $\text{player}_0$ wins in $G[{W_0(+\infty)}\setminus {\mathit{Attr}_{0}(X)}]$ from all the states of this subgame w.r.t. $+\infty$.
	\end{clm}
	
	Indeed, ${W_0(+\infty)}\setminus {\mathit{Attr}_{0}(X)}$ is a $\text{player}_0$-trap in $G[{W_0(+\infty)}]$ thus playing in $G[{W_0(+\infty)}\setminus {\mathit{Attr}_{0}(X)}]$ according to a strategy $g$ that wins for $\text{player}_0$ in $G$ from $W_0(+\infty)$, is the same as playing in $G$ from ${W_0(+\infty)}\setminus {\mathit{Attr}_{0}(X)}$ according to $g$ while $\text{player}_1$ always chooses to stay in ${W_0(+\infty)}\setminus {\mathit{Attr}_{0}(X)}$.
	
	
	Relying on what we have established thus far, we can finally prove~\ref{lem:lemma-6-revised:Part1} and~\ref{lem:lemma-6-revised:Part2}, which clearly imply the correctness of  Lem.~\ref{lem:lemma-6-revised}.
	
	\begin{proof}[Proof of Lem.~\ref{lem:lemma-6-revised}]\hfill
	
	\noindent\textbf{Proof of~\ref{lem:lemma-6-revised:Part1}.} The proof is by induction on $n+d$. Note that if $d=1$, then for all $n >0$, a good-for-energy strategy (which exists by Claim~\ref{clm:lemma-6-revised:Claim2}) wins from $W_0(+\infty)$ w.r.t. $(n-1)K$ for $(n-1)K$, as required.
	This proves the base case where ${n}={d}={1}$. 
	Also, this allows us to assume that $d>1$ from now on.

	For the induction step, we distinguish between two cases. The case where the minimal priority is even (say, $0$) and the case where it is odd (say, $1$). We assume w.l.o.g. that $V=W_0(+\infty)$. This assumption can be made since, otherwise, $|W_0(+\infty)|<n$, and we can simply apply the induction hypothesis over the subgame induced by $W_0(+\infty)$ and obtain a strategy as required in $G$.\footnote{Notice that a strategy that wins for $\text{player}_{0}$ from ${V[W_0(+\infty)]} = {W_0(+\infty)}$ in the subgame $G[W_0(+\infty)]$ does the same in $G$ because $W_0(+\infty)$ is a $\text{player}_{1}$-trap.}
	
	
	\noindent{\bfseries Case 1: The minimal priority is 0.} Let $\gfeStrategy$ be a good-for-energy strategy in $G$, which exists due to Claim~\ref{clm:lemma-6-revised:Claim2}. Let $\Omega_0\subseteq W_0(+\infty)=V$ be the $\text{player}_0$-attractor of all $0$-priority states in $G$.
	Write $|\Omega_0|=k$ and note that $\Omega_0\neq \emptyset$ because $\Omega_0$ includes all $0$-priority states in $V$. Consider the subgame $G' = G[{W_0(+\infty)}\setminus {\Omega_{0}}]$, which has at most $d-1$ different priorities.
	By Claim~\ref{lem:lemma-6-revised:Claim5}, all the states of $G'$ win for $\text{player}_0$ in this subgame w.r.t. $+\infty$.
	Therefore, by the induction hypothesis, $\text{player}_0$ has a strategy $g'$ in $G'$ that wins from ${W_0(+\infty)}\setminus {\Omega_{0}}$ w.r.t. $c'=(d-1)(n-k-1)K$ for the initial credit $c_0'=(n-k-1)K$.
	We argue that the next strategy satisfies the requirements.
	
	\begin{strategy}\label{def:lem:lemma-6-revised:minPrioZeroStrategy} We define a strategy for $\text{player}_{0}$ in $G$ from $W_0(+\infty)$, as follows:
		
		\begin{enumerate}[label=\textbf{Phase~\ref{def:lem:lemma-6-revised:minPrioZeroStrategy}.\arabic*.},ref={Phase~\ref{def:lem:lemma-6-revised:minPrioZeroStrategy}.\arabic*},itemindent=*]
			\item \label{lem:lemma-6-revised:Stage1} Play according to $\gfeStrategy$ until the sum of the edges' weights traversed is at least $d(n-1)K-(n-1)K$. If you reached a state in $\Omega_0$, go to~\ref{lem:lemma-6-revised:Stage2}. Otherwise, go to~\ref{lem:lemma-6-revised:Stage3}.
			
			\item \label{lem:lemma-6-revised:Stage2} Play a strategy to reach a $0$-priority state, and go to~\ref{lem:lemma-6-revised:Stage1}.
			
			\item \label{lem:lemma-6-revised:Stage3} Play according to $g'$ as long as the play stays in $W_0(+\infty)\setminus\Omega_0$. If the play reaches $\Omega_0$, go to~\ref{lem:lemma-6-revised:Stage2}.
		\end{enumerate}
	\end{strategy}
	
	Note that for a play consistent with Strategy~\ref{def:lem:lemma-6-revised:minPrioZeroStrategy}, either of the following holds: (a) the play reaches a deadlock for $\text{player}_{1}$; (b) eventually, the play stays in~\ref{lem:lemma-6-revised:Stage1} forever; (c) eventually, the play stays in~\ref{lem:lemma-6-revised:Stage3} forever; (d) the play reaches~\ref{lem:lemma-6-revised:Stage2} infinitely many times. In all of these cases, $\text{player}_0$ wins the parity objective, either by definition (case a), by the choice of the strategies (cases b, c), or by visiting a $0$-priority state infinitely often (case d). Thus, it is left to show that $\text{player}_{0}$ also wins the energy objective.
	
	Consider a play $\sigma$ consistent with Strategy~\ref{def:lem:lemma-6-revised:minPrioZeroStrategy} and played w.r.t. the upper bound $c=d(n-1)K$ and for the initial credit $(n-1)K$. We shall prove that the energy level always remains non-negative along $\sigma$.
	Let $\sigma[k_0=0],\sigma[k_1],\sigma[k_2],\dots$ be the states along $\sigma$ in which Strategy~\ref{def:lem:lemma-6-revised:minPrioZeroStrategy} turns to~\ref{lem:lemma-6-revised:Stage1}. We prove the next properties by induction on $j$:
	\begin{enumerate} 
		\item \label{lem:lemma-6-revised:minPrioZeroStrategy:property1} The energy level never decreases below $0$ in the interval $\sigma[0\dots k_j]$.
		
		\item \label{lem:lemma-6-revised:minPrioZeroStrategy:property2} $\textsf{EL}_{c}(G, (n-1)K, \sigma[0\dots k_j])\geq (n-1)K$.
		
	\end{enumerate} 
	
	\noindent Note that both properties~\ref{lem:lemma-6-revised:minPrioZeroStrategy:property1} and~\ref{lem:lemma-6-revised:minPrioZeroStrategy:property2} trivially hold for ${k_{0}}$. For the induction step, consider the interval $\sigma[k_j\dots k_{j+1}]$. From $\sigma[k_j]$, $\text{player}_0$ plays according to $\gfeStrategy$. By the induction hypothesis over $j$, $\text{player}_0$ reaches the state $\sigma[k_j]$ while having at least $(n-1)K$ energy units, and, consequently, by Claim~\ref{clm:lemma-6-revised:Claim2}, the energy level remains non-negative until $\text{player}_0$ proceeds with $d(n-1)K$ energy units to either~\ref{lem:lemma-6-revised:Stage2} or~\ref{lem:lemma-6-revised:Stage3}. If $\text{player}_0$ goes to~\ref{lem:lemma-6-revised:Stage2}, as $|\Omega_0|=k$, she spends at most $(k-1)K$ energy units to reach a $0$-priority state and returns to~\ref{lem:lemma-6-revised:Stage1} at step $k_{j+1}$ with at least $d(n-1)K-(k-1)K$ energy units.
	Since $d>1$ and $n \geq k$, it follows that both properties~\ref{lem:lemma-6-revised:minPrioZeroStrategy:property1} and~\ref{lem:lemma-6-revised:minPrioZeroStrategy:property2} hold in this scenario.

	Otherwise, we have that $\text{player}_0$ goes to~\ref{lem:lemma-6-revised:Stage3}. Hence, ${W_0(+\infty)}\setminus {\Omega_{0}} \not = {\emptyset}$, ${|W_0(+\infty)|} = n > {{|\Omega_0|} = {k}}$, and $\text{player}_0$ plays according to the strategy $g'$ with an upper bound and an initial credit both equal to $d(n-1)K$.
	However, the induction hypothesis ensures that $g'$ wins for $\text{player}_0$ in the subgame $G'$ w.r.t. $c'$ for $c'_{0}$.
	Thus, it follows from Claim~\ref{lem:lemma-6-revised:Claim4} that the energy level never drops below $\min\{d(n-1)K-c',d(n-1)K-c_0'\}=(n-1)K+(d-1)kK$ as long as the play stays in ${{W_0(+\infty)}\setminus {\Omega_{0}}}$.
	Note that the play can leave ${{W_0(+\infty)}\setminus {\Omega_{0}}}$ only by traversing through an edge $e'$ that is chosen by $\text{player}_{1}$. 
	When that occurs, $\text{player}_0$ loses at most $K$ energy units (as $w(e') \geq -K$) and switches to~\ref{lem:lemma-6-revised:Stage2}. As in the previous scenario,
	$\text{player}_0$ spends in~\ref{lem:lemma-6-revised:Stage2} at most $(k-1)K$ energy units to reach a $0$-priority state and returns to~\ref{lem:lemma-6-revised:Stage1} at step $k_{j+1}$.
	Therefore, we have that $\textsf{EL}_c(G,(n-1)K,\sigma[0\dots k_{j+1}])\geq(n-1)K+(d-1)kK-K-(k-1)K\geq (n-1)K$. This implies that properties~\ref{lem:lemma-6-revised:minPrioZeroStrategy:property1} and~\ref{lem:lemma-6-revised:minPrioZeroStrategy:property2} hold in this scenario as well.
	
	Consequently, if Strategy~\ref{def:lem:lemma-6-revised:minPrioZeroStrategy} turns to~\ref{lem:lemma-6-revised:Stage1} infinitely many times, $\sigma$ wins the energy objective. Otherwise, there is some step $k_l$, such that:
	\begin{itemize}
		\item The strategy turns to~\ref{lem:lemma-6-revised:Stage1} for the last time in $\sigma[k_l]$.
		\item The energy level never drops below $0$ in $\sigma[0\dots k_l]$.
		\item The energy level of $\sigma[0\dots k_l]$ is at least $(n-1)K$. 
	\end{itemize}
	

	\noindent From $\sigma[k_l]$, $\text{player}_0$ plays according to $\gfeStrategy$ in $W_0(+\infty)$. By Claim~\ref{clm:lemma-6-revised:Claim2}, as long as $\text{player}_0$ plays according to $\gfeStrategy$, the energy level remains non-negative.
	Thus, if that lasts forever, the energy objective is achieved. Otherwise, there is some step $l' > k_{l}$ in which the strategy turns to~\ref{lem:lemma-6-revised:Stage3} with the initial credit ${\textsf{EL}_{c}(G, (n-1)K, \sigma[0\dots l'])} =  {d(n-1)K}$,
	and stays in this phase forever. Hence, from $\sigma[l'] \in {{W_0(+\infty)}\setminus {\Omega_0}}$, the play remains in $W_0(+\infty)\setminus \Omega_0$ and played according to $g'$.
	Consequently, it follows from induction hypothesis on $g'$ that the energy objective is achieved in this case as well.
	


	\noindent{\bfseries Case 2: The minimal priority is 1.} Let $D_{1} \subseteq V_{1}$ be the set of all states in $G$ which are deadlocks for $\text{player}_{1}$.
	
	First, consider the case where $D_{1} \not = \emptyset$. Let $\Omega^{D_{1}}_{0}$ be the $\text{player}_0$-attractor of $D_{1}$ in $G$.
	Note that $D_{1} \subseteq \Omega^{D_{1}}_{0} \not = \emptyset$, write $|\Omega^{D_{1}}_{0}| = k_{D_{1}}$, and
	consider the subgame $G'' = G[W_0(+\infty)\setminus \Omega^{D_{1}}_{0}]$ induced by ${V \setminus \Omega^{D_{1}}_{0}} = {W_0(+\infty)\setminus \Omega^{D_{1}}_{0}}$.
	By Claim~\ref{lem:lemma-6-revised:Claim5}, $\text{player}_0$ wins in $G''$ from all the states of this subgame w.r.t. $+\infty$.
	Hence, the induction hypothesis yields a strategy $g''$ for $\text{player}_0$ in $G''$ that wins from $W_0(+\infty)\setminus \Omega^{D_{1}}_{0}$
	w.r.t. $c''= d(n-k_{D_{1}}-1)K$ for the initial credit $c_0''=(n-k_{D_{1}}-1)K$. We claim that the next strategy satisfies the requirements.

	
	\begin{strategy}\label{def:lem:lemma-6-revised:minPrioOneFirstCaseStrategy}
		We define a strategy for $\text{player}_0$ in $G$ from $W_0(+\infty)$, as follows:
		
		\begin{enumerate}[label=\textbf{Phase~\ref{def:lem:lemma-6-revised:minPrioOneFirstCaseStrategy}.\arabic*.},ref={Phase~\ref{def:lem:lemma-6-revised:minPrioOneFirstCaseStrategy}.\arabic*},itemindent=*]
			\item \label{lem:lemma-6-revised:Case2:Stage-a} If the play is in $\Omega^{D_{1}}_{0}$, go to~\ref{lem:lemma-6-revised:Case2:Stage-b}. Otherwise, go to~\ref{lem:lemma-6-revised:Case2:Stage-c}.
			
			\item \label{lem:lemma-6-revised:Case2:Stage-b} Play a strategy to reach a deadlock for $\text{player}_{1}$ (i.e., a state in $D_{1}$).
			
			\item \label{lem:lemma-6-revised:Case2:Stage-c} Play according to $g''$ as long as the play stays in $W_0(+\infty)\setminus \Omega^{D_{1}}_{0}$. If the play reaches $\Omega^{D_{1}}_{0}$, go to~\ref{lem:lemma-6-revised:Case2:Stage-b}.
		\end{enumerate}
	\end{strategy}
	
	Consider a play $\sigma$ consistent with Strategy~\ref{def:lem:lemma-6-revised:minPrioOneFirstCaseStrategy} and played w.r.t. $c=d(n-1)K$ and for the initial credit $(n-1)K$.
	Then, $\sigma$ either (1) stays in~\ref{lem:lemma-6-revised:Case2:Stage-c} forever, or (2) eventually reaches~\ref{lem:lemma-6-revised:Case2:Stage-b} and subsequently ends in a deadlock for $\text{player}_{1}$.
	In case (1), $\sigma$ is infinite and consistent with $g''$, and as a result, it wins for $\text{player}_0$ in $G$. In case (2), $\sigma$ is finite and wins the parity objective by definition.
	Moreover, we argue that $\sigma$ wins the energy objective in case (2) as well. If $\sigma$ starts from $\Omega^{D_{1}}_{0}$, then, as $|\Omega^{D_{1}}_{0}|= k_{D_{1}}$, $\text{player}_0$ spends at most $(k_{D_{1}}-1)K$ energy units in~\ref{lem:lemma-6-revised:Case2:Stage-b}
	to enforce reaching a deadlock state for $\text{player}_1$. Thus, since $n \geq k_{D_{1}}$, $\sigma$ wins the energy objective.
	Otherwise, $\sigma$ starts from $W_0(+\infty)\setminus \Omega^{D_{1}}_{0}$.
	As long as $\sigma$ stays in $W_0(+\infty)\setminus \Omega^{D_{1}}_{0}$, $\text{player}_0$ plays according to $g''$ (\ref{lem:lemma-6-revised:Case2:Stage-c}) with the upper bound $c$ and the initial credit $(n-1)K$.
	However, recall that the induction hypothesis ensures that $g''$ wins w.r.t. $c''$ for $c_0''$. Thus, it follows from Claim~\ref{lem:lemma-6-revised:Claim4} that
	the energy level during~\ref{lem:lemma-6-revised:Case2:Stage-c} never drops below $\min\{c-c'',(n-1)K-c_0''\}= \min\{dk_{D_{1}}K,k_{D_{1}}K\} = k_{D_{1}}K$. 
	As $\sigma$ leaves $W_0(+\infty)\setminus \Omega^{D_{1}}_{0}$ by traversing through an edge that costs at most $K$ energy units,
	$\text{player}_0$ reaches $\Omega^{D_{1}}_{0}$ with an initial credit at least ${({k_{D_{1}}}-{1})}K$.
	This initial credit is sufficient for winning in~\ref{lem:lemma-6-revised:Case2:Stage-b}.

	Second, consider the remaining case where $D_{1} = \emptyset$.
	Let ${\Omega_1}$ be the $\text{player}_1$-attractor of all $1$-priority states in $G$, let $G'=G[{{W_0(+\infty)}\setminus {\Omega_1}}]$ be the subgame induced by $W_0(+\infty)\setminus \Omega_1$, and let $W'$ be the winning region of $\text{player}_{0}$ w.r.t. $+\infty$ in $G'$, i.e., $W' = W_0^{G'}(+\infty)$.
	We claim that $W' \neq \emptyset$.
	Suppose, towards contradiction, that this claim is false. Then, it follows from Claim~\ref{clm:lemma-6-revised:Claim1} that $\text{player}_1$ has a strategy $h'$ in $G'$ that wins for him from all the states of this subgame w.r.t. $+\infty$.
	
	
	\begin{strategy}\label{def:lem:lemma-6-revised:Case2:CounterStrategy}
	Consider the following strategy for $\text{player}_1$ in $G$ from $W_0(+\infty)$:
	
	\begin{enumerate}[label=\textbf{Phase~\ref{def:lem:lemma-6-revised:Case2:CounterStrategy}.\arabic*.},ref={Phase~\ref{def:lem:lemma-6-revised:Case2:CounterStrategy}.\arabic*},itemindent=*]
		\item \label{lem:lemma-6-revised:Case2:CounterStrategy:Stage-0} If the play is in $\Omega_1$, go to~\ref{lem:lemma-6-revised:Case2:CounterStrategy:Stage-1}. Otherwise, go to~\ref{lem:lemma-6-revised:Case2:CounterStrategy:Stage-2}.
		\item \label{lem:lemma-6-revised:Case2:CounterStrategy:Stage-1} If the current state is a $1$-priority state, choose any successor; otherwise, play a strategy to reach a $1$-priority state. Go to~\ref{lem:lemma-6-revised:Case2:CounterStrategy:Stage-0}.
		\item \label{lem:lemma-6-revised:Case2:CounterStrategy:Stage-2} As long as the play stays in ${{W_0(+\infty)}\setminus {\Omega_1}}$, play according to $h'$. If the play reaches $\Omega_1$, go to~\ref{lem:lemma-6-revised:Case2:CounterStrategy:Stage-1}.
	\end{enumerate}
	
	\end{strategy}
	
	Note that Strategy~\ref{def:lem:lemma-6-revised:Case2:CounterStrategy} is well-defined. That is, since $D_{1} = \emptyset$, there are no deadlock states in $V = W_0(+\infty)$, and consequently, there always exists a successor state that $\text{player}_{1}$ can choose in~\ref{lem:lemma-6-revised:Case2:CounterStrategy:Stage-1}.
	Every play $\sigma$ consistent with this strategy, either visits~\ref{lem:lemma-6-revised:Case2:CounterStrategy:Stage-1} infinitely often or eventually stays in~\ref{lem:lemma-6-revised:Case2:CounterStrategy:Stage-2}.
	In the former case, $\sigma$ visits $1$-priority states infinitely often and thus violates the parity objective, while in the latter case, $\sigma$ wins for $\text{player}_{1}$ due to the strategy $h'$. This contradicts that $\text{player}_{0}$ wins in $G$ from $W_0(+\infty)$.
	
	Let $|W'|=k$. Notice that the subgame $G[W']=G'[W']$ has at most $d-1$ different priorities, $\Omega_{1} \neq \emptyset$, and $k < n$.
	It is not difficult to see that $\text{player}_0$ wins in $G[W']$ from all the states of this subgame w.r.t. $+\infty$.
	As a result, the induction hypothesis yields a strategy $g^{W'}$ in $G[W']$ that wins for $\text{player}_0$ from $W'$ w.r.t. $(d-1)(k-1)K$ for the initial credit $(k-1)K$.
	Moreover, the facts that $W'$ is a $\text{player}_{1}$-trap in $G'$ and $W_0(+\infty)\setminus \Omega_1$ is a $\text{player}_{1}$-trap in $G$,
	imply that $W'$ is also a $\text{player}_{1}$-trap in $G$. Therefore, $g^{W'}$ is also a strategy that wins for $\text{player}_0$ from $W'$ in $G$.
	
	Let $\Omega_{0}^{W'} = \mathit{Attr}_{0}(W')$ be the $\text{player}_0$-attractor of $W'$ in $G$, and let ${|\Omega_{0}^{W'}|} = {k + m} $.
	Consider the subgame $H = G[{W_0(+\infty)} \setminus {\Omega_{0}^{W'}}]$. 
	By Claim~\ref{lem:lemma-6-revised:Claim5}, $\text{player}_0$ wins in $H$ from all the states of this subgame w.r.t. $+\infty$.
	Thus, the induction hypothesis yields a strategy $h$ in $H$ that wins for $\text{player}_0$ from ${W_0(+\infty)} \setminus {\Omega_{0}^{W'}}$
	w.r.t. $c^{H}=d(n-k-m-1)K$ for $c^{H}_{0} = (n-k-m-1)K$. We claim that the next strategy satisfies the requirements.

	\begin{strategy}\label{def:lem:lemma-6-revised:minPrioOneSecondCaseStrategy}
		We define a strategy for $\text{player}_0$ in $G$ from $W_0(+\infty)$, as follows:
		
		\begin{enumerate}[label=\textbf{Phase~\ref{def:lem:lemma-6-revised:minPrioOneSecondCaseStrategy}.\arabic*.},ref={Phase~\ref{def:lem:lemma-6-revised:minPrioOneSecondCaseStrategy}.\arabic*},itemindent=*]
			\item \label{lem:lemma-6-revised:Case2:Stage-A} As long as the play stays in ${W_0(+\infty)} \setminus {\Omega_{0}^{W'}}$, play according to $h$. If the play reaches $\Omega_{0}^{W'}$, go to~\ref{lem:lemma-6-revised:Case2:Stage-B}.
			
			\item \label{lem:lemma-6-revised:Case2:Stage-B} Play a strategy to reach $W'$, and then play according to $g^{W'}$.
			
		\end{enumerate}
	\end{strategy}

	%
	%
	%
	
	
	Let us show that Strategy~\ref{def:lem:lemma-6-revised:minPrioOneSecondCaseStrategy} wins w.r.t. $c=d(n-1)K$ for the initial credit $(n-1)K$.
	Consider a play $\sigma$ consistent with Strategy~\ref{def:lem:lemma-6-revised:minPrioOneSecondCaseStrategy}.
	If $\sigma$ stays in~\ref{lem:lemma-6-revised:Case2:Stage-A} forever, then it is consistent with $h$, and as a result, it wins by the induction hypothesis on $h$.
	Otherwise, $\sigma$ eventually reaches~\ref{lem:lemma-6-revised:Case2:Stage-B}.
	The induction hypothesis on $h$ and Claim~\ref{lem:lemma-6-revised:Claim4} imply that
	the energy level never drops below $\min\{ c-c^{H}, (n-1)K-c^{H}_{0} \}=(k+m)K$ as long as the play stays
	in ${W_0(+\infty)} \setminus {\Omega_{0}^{W'}}$ (i.e.,~\ref{lem:lemma-6-revised:Case2:Stage-A}).
	Hence, when Strategy~\ref{def:lem:lemma-6-revised:minPrioOneSecondCaseStrategy} turns to~\ref{lem:lemma-6-revised:Case2:Stage-B}, the energy level is at least $(k+m-1)K$.
	This holds as $\sigma$ leaves ${W_0(+\infty)} \setminus {\Omega_{0}^{W'}}$ by traversing
	through an edge that costs at most $K$ energy units. In~\ref{lem:lemma-6-revised:Case2:Stage-B}, $\text{player}_{0}$ spends at most $mK$ energy units 
	to reach $W'$ and subsequently starts playing according to $g^{W'}$ with an initial credit at least ${{(k+m-1)K} - {mK}} = {(k-1)K}$.   
	Therefore, the induction hypothesis on $g^{W'}$ guarantees that $\sigma$ wins as required.


	
	\noindent\textbf{Proof of~\ref{lem:lemma-6-revised:Part2}.} Assume that $\text{player}_0$ wins from a state $s \in V$ w.r.t. $+\infty$ for an initial credit of $c_0$. We show that $\text{player}_0$ also wins from $s$ w.r.t. $c = d(n-1)K$ for the initial credit $\min\{c_0,(n-1)K\}$.
	If $(n-1)K\leq c_0$, the claim follows from~\ref{lem:lemma-6-revised:Part1}. Otherwise, $c_0<(n-1)K$, and we set the following strategy that wins from $s$ w.r.t. $c$ for $c_{0}$.
	Let $g_{+\infty}$ be a strategy for $\text{player}_0$ that wins from $s$ w.r.t. $+\infty$ for the initial 
	credit $c_0$. Initially, $\text{player}_0$ plays according to $g_{+\infty}$ and keeps playing according to this strategy as long as $\textsf{EL}_{+\infty}(G,c_0,\sigma[0\dots t])<(n-1)K$, where $\sigma[0 \dots t]$ is the sequence of states traversed so far. If, at some point, $\textsf{EL}_{+\infty}(G,c_0,\sigma[0\dots t])\geq (n-1)K$, $\text{player}_0$ switches to a strategy that satisfies the requirements of~\ref{lem:lemma-6-revised:Part1}.
	Clearly, all the states traversed while playing according to $g_{+\infty}$ belong to $W_0(+\infty)$, and thus, if necessary, $\text{player}_0$ can always switch to a strategy that exists due to~\ref{lem:lemma-6-revised:Part1}.
	It is not difficult to see that the strategy we have described wins for $\text{player}_0$ as required. 
\end{proof}

\section{Extended Version of Sect.~\ref{sec:sufficientbound:WGStoParityEnergyGames}}\label{app:sufficient-bound-proof}

In this appendix, we present the full proof for the main result of Sect.~\ref{sec:sufficientbound}:

\begin{quote}
\textbf{Theorem.~\ref{Thm:a-sufficient-bound}.}
		Let $\WGS$ be a WGS, $N=|2^{\mathcal{V}}|$, and let $K$ 
		be the maximal transition weight in $G^w$, in absolute value. Take $\psi\in \sysFormulas$,
		a closed $\sysmu$ formula that matches $\varphi$, and let $m$ be its length
		and $d$ its alternation depth. Then, if the system wins from a state $s$ w.r.t. 
		$+\infty$ for an initial credit $c_0$, then it also wins from $s$  
		w.r.t. $(d+1)((N^2+N)m-1)K$ 
		for an initial credit 
		$\min\{c_0,((N^2+N)m-1)K\}$.
\end{quote}

\noindent We fix a WGS $\WGS$ and a closed $\sysmu$ formula $\psi\in \sysFormulas$ that matches the winning condition of $G^{w}$, $\varphi$, as in Thm.~\ref{Thm:a-sufficient-bound}. Recall that the reduction to energy parity games, outlined in Sect.~\ref{sec:sufficientbound:WGStoParityEnergyGames}, involved the construction of several game graphs.
Throughout this appendix, whenever we define a game graph $H$ (i.e., component (1) in Def.~\ref{def:energyParityGame}), the terms $V(H)$ and $E(H)$ denote the set of states and edges of $H$, respectively.

Take a finite upper bound $c \in \mathbb{N}$ and construct the (symbolic) GS $G^{*} = \langle \mathcal{V}^{*}, \mathcal{X},
\mathcal{Y}^{*}, \rho^{e},\rho^{s*}, \linebreak\varphi \rangle$ from $G^{w}$ and $c$, as defined in Def.~\ref{def:naiveReduction}.
We transform $G^{*}$ into an (explicit, bipartite) game graph $G_c$ by adding intermediate states to distinguish between steps performed by the environment and the system players.

\begin{defi}[The game graph $G_c$]\label{def:concreteGstar}
	Let $G_c=(V=V_0 \cup V_1,E)$ where
	
	\begin{itemize}
		\item $V_0 = 2^\mathcal{V}\times2^\mathcal{X}\times\{0,\dots,c\}$ and $V_1 = 2^\mathcal{V}\times\{0,\dots,c\}$.
		
		\item For $(s,c_1)\in V_1$ and $(s,u,c_1)\in V_0$, $((s,c_1),(s,u,c_1))\in E$ if $(s,p(u))\models \rho^e$. 
		
		\item For $(s,u,c_1)\in V_0$ and $(t,c_2)\in V_1$, $((s,u,c_1),(t,c_2))\in E$ if $u=t|_\mathcal{X}$, $(s,p(t))\models \rho^s$, and $\min\{c,c_1 + w^s(s,p(t))\} \geq c_2$.
	\end{itemize}  
	
\end{defi}

The game graph $G_c$, defined in Def.~\ref{def:concreteGstar}, simulates the GS $G^*$. That is, edges of the form $((s,c_1),(s,u,c_1))$, chosen by $\text{player}_{1}$,
correspond to environment's transitions, $((s,c_1),p(u))\models \rho^e$, while edges of the form $((s,u,c_1),(t,c_2))$, chosen by $\text{player}_{0}$,
correspond to system's transitions, $((s, c_1), p(t, c_2))\models \rho^{s*}$ where $u=t|_\mathcal{X}$.\footnote{Note that, throughout this appendix, we use the state $(s,c_{0}) \in 2^\mathcal V\times\{0,\dots,c\}$ in $G_{c}$ interchangeably with the corresponding state in $G^{*}$, $(s,c_0)\in 2^{\mathcal{V}^{*}}$.}
Thus, each play in $G^{*}$ corresponds to a play in $G_c$ during which the players take steps in an alternating manner, and vice versa.

In order to interpret the closed $\sysmu$ formula $\psi$ over the graph $G_c$, rather than $G^{*}$,
we \emph{split} the controllable predecessor operator $\circlediamond$ 
by replacing every sub-formula of the form $\circlediamond \beta$ in $\psi$ with
 $\square \bigdiamond \beta$.
The symbols $\square$ and 
$\bigdiamond$ 
denote the \emph{classical $\mu$-calculus predecessor operators}~\cite{2001automata,Kozen}. Their semantics
is defined w.r.t. a graph $H=(V,E)$ and a valuation ${\mathcal{E}} : {{\mathit{Var}}\rightarrow {(V \rightarrow \{0,1\})}}$, as follows:
\begin{itemize}
	
	\item $\llbracket \square\beta \rrbracket^H_\mathcal E=\{v\in V \mid \forall u\in V((v,u)\in E \Rightarrow u\in \llbracket \beta \rrbracket^H_\mathcal E) \}$.
	
	\item $\llbracket \bigdiamond\beta \rrbracket^H_\mathcal E=\{v\in V \mid \exists u\in V((v,u)\in E \wedge u\in \llbracket \beta \rrbracket^H_\mathcal E) \}$.
	
\end{itemize}  
For simplicity, although we obtain a formula with a different syntax, we use $\psi$ to denote the translated formula as well. 
We argue that the graph $G_c$, like the GS $G^*$, simulates the WGS $G^w$. The next lemma formally captures this claim.

\begin{lem}
	\label{lem:G^w-equiv-Gc}
	The system wins in $G^w$ from $s \in 2^{\mathcal V}$ w.r.t. $c$ for an initial credit $c_0 \leq c$, if and only if $(s,c_0)\in \llbracket \psi \rrbracket^{G_c}$. 
\end{lem}

\begin{proof}
	Relying on Thm.~\ref{thm:reductionCorrectness}, it is sufficient to show that ${(s,c_0)\in \llbracket \psi \rrbracket ^{G^*}} \Leftrightarrow {(s,c_0)\in \llbracket \psi \rrbracket ^{G_c}}$.
	This claim is implied by the following generalized statement.
	We say that a valuation $\mathcal{E}'$ over $G_c$ \emph{extends} a valuation $\mathcal{E}$ over $G^*$ if for every relational variable, $X \in \mathit{Var}$, and every state, $(s,c_0)\in 2^\mathcal V\times\{0,\dots,c\}$, it holds that $(s,c_0)\in \mathcal{E}(X) \Leftrightarrow (s,c_{0})\in\mathcal{E}'(X)$. 
	
	\begin{clm}
	For $\beta\in \sysFormulas$, $(s,c_{0})\in 2^\mathcal V\times\{0,\dots,c\}$, a valuation $\mathcal{E}$ over $G^*$, and a valuation $\mathcal{E}'$ over $G_c$ that extends $\mathcal{E}$, $(s,c_{0})\in \llbracket \beta \rrbracket ^{G^*}_\mathcal{E} \Leftrightarrow  (s,c_{0})\in \llbracket \beta \rrbracket^{G_c}_{\mathcal{E}'}$. 
	\end{clm} 
	Before proving the above claim, we clarify that $\beta$ is used here to denote two similar yet slightly different formulae. That is, in $\llbracket \beta \rrbracket^{G_c}_{\mathcal{E}'}$, the operators $\square\bigdiamond$ replace every occurrence of $\circlediamond$ in $\llbracket \beta \rrbracket ^{G^*}_\mathcal{E}$. The statement is proved by structural induction on $\beta$, where the cases $\beta=v,\neg v, X,\beta_1\wedge \beta_2, \beta_1\vee \beta_2$ are simple. Thus, we focus on the cases where $\beta=\circlediamond\beta'$ and $\beta=\eta X(\beta')$ for $\eta\in\{\mu,\nu\}$.
	\begin{description}
		\item[$\beta=\circlediamond\beta'$]  
		We have that
		\begin{align*}
		&(s,c_{0})\in \llbracket \circlediamond \beta'\rrbracket^{G^*}_{\mathcal{E}}
    \Leftrightarrow\text{(by Def.~\ref{def:naiveReduction})}\\
    &{\forall t_\mathcal X\in 2^\mathcal X} \big[{ (s,p(t_\mathcal X))\models \rho^e} \Rightarrow 
    {\exists t_\mathcal Y\in 2^\mathcal Y \exists c_{1}\leq c :
      [ ((s,p(t=t_\mathcal X\cup t_\mathcal Y))\models \rho^s)} \wedge \\
    &{(c_{0}+w^s(s,p(t)) \geq c_{1})} \wedge {(t,c_{1})\in \llbracket \beta'\rrbracket ^{G^*}_{\mathcal E}  ]}\big] \Leftrightarrow\text{(by the induction hypothesis)}\\
    &{\forall t_\mathcal X\in 2^\mathcal X} \big[{ (s,p(t_\mathcal X))\models \rho^e} \Rightarrow 
    {\exists t_\mathcal Y\in 2^\mathcal Y \exists c_{1}\leq c } : 
    [ ((s,p(t=t_\mathcal X\cup t_\mathcal Y))\models \rho^s) \wedge\\
    &{(c_{0}+w^s(s,p(t)) \geq c_{1})} \wedge {(t,c_{1})\in \llbracket \beta'\rrbracket ^{G_c}_{\mathcal E'} ]}\big] \Leftrightarrow\text{(by Def.~\ref{def:concreteGstar})}\\
    &{\forall t_\mathcal X\in 2^\mathcal X} \big[{ (s,p(t_\mathcal X))\models \rho^e} \Rightarrow
    {\exists t_\mathcal Y\in 2^\mathcal Y \exists c_{1}\leq c } : [ {((s, t_\mathcal X, c_{0} ),(t=t_\mathcal X\cup t_\mathcal Y,c_{1}))} \in {E(G_c)} \wedge\\
    &{(t,c_{1})\in \llbracket \beta'\rrbracket ^{G_c}_{\mathcal E'}}  ]\big] \Leftrightarrow\text{(by Def.~\ref{def:concreteGstar})}\\
    &{\forall t_\mathcal X \in 2^\mathcal X} \big[ ((s,c_{0}),(s,t_\mathcal X,c_{0}))\in E(G_c)\Rightarrow (s,t_\mathcal X,c_{0})\in\llbracket \bigdiamond \beta'\rrbracket^{G_c}_{\mathcal E'}\big] \Leftrightarrow
    (s,c_{0})\in \llbracket \square\bigdiamond \beta'\rrbracket^{G_c}_{\mathcal{E}'}.&
		\end{align*}

		\item[$\beta=\eta X(\beta')$] We prove only for the case where $\eta=\mu$ as the case of the greatest fixed point (i.e., $\eta=\nu$) is dealt similarly. Write ${\llbracket \mu X(\beta')\rrbracket^{G^*}_{\mathcal E}}={\bigcup_{i=0}^{+\infty} S_i }$ and ${\llbracket \mu X(\beta')\rrbracket^{G_c}_{\mathcal E'}}={\bigcup_{i=0}^{+\infty} S'_i}$, as in Def.~\ref{def:prop_mu_calculus_semantics}. We show by induction that for every $i$, $(s,c_{0})\in S_i\Leftrightarrow (s,c_{0})\in S_i'$. This holds trivially for $i=0$ as $S_0=S_0'=\emptyset$. 
		
		For the induction step, consider the sets $S_{i+1}=\llbracket \beta'\rrbracket^{G^*}_{\mathcal E[X \mapsto S_i]}$ and $S'_{i+1}=\llbracket \beta'\rrbracket^{G_c}_{\mathcal E'[X \mapsto S_i']}$. By  applying the induction hypothesis over $S_i$, since $\mathcal E'$ extends $\mathcal E$, we conclude that $\mathcal E'[X \mapsto S'_i]$ extends $\mathcal E[X \mapsto S_i]$. Therefore, the structural induction hypothesis ensures that $(s, c_{0})\in S_{i+1} \Leftrightarrow (s,c_{0})\in S_{i+1}'$. Consequently, $(s,c_{0})\in \bigcup_{i=0}^{+\infty} S_i \Leftrightarrow(s,c_{0})\in \bigcup_{i=0}^{+\infty} S'_i$, as required.\qedhere 
	\end{description}
\end{proof}

\noindent Equipped with the (explicit) game graph $G_c$, which we defined in Def.~\ref{def:concreteGstar} and showed in Lem.~\ref{lem:G^w-equiv-Gc} to simulate the WGS $G^w$, we can now invoke the seminal reduction from model-checking of
$\mu$-calculus formulae to parity games~\cite{EmersonJ91}. Below, Def.~\ref{def:modelCheckingGame} applies this reduction to $G_c$ and the formula $\psi$.
We also refer the reader to~\cite[Chapter~10]{2001automata} on which the reduction we present is based\footnote{Note that we consider $\min$-even parity winning conditions while~\cite[Chapter~10]{2001automata} considers $\max$-even ones. Accordingly, the priority function defined in Def.~\ref{def:modelCheckingGame} is obtained from that in~\cite[Chapter~10]{2001automata} by inverting the order on the priorities.}.
Def.~\ref{def:modelCheckingGame} assumes that all relational variables in $\psi$ are quantified by fixed-point operators exactly once. If that is not the case, the variables can be renamed to fulfil this requirement, without affecting  the formula's semantics.

\begin{defi}[The parity game $G_c\times \psi$]\label{def:modelCheckingGame}
	Let $G_c\times \psi= \langle (V=V_0\cup V_1,E),\prio \rangle$ be the parity game (cf. Def.~\ref{def:energyParityGame}) defined by: 
	\begin{itemize}
		
		\item $V=\{(S,\psi') \mid S\in V(G_c)$, $\psi' \text{ is a sub-formula of }\psi\}$.
		
		\item  $V_0\subseteq V$ consists of all states 
		\begin{itemize}
			\item $(S,v)$ where $S \not\in \llbracket v \rrbracket ^{G_c}$.
			\item $(S,\neg v)$ where $S \in \llbracket v \rrbracket ^{G_c}$. 
			\item $(S, X)$ where $X\in \mathit{Var}$ is a relational variable.
			\item $(S, \eta X(\psi_1))$ for $\eta\in \{\mu,\nu\}$.
			\item $(S, \psi_1\vee \psi_2)$.
			\item $(S, \bigdiamond \psi_1)$.  
		\end{itemize}
		
		\item $V_1=V\setminus V_0$.
		
		\item A pair, $P\in V\times V$, belongs to $E$ if either of the following holds:
		\begin{itemize}
			
			
			
			\item $P=((S,\psi_1\wedge \psi_2),(S,\psi') )$ where $\psi'\in\{\psi_1,\psi_2\}$.
			
			\item $P=((S,\psi_1\vee \psi_2),(S,\psi') )$ where $\psi'\in\{\psi_1,\psi_2\}$.
			
			\item $P=((S,\eta X(\psi')),(S,\psi'))$ for $\eta\in \{\mu,\nu\}$.
			
			\item $P=((S,X),(S,\eta X(\psi')))$ where $\eta X(\psi')$ is the (unique) sub-formula of $\psi$ that binds $X$.
			
			\item $P=( (S,\bigdiamond \psi'), (T,\psi') )$ where $(S,T)\in E({G_c})$.
			
			\item $P=( (S,\square \psi'), (T,\psi') )$ where $(S,T)\in E({G_c})$.
			
		\end{itemize}
		\item For defining the priorities, take some even number $M\geq ad(\psi)$, where $ad(\psi)$ denotes the alternation depth of $\psi$.
		\begin{itemize}
			\item For $Q=(S,\psi'=\nu X(\xi))$, $\prio(Q)=M - 2 \lceil (ad(\psi')-1) / 2 \rceil$.
			
			\item For $Q=(S,\psi'=\mu X(\xi))$, $\prio(Q)=M-2\lfloor (ad(\psi')-1) / 2 \rfloor-1$.
			
			\item Otherwise, $\prio(Q)=M$.
			
		\end{itemize} 
		
	\end{itemize}
	
\end{defi}

\noindent By~\cite{EmersonJ91}, a strategy that wins for $\text{player}_0$ in the parity game $G_c\times \psi$, defined in Def.~\ref{def:modelCheckingGame}, corresponds to the value of the formula $\psi$ w.r.t. $G_c$:

\begin{cor}
	\label{cor:Gc-equiv-(Gc-times-psi)}
	For $S\in V(G_c)$, $S\in \llbracket \psi \rrbracket^{G_c}\Leftrightarrow$ $\text{player}_0$ has a winning strategy from $(S,\psi)$ in the parity game $G_c\times \psi$.
\end{cor} 

The next step in our reduction is to add weights to the edges of $G_c\times \psi$, namely to transform $G_c\times \psi$ into an energy parity game (as defined in Def.~\ref{def:energyParityGame}). We are interested in edges whose source states are of the form $(S,\square\psi')$ or $(S,\bigdiamond\psi')$ because their states' first components change according to transitions in the WGS $G^{w}$ (see Def.~\ref{def:modelCheckingGame} and Def.~\ref{def:concreteGstar}).
As we formally define below in Def.~\ref{def:addingWgithFunction}, such edges inherit their weights from the corresponding transitions in $G^w$.
\begin{defi}[The weight function $w$]\label{def:addingWgithFunction}
	We define the weight function, ${w}:{{E(G_c\times \psi)}\into {\mathbb{Z}}}$, as follows. For edges of the form $e=((s,u,c_0),\square\psi'),((t,c_1),\psi')$ or $e=((s,u,c_0),\bigdiamond \psi'),\linebreak((t,c_1),\psi')$, $w(e)=w^s(s,p(t))$. For every other edge, ${e} \in {E(G_c\times \psi)}$, $w(e)=0$.   
\end{defi}

\begin{lem}
	\label{lem:(Gc-times-psi)-equiv-(Gc-times-psi,w)}
	Let $S=(s,c_0)$ or $S=(s,u,c_0)$ be a state of $G_c$. Then, $\text{player}_0$ has a winning strategy from $(S,\psi')$ in the parity game $G_c\times\psi$ iff $\text{player}_0$ has a winning strategy from $(S,\psi')$ in the energy parity game $\langle G_c\times\psi,w \rangle$ w.r.t. $c$ for the initial credit $c_0$. 
\end{lem}

\begin{proof}
	
	The ``if" direction is trivial thus we focus on the ``only if" direction. We claim that if $g$ is a strategy that wins for $\text{player}_0$ from $(S,\psi')$ in $G_c\times \psi$, then it also wins for $\text{player}_0$ from $(S,\psi')$ in $\langle G_c\times\psi,w \rangle$ w.r.t. $c$ for the initial credit $c_0$. 
	
	Consider a play, $(S_0=S,\psi_0=\psi'),(S_1,\psi_1),(S_2,\psi_2),\dots$, consistent with $g$. Since $g$ wins from $(s,\psi')$, the play satisfies the parity objective and it is left to show that the energy objective is achieved as well. Each $S_i$ is of the form $S_i=(s_i,c_i)$ or $S_i=(s_i,u_i,c_i)$. To show that the energy level is always non-negative, we prove by induction on $i$ that ${\sf EL}_c(\langle G_c\times\psi,w \rangle,c_0,(S_0,\psi_0),\dots, (S_i,\psi_i))\geq c_i$. Note that the statement holds for $i=0$ and, for the induction step, assume that it holds for some $i\geq 0$. By the induction hypothesis, $c'' \geq c_{i}$ where $c'' := {\sf EL}_c(\langle G_c\times\psi,w \rangle,c_0,(S_0,\psi_0),\dots, (S_i,\psi_i))$. The interesting case is when $S_i=(s_i,u_i,c_i)$, $S_{i+1}=(s_{i+1},c_{i+1})$ and $\psi_i\in\{ \square \xi, \bigdiamond \xi \}$ for some $\xi$, since in all other cases, $w((S_i,\psi_i),(S_{i+1},\psi_{i+1}))=0$ and $c_{i+1}=c_i$.
	In this case, ${\sf EL}_c(\langle G_c\times\psi,w \rangle,c_0,(S_0,\psi_0),\dots, (S_{i+1},\psi_{i+1}))=\min\{c,c''+w^s(s_i,p(s_{i+1}))\}$.
	Since $(S_i,S_{i+1})$ is an edge of $G_c$, $\min\{c,c_i+w^s(s_i,p( s_{i+1} ))\} \geq c_{i+1}$.
	Therefore, as $c''\geq c_i$, $\min\{c,c''+w^s(s_i, p( s_{i+1} ) )\} \geq \min\{c,c_i+w^s(s_i,p(s_{i+1}) )\} \geq c_{i+1}$, as required. 
\end{proof}

Our next goal is to eliminate the energy component from the states of $G_c\times \psi$, so that the number of states will be independent of the choice of the upper bound.
Formally, for a state $S$ of $G_c$, let $\faktor{S}{c}$ denote its \emph{reduced version}, defined as follows:

\begin{align}
&{\faktor{S}{c}} :=
\begin{cases}
(s),~ & \text{if } S=(s,c_{0})\\
(s,u),~& \text{if } S=(s,u,c_{0})\\
\end{cases}&
\end{align}

Accordingly, we construct the \emph{reduced parity game} $\faktor{G_c\times \psi}{c}$ and its \emph{reduced weight function} $\faktor{w}{c}$. 

\begin{defi}[The reduced game $\langle \faktor{G_c\times \psi}{c},\faktor{w}{c} \rangle$]\label{def:reducedEnergyParityGame}
	Let $\langle \faktor{G_c\times \psi}{c},\faktor{w}{c} \rangle = \langle ((\faktor{V}{c}= \faktor{V_0}{c}\cup \faktor{V_1}{c},\faktor{E}{c}),\faktor{\prio}{c}), \faktor{w}{c} \rangle$ be the energy parity game defined by:
	\begin{itemize}
		\item $\faktor{V_0}{c}=\{(\faktor{S}{c},\psi') \mid (S,\psi')\in V_0(G_c\times \psi)\}$.
		
		\item $\faktor{V_1}{c}=\{(\faktor{S}{c},\psi') \mid (S,\psi')\in V_1(G_c\times \psi)\}$.
		
		\item $\faktor{E}{c}=\{((\faktor{S_1}{c},\psi_1),(\faktor{S_2}{c},\psi_2)) \mid ((S_1,\psi_1),(S_2,\psi_2))\in E(G_c\times\psi)\}$.
		
		\item $\faktor{\prio}{c}(\faktor{S}{c},\psi')=\prio(S,\psi')$. Note that $\faktor{\prio}{c}$ is well-defined as the priority $\prio(S,\psi')$ is solely determined by $\psi'$.
		
		\item If $e=((S_1,\psi_1),(S_2,\psi_2))\in E(G_c\times \psi)$, then $\faktor{w}{c}((\faktor{S_1}{c},\psi_1),(\faktor{S_2}{c},\psi_2))=w(e)$. Note that $\faktor{w}{c}$ is well-defined as the weight of the edge $e$ is independent of the energy components of $S_1$ and $S_2$.
	\end{itemize}
\end{defi}

\begin{lem}
	\label{lem:(Gc-times-psi,w)-equiv-(G_c-times-psi)/c}
	Let $c_{0} \leq c$. Then, $\text{player}_0$ has a wining strategy in $\langle G_c\times \psi,w \rangle$ from $((s_0,c_0),\psi)$ w.r.t. $c$ for the initial credit $c_0$ iff $\text{player}_0$ has a winning strategy in $\langle \faktor{G_c\times \psi}{c},\faktor{w}{c} \rangle$ from $((s_0),\psi)$ w.r.t. $c$ for the initial credit $c_0$.
\end{lem}

As a step towards proving Lem.~\ref{lem:(Gc-times-psi,w)-equiv-(G_c-times-psi)/c}, we first prove Lem.~\ref{lem:(Gc-times-psi,w)-equiv-(G_c-times-psi)/c:prefixStructure}.


\begin{lem}\label{lem:(Gc-times-psi,w)-equiv-(G_c-times-psi)/c:prefixStructure}
	For $i \geq 1$, let $({{(S_0,\psi_0)} = {((s_{0},c_0)}},\psi),\dots,(S_{i-1},\psi_{i-1}))$ be a path in $\langle G_c\times \psi,w \rangle$. Then,
	$S_{i-1}$ is of the form $S_{i-1}=(s_{i-1},u_{i-1},c_{i-1})$ iff $\psi_{i-1}=\bigdiamond\phi$.
\end{lem}
\begin{proof}[Proof of Lem.~\ref{lem:(Gc-times-psi,w)-equiv-(G_c-times-psi)/c:prefixStructure}]\hfill

	We prove the statement by induction on $i$. For $i=1$, it holds that
	$S_0=(s_{0},c_0)$, and, since $\psi_{0}=\psi \in \sysFormulas$, also $\psi_{0} \not =\bigdiamond\phi$. For the induction step, we show
	that the statement holds for $i > 1$ while assuming that it holds for every $1 \leq i_{0} < i$.
	If $\psi_{i-1}=\bigdiamond\phi$, as $\psi\in \sysFormulas$, by the construction, 
	$\psi_{i-2}=\square\bigdiamond\phi$. By the induction hypothesis, $S_{i-2}=(s_{i-2},c_{i-2})$ and hence, ${S_{i-1}}={(s_{i-1},u_{i-1},c_{i-1})}$. For the other direction, assume that ${S_{i-1}}={(s_{i-1},u_{i-1},c_{i-1})}$. Assume that the edge
	$\big((S_{i-2},\psi_{i-2}),(S_{i-1},\psi_{i-1})\big)$ conforms to either of the first four cases in Def.~\ref{def:modelCheckingGame}, and note that in all of these cases, $S_{i-2}=S_{i-1}$. Hence, $\psi_{i-2}$ has either of the following forms: $\phi_1\wedge \phi_2, \phi_1\vee \phi_2,\eta X(\phi),X$, where $\eta\in \{\mu,\nu\}$. 
	Then, by the induction hypothesis, $S_{i-2}=(s_{i-2},c_{i-2})=S_{i-1}$, in contradiction to the assumption. Now, assume that the edge $\big((S_{i-2},\psi_{i-2}),(S_{i-1},\psi_{i-1})\big)$ conforms to the fifth case in Def.~\ref{def:modelCheckingGame}. Then, $\psi_{i-2}=\bigdiamond \psi_{i-1}$ and, by the induction hypothesis, $S_{i-2}=(s_{i-2},u_{i-2},c_{i-2})$. By the construction, $S_{i-1}=(s_{i-1},c_{i-1})$, in contradiction to the assumption. Therefore, the edge $\big((S_{i-2},\psi_{i-2}),(S_{i-1},\psi_{i-1})\big)$ conforms to the last remaining case in Def.~\ref{def:modelCheckingGame} and $\psi_{i-2}=\square \psi_{i-1}$.
	Since $\psi\in \sysFormulas$, it follows that $\psi_{i-1}=\bigdiamond\phi$, as required.
\end{proof}

We now turn to prove Lem.~\ref{lem:(Gc-times-psi,w)-equiv-(G_c-times-psi)/c}.

\begin{proof}[Proof of Lem.~\ref{lem:(Gc-times-psi,w)-equiv-(G_c-times-psi)/c}]\hfill \begin{description}
		\item[``only if''] First, we prove the ``only if" statement. Assume that $g$ is a strategy that wins for $\text{player}_0$ in $\langle G_c\times \psi,w \rangle$ from $((s_0,c_0),\psi)$ w.r.t. $c$ for the initial credit $c_0$.
		We define a strategy for $\text{player}_0$ in $\langle \faktor{G_c\times \psi}{c},\faktor{w}{c} \rangle$, $\faktor{g}{c}$, as follows.
		For $i > 0$, take a sequence of states in $\faktor{G_c\times \psi}{c}$, $((T_0,\psi_0)=((s_{0}), \psi),\dots,(T_{i-1},\psi_{i-1}))$, that ends in a $\text{player}_0$ state,
		and each $T_j$ is either $T_j=(s_j)$ or $T_j=(s_j,u_j)$.
		If there are values $c_{1}, \ldots , c_{i-1} \in ([0,c])^{i-1}$ such that 
		the sequence, $((S_0,\psi_0)=((s_{0}, c_{0}), \psi),\dots,(S_{i-1},\psi_{i-1}))$ where for each $0<j\leq i-1$, $S_{j} = (T_{j}, c_{j})$,
		is a prefix of a play in $G_c\times \psi$, consistent with $g$, write 
		$g((S_0,\psi_0),\dots,(S_{i-1},\psi_{i-1}))=(S_{i},\psi_{i})$ and define $\faktor{g}{c}((T_0,\psi_0),\dots,(T_{i-1},\psi_{i-1}))=(\faktor{S_{i}}{c},\psi_{i})$.
		
		Let $((T_0,\psi_0)= ((s_{0}),\psi),\dots,(T_{i-1},\psi_{i-1}))$ be a prefix in $\langle \faktor{G_c\times \psi}{c},\faktor{w}{c} \rangle$, consistent with $\faktor{g}{c}$.
		To prove that $\faktor{g}{c}$ is indeed a well-defined, winning strategy, we argue by induction on $i\geq 1$:
		
		\begin{clm}
		 There are unique $c_{1}, \ldots , c_{i-1} \in ([0,c])^{i-1}$ such that $(((T_{0},c_{0}), \psi_{0})=((s_{0}, c_{0}), \psi),((T_{1}, c_{1}), \psi_{1}),\ldots,((T_{i-1}, c_{i-1}),\psi_{i-1}))$ is a prefix of a play in $\langle G_c\times \psi,w \rangle$, consistent with $g$.
		\end{clm} 
		
		Note that the above claim is immediate for $i=1$. We prove that it holds for ${i} > {1}$ while assuming that it holds for every $1 \leq j < i$.
		By the induction hypothesis, there are unique $c_{1}, \ldots , c_{i-2} \in ([0,c])^{i-2}$
		such that $(((T_{0},c_{0}), \psi_{0})=((s_{0}, c_{0}), \psi),((T_{1}, c_{1}), \psi_{1}),\ldots,\linebreak ((T_{i-2}, c_{i-2}),\psi_{i-2}))$
		is a prefix consistent with $g$.
		(1) Consider the case where $(T_{i-2},\psi_{i-2})$ is a $\text{player}_{1}$ state. Since $((T_{i-2}, c_{i-2}),\psi_{i-2}))$ is also a $\text{player}_{1}$ state,
		by Def.~\ref{def:modelCheckingGame}, we have that either $\psi_{i-2}=\square \psi_{i-1}$ or ${\psi_{i-2}} = {\xi_{1} \wedge \xi_{2}}$.
		Therefore, as $\psi_{i-2} \not = \bigdiamond \psi_{i-1}$, by Lem.~\ref{lem:(Gc-times-psi,w)-equiv-(G_c-times-psi)/c:prefixStructure}, $(T_{i-2}, c_{i-2}) = (s_{i-2}, c_{i-2})$. Hence,
		by Def.~\ref{def:reducedEnergyParityGame} and Def.~\ref{def:modelCheckingGame}, 
		if $\psi_{i-2}=\square \psi_{i-1}$, then $T_{i-1}= (s_{i-2}, u_{i-1})$, and otherwise, $T_{i-1}=(s_{i-2})=T_{i-2}$.
		Moreover, in either case, $(S_{i-1}, \psi_{i-1}) = ((T_{i-1},c_{i-2}), \psi_{i-1})$ is the only successor of $((T_{i-2},c_{i-2}), \psi_{i-2})$ in $G_c\times \psi$ such that $(\faktor{S_{i-1}}{c}, \psi_{i-1}) = (T_{i-1},\psi_{i-1})$.
		Consequently, $c_{1},\ldots, c_{i-2}, c_{i-2}$ are unique values such that $(((s_{0}, c_{0}), \psi),((T_{1}, c_{1}), \psi_{1}),\ldots,((T_{i-2}, c_{i-2}),\psi_{i-2}),((T_{i-1}, c_{i-2}),\psi_{i-1}))$
		is a prefix of a play consistent with $g$.
		(2) Consider the case where $(T_{i-2},\psi_{i-2})$ is a $\text{player}_{0}$ state. Since
		$(((s_{0}, c_{0}), \psi),((T_{1}, c_{1}), \psi_{1}),\ldots,((T_{i-2}, c_{i-2}),\psi_{i-2}))$ is a prefix of a play consistent with $g$,
		there exists ${((T_{i-1}, c_{i-1}), \psi_{i-1})} \in {V(G_c\times \psi)}$ such that 
		$g(((s_{0}, c_{0}), \psi),((T_{1}, c_{1}), \psi_{1}),\ldots,\linebreak ((T_{i-2}, c_{i-2}),\psi_{i-2})) = {((T_{i-1}, c_{i-1}), \psi_{i-1})}$.
		As $c_{i-1}$ is uniquely determined by $g$, it follows that $c_{1},\ldots, c_{i-2}, c_{i-1}$ are unique values
		such that $(((s_{0}, c_{0}), \psi),((T_{1}, c_{1}), \psi_{1}),\ldots,\linebreak((T_{i-2}, c_{i-2}),\psi_{i-2}),((T_{i-1}, c_{i-1}),\psi_{i-1}))$ is consistent with $g$, as required.
		
		Let ${\faktor{\sigma}{c}} = {{{((T_0,\psi_0)}= {((s_{0}),\psi)}},(T_{1},\psi_{1}),\dots)}$ be a play in $\langle \faktor{G_c\times \psi}{c},\faktor{w}{c} \rangle$ consistent with $\faktor{g}{c}$.
		A corollary of the above claim is that $\faktor{\sigma}{c}$ corresponds to a unique play, ${\sigma} = {({{((T_{0},c_{0}), \psi_{0})}={((s_{0}, c_{0}), \psi)}},((T_{1}, c_{1}), \psi_{1}),\ldots)}$,
		in $\langle G_c\times \psi,w \rangle$, consistent with $g$. Since $g$ wins for $\text{player}_{0}$, the play, $\sigma$, wins for $\text{player}_{0}$. 
		By Def.~\ref{def:reducedEnergyParityGame}, for each $j\geq 0$, $\faktor{\prio}{c}(T_{j},\psi_{j})=\prio((T_{j}, c_{j}), \psi_{j})$
		and $\faktor{w}{c}((T_{j},\psi_{j}),(T_{j+1},\psi_{j+1}))=w(((T_{j},c_{j}),\psi_{j}),((T_{j+1},c_{j+1}),\psi_{j+1}))$.
		Therefore, the priorities and weights traversed along $\faktor{\sigma}{c}$ are the same as those in $\sigma$, hence
		$\faktor{\sigma}{c}$ wins for $\text{player}_{0}$ in $\langle \faktor{G_c\times \psi}{c},\faktor{w}{c} \rangle$.
		This implies that $\faktor{g}{c}$ is a strategy that wins for $\text{player}_{0}$ in $\langle \faktor{G_c\times \psi}{c},\faktor{w}{c} \rangle$ from $((s_{0}), \psi)$, as required.

		\item[``if''] Second, we prove the ``if" statement. Assume that $\faktor{g}{c}$ is a strategy that wins for $\text{player}_0$ in $\langle \faktor{G_c\times \psi}{c},\faktor{w}{c} \rangle$ from $((s_{0}),\psi)$ w.r.t. $c$ for an initial credit $c_0$.
		We construct a strategy, $g$, that wins for $\text{player}_0$ in $\langle G_c\times \psi,w\rangle$ from $((s_0,c_0),\psi)$ w.r.t. $c$ for the initial credit $c_0$. We construct $g$ by $g=\bigcup_{i=1}^\infty g_i$ where each $g_i$ is a  partial function $g_i:(V(G_c\times\psi))^i\into V(G_c\times \psi)$.  We construct the functions $g_1,g_2,\dots$ by induction. Take $i> 0$ and assume that the functions $g_1,\dots,g_{i-1}$ have been defined. Take a sequence, $((S_0,\psi_0)= ((s_{0},c_0),\psi),\dots,(S_{i-1},\psi_{i-1}))$, consistent with $\bigcup_{k=1}^{i-1} g_k$, where each $S_j$ is either $S_j=(s_j,c_j)$ or $S_j=(s_j,u_j,c_j)$. 
		If $(S_{i-1},\psi_{i-1})$ is a $\text{player}_0$ state, $g_i((S_0,\psi_0),\dots,(S_{i-1},\psi_{i-1}))$ is defined as follows. Write $T_j=\faktor{S_j}{c}$ for each $j \leq i-1$ and  $\faktor{g}{c}((T_0,\psi_0),\dots,(T_{i-1},\psi_{i-1}))=(T_i,\psi_i)$. Take the largest $c_i$ such that $\big((S_{i-1},\psi_{i-1}),((T_i,c_i),\psi_i)\big)\in E(G_c\times\psi)$, and define $g_i((S_0,\psi_0),\dots,(S_{i-1},\psi_{i-1}))=(S_{i},\psi_i)$ where ${S_{i}} = {(T_i,c_i)}$.

		To prove that $g$ is indeed a winning strategy, we argue by induction on the length of the prefix, $((S_0,\psi_0)= ((s,c_0),\psi),\dots,(S_{i-1},\psi_{i-1}))$:
		
		\begin{enumerate}
			\item \label{lem:(Gc-times-psi,w)-equiv-(G_c-times-psi)/c:IT1} $((T_0,\psi_0),\dots,(T_{i-1},\psi_{i-1}))$ is consistent with $\faktor{g}{c}$.
			
			
			\item \label{lem:(Gc-times-psi,w)-equiv-(G_c-times-psi)/c:IT3}
			${\sf EL}_c(\langle G_c\times\psi,w \rangle,c_0,(S_0,\psi_0),\dots, (S_{i-1},\psi_{i-1})) =
			{\sf EL}_c(\langle \faktor{G_c\times\psi}{c},\faktor{w}{c} \rangle,c_0,(T_0,\psi_0),\dots,\linebreak (T_{i-1},\psi_{i-1})) = c_{i-1}$.
			
			\item \label{lem:(Gc-times-psi,w)-equiv-(G_c-times-psi)/c:IT4} If $(S_{i-1},\psi_{i-1})$ is a $\text{player}_0$ state, then ${g_i((S_0,\psi_0),\dots,(S_{i-1},\psi_{i-1}))}$ is well-defined. That is, there exists $T_i$ such that
			$\faktor{g}{c}((T_0,\psi_0),\dots,(T_{i-1},\psi_{i-1})) = {(T_i,\psi_i)}$, and there exists 
			$c_i \in [0,c]$ such that ${\big((S_{i-1},\psi_{i-1}),((T_i,c_i),\psi_i)\big)}\in {E(G_c\times\psi)}$.
			
		\end{enumerate}
		
		\noindent We leave it to the reader to verify that the above properties \ref{lem:(Gc-times-psi,w)-equiv-(G_c-times-psi)/c:IT1}~-~\ref{lem:(Gc-times-psi,w)-equiv-(G_c-times-psi)/c:IT4} hold for ${i} = {1}$.
		Properties \ref{lem:(Gc-times-psi,w)-equiv-(G_c-times-psi)/c:IT1}~-~\ref{lem:(Gc-times-psi,w)-equiv-(G_c-times-psi)/c:IT3} are immediate, and property~\ref{lem:(Gc-times-psi,w)-equiv-(G_c-times-psi)/c:IT4} holds by arguments similar to those in the general case where $i > 1$, yet simpler.
		We turn to prove that these properties hold for $i>1$ while assuming that they hold for every $1\leq i_0<i$.
		
		We first prove property~\ref{lem:(Gc-times-psi,w)-equiv-(G_c-times-psi)/c:IT1}.
		By the induction hypothesis, $((T_0,\psi_0),\dots,(T_{i-2},\psi_{i-2}))$ is consistent with $\faktor{g}{c}$.
		If $(S_{i-2},\psi_{i-2})$ is a $\text{player}_1$ state, then the same applies to
		$(T_{i-2},\psi_{i-2})$, and thus $((T_0,\psi_0),\dots,(T_{i-1},\psi_{i-1}))$
		is consistent with $\faktor{g}{c}$. Otherwise, we have the case where both $(S_{i-2},\psi_{i-2})$ and $(T_{i-2},\psi_{i-2})$ are $\text{player}_0$ states.
		By property~\ref{lem:(Gc-times-psi,w)-equiv-(G_c-times-psi)/c:IT4},
		$\faktor{g}{c}((T_0,\psi_0),\dots,(T_{i-2},\psi_{i-2}))=(T_{i-1},\psi_{i-1})$. Therefore, it follows from the induction hypothesis that
		$((T_0,\psi_0),\dots,(T_{i-1},\psi_{i-1}))$ is consistent with $\faktor{g}{c}$ in this case as well.

		We now turn to prove property~\ref{lem:(Gc-times-psi,w)-equiv-(G_c-times-psi)/c:IT3}. By the induction hypothesis, it holds that ${\sf EL}_c(\langle G_c\times\psi,\linebreak w \rangle,c_0,(S_0,\psi_0),\dots, (S_{i-2},\psi_{i-2}))= 
		{\sf EL}_c(\langle \faktor{G_c\times\psi}{c},\faktor{w}{c} \rangle, c_0,(T_0,\psi_0),\dots,(T_{i-2},\psi_{i-2})) = {c_{i-2}}$.
		%
		First, consider the case where ${\psi_{i-2}} \not = {\bigdiamond \psi_{i-1}}$.
		By Lem.~\ref{lem:(Gc-times-psi,w)-equiv-(G_c-times-psi)/c:prefixStructure},
		$S_{i-2} = (s_{i-2}, c_{i-2})$, and by Def.~\ref{def:addingWgithFunction} and Def.~\ref{def:reducedEnergyParityGame}, $w((S_{i-2},\psi_{i-2}),(S_{i-1},\psi_{i-1})) = \faktor{w}{c}((T_{i-2},\psi_{i-2}),(T_{i-1},\psi_{i-1})) = 0$.
		Also, by Def.~\ref{def:modelCheckingGame}, $c_{i-2} = c_{i-1}$. Thus, the property holds in this case.
		Second, consider the remaining case where $\psi_{i-2} = \bigdiamond \psi_{i-1}$. In this case, $(S_{i-2}, \psi_{i-2})$ is a $\text{player}_0$ state.
		By Lem.~\ref{lem:(Gc-times-psi,w)-equiv-(G_c-times-psi)/c:prefixStructure} and Def.~\ref{def:modelCheckingGame}, ${S_{i-2}}={(s_{i-2},u_{i-2},c_{i-2})}$ and ${S_{i-1}}={(s_{i-1},c_{i-1})}$. Therefore, by Def.~\ref{def:addingWgithFunction} and Def.~\ref{def:reducedEnergyParityGame}, we have that $w((S_{i-2},\psi_{i-2}),(S_{i-1},\psi_{i-1})) = \faktor{w}{c}((T_{i-2},\psi_{i-2}),\linebreak(T_{i-1},\psi_{i-1})) = w^s(s_{i-2},p(s_{i-1}))$.
		 This, together with the induction hypothesis, implies:
		 \begin{align*}
		 &{\sf EL}_c(\langle G_c\times\psi,w \rangle,c_0,(S_0,\psi_0),\dots, (S_{i-1},\psi_{i-1})) =\\
		 &{\sf EL}_c(\langle \faktor{G_c\times\psi}{c},\faktor{w}{c} \rangle, c_0,(T_0,\psi_0),\dots, (T_{i-1},\psi_{i-1})) =\\
		 &\min\{c,c_{i-2}+w^s(s_{i-2}, p(s_{i-1}))\}.&
		 \end{align*}
		But, since $\faktor{g}{c}$ wins for $\text{player}_0$ and, 
		by property~\ref{lem:(Gc-times-psi,w)-equiv-(G_c-times-psi)/c:IT1},
		$(T_0,\psi_0),\dots,(T_{i-1},\psi_{i-1})$ is consistent with $\faktor{g}{c}$,
		we have that ${\hat{c}} = {{\min\{c,c_{i-2}+w^s(s_{i-2},p(s_{i-1}))\}} \geq {0}}$.
		Also, by the construction, $\hat{c}$ is the maximal value such that 
		$\big(((s_{i-2},u_{i-2},c_{i-2}),\psi_{i-2}), ((s_{i-1},\hat{c}),\psi_{i-1})\big) \in {E(G_c\times \psi)}$. Thus, by the definition of $g_{i-1}$, we obtain that $c_{i-1} = \hat{c}$, which concludes the 
		proof for property~\ref{lem:(Gc-times-psi,w)-equiv-(G_c-times-psi)/c:IT3}.
		
		To prove property~\ref{lem:(Gc-times-psi,w)-equiv-(G_c-times-psi)/c:IT4}, first note that the existence of $(T_i,\psi_i)$ such that $\faktor{g}{c} ((T_0,\psi_0),\dots,\linebreak(T_{i-1},\psi_{i-1}))=(T_i,\psi_i)$ follows immediately from property~\ref{lem:(Gc-times-psi,w)-equiv-(G_c-times-psi)/c:IT1}, as $((T_0,\psi_0),\dots,(T_{i-1},\psi_{i-1}))$ is consistent with $\faktor{g}{c}$.
		Thus, it remains to prove that for some $c_i$, $\big((S_{i-1},\psi_{i-1}),((T_i,c_i),\psi_i)\big)\linebreak\in E(G_c\times\psi)$.
		First, consider the case where $\psi_{i-1} \not = \bigdiamond \psi_{i}$.
		By Lem.~\ref{lem:(Gc-times-psi,w)-equiv-(G_c-times-psi)/c:prefixStructure},
		$S_{i-1} = (s_{i-1}, c_{i-1})$. Thus, by Def.~\ref{def:modelCheckingGame} and Def.~\ref{def:reducedEnergyParityGame}, since $(S_{i-1},\psi_{i-1})$ is a $\text{player}_{0}$ state,
		it holds that $T_{i-1} = (s_{i-1}) = T_{i}$, and for $c_{i} := c_{i-1}$, $\big(((s_{i-1}, c_{i-1}),\psi_{i-1}),((s_{i-1},c_{i}),\psi_i)\big)\in E(G_c\times\psi)$.
		Second, consider the remaining case where $\psi_{i-1} = \bigdiamond \psi_{i}$.
		By Lem.~\ref{lem:(Gc-times-psi,w)-equiv-(G_c-times-psi)/c:prefixStructure},
		${S_{i-1}}={(s_{i-1},u_{i-1},c_{i-1})}$, hence
		${T_{i-1}}={(s_{i-1},u_{i-1})}$. Since $\big((T_{i-1},\bigdiamond \psi_{i}),(T_{i},\psi_{i})\big)\in E(\faktor{G_c\times \psi}{c})$,
		it follows from Def.~\ref{def:reducedEnergyParityGame} that $T_{i} = (s_{i})$.
		By Def.~\ref{def:modelCheckingGame} and Def.~\ref{def:concreteGstar}, it holds that
		for every $c' \in [0,c]$, $\big(((s_{i-1}, u_{i-1}, c_{i-1}),\bigdiamond \psi_{i}),((s_{i},c'),\psi_{i})\big)\in E(G_c\times\psi)$
		iff $\min\{c,c_{i-1}+w^{s}(s_{i-1},p(s_{i}))\} \geq c'$.
		Thus, it is sufficient to show that ${\min\{c,c_{i-1}+w^{s}(s_{i-1},p(s_{i}))\}} \geq {0}$.
		By property~\ref{lem:(Gc-times-psi,w)-equiv-(G_c-times-psi)/c:IT3}, we have that ${c_{i-1}} = {{\sf EL}_c(\langle \faktor{G_c\times\psi}{c},\faktor{w}{c} \rangle,c_0,(T_0,\psi_0),\dots, (T_{i-1},\bigdiamond \psi_{i}))}$.
		Moreover, by Def.~\ref{def:reducedEnergyParityGame}, $\faktor{w}{c}((T_{i-1},\bigdiamond \psi_{i}),(T_{i},\psi_{i})) = {\faktor{w}{c}(((s_{i-1},u_{i-1}),\bigdiamond \psi_{i}), ((s_{i}),\psi_{i}))} = {w^{s}(s_{i-1}, p(s_{i}))}$.
		Thus,  ${\min\{c,c_{i-1}+w^{s}(s_{i-1},p(s_{i}))\}} = {\sf EL}_c(\langle \faktor{G_c\times\psi}{c},\faktor{w}{c} \rangle, c_0,(T_0,\psi_0),\dots, (T_{i},\psi_{i}))$.
		Since $\faktor{g}{c}$ wins for $\text{player}_{0}$ and the prefix $((T_0,\psi_0),\dots,(T_{i},\psi_{i}))$ is consistent with $\faktor{g}{c}$,
		it follows that $\min\{c,c_{i-1}+w^{s}(s_{i-1},p(s_{i}))\} \geq 0$, 
		as required. \qedhere
	\end{description}
\end{proof}

\noindent We can now conclude that the graph $\faktor{G_c\times \psi}{c}$ simulates the WGS $G^w$. By Lem.~\ref{lem:G^w-equiv-Gc}, Cor.~\ref{cor:Gc-equiv-(Gc-times-psi)}, Lem.~\ref{lem:(Gc-times-psi)-equiv-(Gc-times-psi,w)} and Lem.~\ref{lem:(Gc-times-psi,w)-equiv-(G_c-times-psi)/c}, the following holds:

\begin{cor}
	\label{cor:G^w-equiv-(Gc-times-psi)/c}  
	Take a finite upper bound, $c\in \mathbb{N}$, and $c_0\leq c$. Then, the system has a winning strategy in the WGS $G^w$ from $s$ w.r.t. c for the initial credit $c_0$, iff $\text{player}_0$ has a winning strategy in the energy parity game $\langle \faktor{G_c\times \psi}{c},\faktor{w}{c} \rangle$ from $((s),\psi)$ w.r.t. $c$ for the initial credit $c_0$. \qed
\end{cor}
The crux of the reduction, which allows us to conclude the sufficiency of the upper bound specified in Thm.~\ref{Thm:a-sufficient-bound}, is that the energy parity game $\langle \faktor{G_c\times \psi}{c},\faktor{w}{c} \rangle$, defined in Def.~\ref{def:reducedEnergyParityGame},
is \emph{not} dependent on the choice of the bound $c$, provided that $c$ is sufficiently large. That is formally stated by the next lemma.

\begin{lem}
	\label{lemma:{Gc-times-psi}/c-independent-of-c}
	Let $K$ be the maximal transition weight in the WGS, $G^{w}$, in absolute value.
	Then, for $c,c'\geq K$, $c,c'\in \mathbb N$, ${\langle \faktor{G_c\times \psi}{c},\faktor{w}{c} \rangle} = {\langle \faktor{G_{c'}\times \psi}{c'},\faktor{w}{c'} \rangle}$. 
\end{lem}   
\begin{proof}
	Clearly, it holds that $V(\faktor{G_c\times \psi}{c})=V(\faktor{G_{c'}\times \psi}{c'})$. We show that $E(\faktor{G_c\times \psi}{c})=E(\faktor{G_{c'}\times \psi}{c'})$, and that each edge is assigned the same weight in both games. By symmetry, it suffices to show that for each $e\in E(\faktor{G_c\times \psi}{c})$, $e$ also belongs to $E(\faktor{G_{c'}\times \psi}{c'})$ and has the same weight. Take such an edge, $e$.
	By Def.~\ref{def:reducedEnergyParityGame} and Def.~\ref{def:modelCheckingGame}, $e$ has either of the following forms:
	\begin{enumerate}
		\item\label{lemma:{Gc-times-psi}/c-independent-of-c:case1} $e=(((s,u),\xi),((t),\xi'))$ where $\xi \in \{\bigdiamond \xi', \square \xi' \}$: In this case,
		for some $a,b\in[0,c]$, $(((s,u,a),\xi),((t,b),\xi'))\in E(G_c\times \psi)$. Thus, by Def.~\ref{def:concreteGstar}, $(s,p(t)) \models \rho^{s}$
		and $t|_{\mathcal{X}} = u$. 
		Also, by the definition of $K$, $\min\{c', K + w^s(s,p(t))\} \geq 0$. 
		Consequently, since $c'\geq K$, $(((s,u,K),\xi),((t,0),\xi'))\in E(G_{c'}\times \psi)$, hence, $e\in E(\faktor{G_{c'}\times \psi}{c'})$.
		Furthermore, the weight of $e$ is $w^s(s,p(t))$ in both games.
		\item\label{lemma:{Gc-times-psi}/c-independent-of-c:case2} $e=(((s),\xi),((s,u),\xi'))$ where $\xi \in \{\bigdiamond \xi', \square \xi' \}$: In this case,
		for some $a\in[0,c]$, $(((s,a),\xi),((s,u,a),\xi'))\in E(G_c\times \psi)$. Thus, by Def.~\ref{def:concreteGstar}, $(s,p(u))\models \rho^{e}$. Consequently,
		it holds that $(((s,0),\xi),((s,u,0),\xi'))\in E(G_{c'}\times \psi)$, hence $e\in E(\faktor{G_{c'}\times \psi}{c'})$.
		\item\label{lemma:{Gc-times-psi}/c-independent-of-c:case3} $e=((T,\xi),(T,\xi'))$ where $\xi \in \{\phi_1\wedge \phi_2, \phi_1\vee \phi_2,\mu X(\phi),\nu X(\phi),X\}$: In this case, for some $a\in [0,c]$, $(((T,a),\xi), ((T,a),\xi'))\in E(G_c\times \psi)$. Therefore,
		$(((T,0),\xi), ((T,0),\xi'))\in E(G_{c'}\times \psi)$, hence $e\in E(\faktor{G_{c'}\times \psi}{c'})$.
	\end{enumerate}
	
	\noindent Moreover, in cases~\ref{lemma:{Gc-times-psi}/c-independent-of-c:case2} and \ref{lemma:{Gc-times-psi}/c-independent-of-c:case3}, the weight of $e$ is $0$ in both games.
\end{proof}

We can now prove Thm.~\ref{Thm:a-sufficient-bound}.

\begin{proof}[Proof of Thm.~\ref{Thm:a-sufficient-bound}]

	Assume that $s$ wins in $G^w$ for the system player w.r.t. $+\infty$ for an initial credit $c_0$. By Lem.~\ref{lem:from-infinite-to-finite}, for some natural number $c\geq K$ and $c_0'=\min\{c_0,c\}$, $s$ wins for the system player w.r.t. $c$ for an initial credit $c_0'$. By Cor.~\ref{cor:G^w-equiv-(Gc-times-psi)/c}, $((s),\psi)$ wins for $\text{player}_0$ in $\langle \faktor{G_c\times \psi}{c},\faktor{w}{c} \rangle$ w.r.t. $c$ for an initial credit $c_0'$. Now, by Def.~\ref{def:reducedEnergyParityGame}, $\faktor{G_c\times \psi}{c}$ has at most $(N^2+N)m$ states and $d+1$ different priorities, and $K$ is the maximal weight of its edges, in absolute value. By Lem.~\ref{lem:lemma-6-revised}, $((s),\psi)$ wins for $\text{player}_0$ w.r.t. $b=(d+1)((N^2+N)m-1)K$ for an initial credit $c_0''=\min\{c_0',((N^2+N)m-1)K\}$. As $b\geq K$, by Lem.~\ref{lemma:{Gc-times-psi}/c-independent-of-c},
	${\langle \faktor{G_b\times \psi}{b},\faktor{w}{b} \rangle} = {\langle \faktor{G_{c}\times \psi}{c},\faktor{w}{c} \rangle}$.
	Therefore, by Cor.~\ref{cor:G^w-equiv-(Gc-times-psi)/c}, $s$ wins for the system player in $G^w$ w.r.t. $b$ for an initial credit $c_0''\leq \min\{c_0,((N^2+N)m-1)K\}$, as required. 
\end{proof}

%
\end{document}